\documentclass[prodmode,acmtocl]{acmsmall} 
\pdfoutput=1

\newcommand{\journalOrReport}[2]{#1} 

\usepackage[ruled]{algorithm2e}

\SetAlFnt{\small}
\SetAlCapFnt{\small}
\SetAlCapNameFnt{\small}
\SetAlCapHSkip{0pt}
\IncMargin{-\parindent}

\usepackage[english]{babel}
\usepackage{amsmath}
\usepackage{amssymb}
\usepackage{latexsym}
\usepackage{amstext}
\usepackage{stmaryrd}
\usepackage{url}
\usepackage{color}
\usepackage{enumerate}
\usepackage{wrapfig}
\usepackage{fancyvrb}
\usepackage{multirow}

\acmVolume{V}
\acmNumber{N}
\acmArticle{A}
\acmYear{YYYY}
\acmMonth{0}

\newcommand{\setsorts}{\mathcal{S}}
\newcommand{\setvars}{\mathcal{V}}
\newcommand{\setivars}{\mathcal{V}_{\mathtt{init}}}
\newcommand{\arrtype}{\Rightarrow}
\newcommand{\arrfunc}{\Longrightarrow}
\newcommand{\T}{\mathcal{I}}
\newcommand{\Terms}{\mathcal{T}\!\!\mathit{erms}}
\newcommand{\J}{\mathcal{J}}

\newcommand{\interpret}[1]{\llbracket #1 \rrbracket}
\newcommand{\Bool}{\mathbb{B}}

\newcommand{\Constructors}{\mathcal{C}\mathit{ons}}
\newcommand{\Defineds}{\mathcal{D}}

\newcommand{\arrz}{\to}
\newcommand{\arr}[1]{\arrz_{#1}}
\newcommand{\arrr}[1]{\arrz_{#1}^*}

\newcommand{\lrarr}[1]{\leftrightarrow_{#1}}
\newcommand{\lrarrr}[1]{\leftrightarrow_{#1}^*}
\newcommand{\arrzbase}{\arrz_{\mathtt{base}}}
\newcommand{\arrzrule}{\arrz_{\mathtt{rule}}}
\newcommand{\arrzcalc}{\arrz_{\mathtt{calc}}}

\newcommand{\constraint}[1]{[#1]}

\newcommand{\FV}{\mathit{Var}}
\newcommand{\domain}{\mathit{Dom}}
\newcommand{\Rules}{\mathcal{R}}
\newcommand{\Rulescalc}{\mathcal{R}_{\mathtt{calc}}}
\newcommand{\subst}[2]{#1#2}

\newcommand{\Sigmaterms}{\Sigma_{\mathit{terms}}}
\newcommand{\Sigmalogic}{\Sigma_{\mathit{theory}}}
\newcommand{\Sigmatheory}{\Sigma_{\mathit{theory}}}
\newcommand{\Sigmacore}{\Sigma_{\mathit{theory}}^{\mathit{core}}}
\newcommand{\Sigmaint}{\Sigma_{\mathit{theory}}^{\mathit{int}}}

\newcommand{\Values}{\mathcal{V}al}

\newcommand{\E}{\mathcal{E}}
\renewcommand{\H}{\mathcal{H}}

\newcommand{\flag}{\mathit{flag}}

\newcommand{\vdashrinoblank}{\vdash_{\mathtt{ri}}}
\newcommand{\vdashri}{\ \vdash_{\mathtt{ri}}\ }
\newcommand{\vdashristar}{\vdash_{\mathtt{ri}}^*}
\newcommand{\Expd}{\mathit{Expd}}
\newcommand{\Expdapp}[4]{\Expd(#1 \approx #2\ \constraint{#3},#4)}

\newcommand{\LVar}{\mathit{LVar}}

\newcommand{\equalgen}{\sim}
\newcommand{\coterm}[2]{#1\,\constraint{#2}}
\newcommand{\subpos}[2]{#1_{|#2}}
\newcommand{\subreplace}[3]{#1[#2]_{#3}}
\newcommand{\posroot}{\epsilon}
\newcommand{\seq}[1]{\overrightarrow{\!#1}}
\newcommand{\complflag}{\textsc{complete}}
\newcommand{\incomplflag}{\textsc{incomplete}}
\newcommand{\afun}{f}
\newcommand{\bfun}{g}

\newcommand{\avar}{x}
\newcommand{\bvar}{y}
\newcommand{\cvar}{z}
\newcommand{\dvar}{u}
\newcommand{\aterm}{s}
\newcommand{\bterm}{t}
\newcommand{\cterm}{u}
\newcommand{\dterm}{w}
\newcommand{\eterm}{q}
\newcommand{\asort}{\iota}
\newcommand{\bsort}{\kappa}

\newcommand{\symb}[1]{\mathsf{#1}}
\newcommand{\N}{\mathbb{N}}
\newcommand{\Z}{\mathbb{Z}}

\newcommand{\nul}{\symb{0}}
\newcommand{\one}{\symb{1}}
\newcommand{\two}{\symb{2}}

\newcommand{\summ}{\symb{sum}}
\newcommand{\fact}{\symb{fact}}

\newcommand{\strue}{\symb{true}}
\newcommand{\sfalse}{\symb{false}}

\newcommand{\uu}{\symb{u}}
\newcommand{\error}{\symb{error}}
\newcommand{\return}{\symb{return}}
\newcommand{\Int}{\symb{int}}
\newcommand{\IArr}{\symb{array(int)}}
\newcommand{\BOOL}{\symb{bool}}

\newcommand{\result}[1]{\symb{result}_{#1}}
\newcommand{\ret}[1]{\symb{return}_{#1}}

\newcommand{\intint}[2]{\{#1,\ldots,#2\}}

\newcommand{\aaa}{[}

\newcommand{\quant}[3]{#1#2 (#3)}
\newcommand{\bquant}[3]{\quant{#1}{#2}{#3}}
\newcommand{\parlr}[1]{~\longleftrightarrow\!\!\!\!\!\!\!\!\|~~_{#1}~}

\newcommand{\appshort}{Appendix~}
\newcommand{\secshort}{Section~}
\newcommand{\secsshort}{Sections~}
\newcommand{\exshort}{Example~}
\newcommand{\exsshort}{Examples~}
\newcommand{\thmshort}{Theorem~}
\newcommand{\lemshort}{Lemma~}
\newcommand{\defshort}{Definition~}

\renewcommand{\secshort}{\S~} 
\renewcommand{\secsshort}{\S\S~}
\renewcommand{\exshort}{Ex.~}
\renewcommand{\thmshort}{Thm.~}
\renewcommand{\lemshort}{Lemma~}
\renewcommand{\defshort}{Def.~}

\newcommand{\tool}[1]{\textsf{#1}}

\newcommand{\ctrl}{\tool{Ctrl}}
\newcommand{\infer}{\tool{Infer}}
\newcommand{\teetwo}{\tool{T2}}
\newcommand{\zeethree}{\tool{Z3}}

\newcommand{\llreve}{\tool{llr\^eve}}

\newcommand{\Prime}{${}^\prime$}

\newcommand{\myparagraph}[1]{\medskip\noindent\textit{#1.}}
\definecolor{comment}{rgb}{0.92, 0.92, 0.92}
\newcommand{\commentbox}[1]{\colorbox{comment}{\parbox{13cm}{
\emph{Comment:} #1}}}

\newcounter{tmp-section-Counter}
\newcounter{tmp-lemma-Counter}

\begin{document}

\markboth{C. Fuhs, C. Kop, N. Nishida}{Verifying Procedural Programs via Constrained Rewriting Induction}

\title{Verifying Procedural Programs via Constrained Rewriting Induction}
\author{%
CARSTEN FUHS
\affil{Birkbeck, University of London}
CYNTHIA KOP
\affil{University of Innsbruck and University of Copenhagen}
NAOKI NISHIDA
\affil{Nagoya University}
}

\begin{abstract}
This paper aims to develop a verification method for procedural
programs via a transformation into
Logically Constrained Term Rewriting Systems (LCTRSs).
To this end, we
extend
transformation methods based
on integer TRSs
to
handle arbitrary data types, global variables, function
calls and arrays, as well as encode safety checks.
Then we adapt existing rewriting induction methods to LCTRSs and
propose a simple yet effective method to generalize equations.
We show that we can automatically verify memory safety and prove
correctness of realistic functions.
Our approach proves equivalence between two implementations, so in
contrast to other works, we do not require an explicit specification
in a separate specification language.
\end{abstract}




\acmformat{Carsten Fuhs, Cynthia Kop, and Naoki Nishida, 2017.
Verifying Procedural Programs via Constrained Rewriting Induction.}

\begin{bottomstuff}
This work is supported by
Austrian Science Fund (FWF) international project I963,
Marie Sk\l{}odowska-Curie action ``HORIP'' (H2020-MSCA-IF-2014, 658162),
the Japan Society for the Promotion of
Science (JSPS), and Nagoya University's Graduate Program for
Real-World Data Circulation Leaders from \emph{MEXT}, Japan.

Authors' addresses:
 C.~Fuhs, Dept.\ of Comp.\ Sci.\ and Inf.\ Sys., Birkbeck, Univ.\ of London, UK;
 C.~Kop, Dept.\ of Comp.\ Sci., Univ.\ of Copenhagen, Denmark;
 N.~Nishida, Grad.\ School of Informatics, Nagoya Univ., Japan.
\end{bottomstuff}

\maketitle

\bigskip
\colorbox{comment}{\parbox{13cm}{This is an author copy of the paper \emph{Verifying Procedural
Programs via Constrained Rewriting Induction}, published at ACM TOCL in 2017.  The published paper
can be found at \url{https://dl.acm.org/doi/10.1145/3060143}.

\medskip
The contents are the same as in the published paper, except for one modification: in
Definition~\ref{def:expansion} (the definition of \textsc{Expansion}), we had erroneously failed to
include the condition that $\gamma(x) \in \Values \cup \setvars$ for all $x \in \FV(\varphi) \cup
\FV(\psi)$.
Without this condition, the definition is not necessarily well-defined, since for instance an
equation $\afun(x,x) \approx \symb{a}$ and a rule $\afun(y, \bfun(z)) \arrz \symb{b}\ [y > 0]$ would
cause an equation to be created with $\bfun(z) > 0$ in the constraint, which is not legal (assuming
that $\bfun$ is not a theory symbol).}}

\newpage

\section{Introduction}
\label{sec:intro}

Ensuring with certainty that a program always behaves correctly is a
hard problem. One approach to this is formal verification---proving
with mathematical rigor that all executions of the program will have
the expected outcome.
Several methods for this have been investigated (see
e.g., \cite{hut:rya:00}).  However, classically many of them
require expert knowledge to
manually prove relevant properties about the code.

Instead, we hope to raise the degree of automation, ideally
creating a fully automatic verification / refutation process
and tools to
raise
developer productivity. Indeed, over the last years
automatic provers for program verification have flourished, as
witnessed, e.g., by tool competitions like
SV-COMP
\citeA{sv-comp} and the Termination Competition
(\url{http://termination-portal.org/wiki/Termination_Competition}).
Program verification is also recognized in industry,
cf.\ e.g.\ 
Facebook's safety prover \infer\ 
\cite{cal:dis:dub:gab:hoo:luc:ohe:pap:pur:rod:15}
or
Microsoft's
temporal prover \teetwo\ \cite{bro:coo:ish:khl:pit:16}.
However, these tools generally use specific reasoning techniques for
imperative programs and benefit from the progress
in automated
theorem proving over the last decades
only to a limited extent.
This suggests
likely
avenues for improvement.

One such avenue
is \emph{inductive theorem proving}.
This method is well investigated in
functional
programming
\cite{bun:01}
and term rewriting,
the underlying core calculus of functional programming.
To check a functional program $f$
against a specification
by a reference implementation
$f_{\mathit{spec}}$,
it suffices that $f(\seq{x}) \approx f_{\mathit{spec}}(\seq{x})$ is an
\emph{inductive theorem}.
Thus,
no explicit specification language is needed:
giving a (possibly not optimized) reference implementation
$f_{\mathit{spec}}$ in the same programming language suffices.

To analyze imperative programs (in C, Java, etc.), recent works have
applied transformations into term rewrite systems
(e.g., \cite{ott:bro:ess:gie:10}).
In particular, \emph{constrained rewriting systems} are popular
as target language, since logical constraints
to model the control flow can be separated from terms
to model intermediate
states~\cite{fur:nis:sak:kus:sak:08,fal:kap:09,sak:nis:sak:sak:kus:09,nak:nis:kus:sak:sak:10,fal:kap:sin:11}.
Unifying existing approaches,
\citeN{kop:nis:13} have proposed
the framework of \emph{logically constrained term
rewriting systems (LCTRSs)}.

\myparagraph{Aims}
The aim of this paper is twofold.  First, we propose a new
transforma\-tion method from procedural programs into constrained term
rewriting.  This transformation makes it possible
to use the many methods available to term rewriting also to analyze
imperative programs.  Unlike previous methods, we do not limit
interest to integer functions.

Second, we develop a verification method for LCTRSs,
based on
rewriting induction~\cite{red:90}---a well-investigated method of
inductive theorem proving---to prove (total) equivalence of two
functions.  We also supply two generalization techniques, the main one
of which is specialized for transformed iterative functions.

The applications are many.
First, checking
equivalence between different
implementations comes to mind. This allows the user
to determine automatically
if a modification in the program has changed
its semantics
(see e.g.\ \cite{god:str:13,lah:haw:kaw:reb:12}).
Proposing equivalent replacements
may even be done automatically,
via algorithm recognition (see e.g.\ \cite{ali:bar:03}).

In compilation, automated equivalence checking
can validate correctness of compiler optimizations on a
per-instance basis
\cite{nec:00,pnu:sie:sin:98}
or
once-and-for-all for
a given
optimization template
\cite{kun:tat:ler:09,lop:mon:16}.
Equivalence checking is also used in proofs of secure information flow
\cite{ter:aik:05}
and can be used to
prove safety properties,
e.g., memory safety.

\myparagraph{Why LCTRSs}
Direct support of basic types like the integers, and of constraints to
restrict evaluation---features
absent in basic TRSs---is essential to handle realistic programs.
Unlike
earlier constrained
rewriting systems,
LCTRSs do not limit the underlying theory to (linear) integer
arithmetic: we might use (combinations of) arbitrary first-order
theories, including, e.g.,
$n$-dimensional integer arrays,
floating point numbers, and bitvectors.  This makes it possible
to natively handle sophisticated programs.

Despite the generality, we
get strong results on LCTRSs by
reducing analysis problems like termination and equivalence to a
sequence of satisfiability problems over the underlying
theories.
Automatic
tools---like our
tool \ctrl\ \cite{kop:nis:15}
for rewriting, termination, and inductive theorem
proving---can defer such queries to an external
\emph{SAT Modulo Theories (SMT)} solver \cite{nie:oli:tin:06}, as a
black box.
Future advances in the SMT world then directly transfer to analysis
of LCTRSs.

\myparagraph{Structure}
We first recall the LCTRS formalism
from~\cite{kop:nis:13} (\secshort\ref{sec:prelim}) and
show a way to translate procedural programs to LCTRSs
(\secshort\ref{sec:transformations}).
Then we
lift
rewriting induction
methods for
constrained rewriting
to LCTRSs (\secshort\ref{sec:ri})
and strengthen them  with two dedicated generalization techniques
(\secshort\ref{sec:lemma-gen}).
Finally we discuss automation and experimental results
(\secshort\ref{sec:experiments})
as well as related and future work
(\secsshort\ref{sec:related-work}--\ref{sec:future}) and conclude.

\myparagraph{Contributions over the conference version}
The present paper provides several additional
contributions over the conference version~\cite{kop:nis:14}:
(1)
We significantly extend our method to translate procedural
programs to LCTRSs.
(2)
We extend
our theory of constrained inductive theorem proving to
\emph{disproving} equivalence (following
\cite{sak:nis:sak:sak:kus:09,fal:kap:12})
and add several
inference rules.
(3)
We provide an additional
generalization
technique and a detailed proof strategy to automate
rewriting induction for translated procedural programs.
(4)
We have improved the implementation
and added an automatic
translation from C programs to LCTRSs.

\subsection{Motivating Example}
\label{subsec:motiv}

Aside from business applications,
automatic equivalence proving can be used as an aid in grading
student programming assignments.  Combining a test run of the
assignments on a set of sample inputs (which identifies many
incorrect programs, but leaves false positives) with an automatic
correctness check can save teachers a lot of time.

\begin{example}\label{exa:motivating}
Consider the following programming assignment.

\begin{center}
\begin{minipage}[t]{0.9\textwidth}
\it
Write a function {\tt sum} which, given an integer array and its
length as input, returns the sum of its elements.  Do not modify the
input array.
\end{minipage}
\end{center}

\noindent
We consider four different C implementations of this exercise:

\medskip\noindent
\begin{tabular}{c|c}
\begin{minipage}[h]{0.48\textwidth}
\begin{verbatim}
int sum1(int arr[],int n){
  int ret=0;
  for(int i=0;i<n;i++)
    ret+=arr[i];
  return ret;
}
\end{verbatim}
\end{minipage}
&
\begin{minipage}[h]{0.48\textwidth}
\begin{verbatim}
 int sum2(int arr[], int n) {
   int ret, i;
   for (i = 0; i < n; i++) {
     ret += arr[i];
   }
   return ret;
 }
\end{verbatim}
\vspace{-9pt}
\end{minipage}
\\
\\
\hline
\\
\begin{minipage}[h]{0.48\textwidth}
\begin{verbatim}
int sum3(int arr[], int len) {
  int i;
  for (i = 0; i < len-1; i++)
    arr[i+1] += arr[i];
  return arr[len-1];
}
\end{verbatim}
\vspace{-12pt}
\end{minipage}
&
\begin{minipage}[h]{0.48\textwidth}
\begin{verbatim}
 int sum4(int *arr, int k) {
   if (k <= 0) return 0;
   return arr[k-1] +
          sum4(arr, k-1);
 }
\end{verbatim}
\end{minipage}
\\
\end{tabular}

\bigskip
The first solution is correct.  The second is not, because
\texttt{ret} is not initialized---which may be missed
in standard tests depending on the compiler used.
The third solution is incorrect
because the array is modified against the instructions, and
moreover, gives a random result or segmentation fault
if \texttt{len} $=0$.
The fourth solution
is correct.

These implementations can be transformed into the following LCTRSs:
\[
\begin{array}{crcl}
\text{(1a)} &
\symb{sum1}(arr,n) & \arrz & \symb{u}(arr,n,\nul,\nul) \\
\text{(1b)} &
\symb{u}(arr,n,ret,i) & \arrz & \symb{error} \hfill \constraint{i < n \wedge (i < \nul \vee i \geq \symb{size}(arr))} \\
\text{(1c)} &
\symb{u}(arr,n,ret,i) & \arrz & \symb{u}(arr,n,ret+\symb{select}(arr,i),i+\one) \hfill \constraint{i \!<\! n \wedge \nul \leq i \!<\! \symb{size}(arr)} \\
\text{(1d)} &
\symb{u}(arr,n,ret,i) & \arrz & \symb{return}(arr,ret) \hfill \constraint{i \geq n} \\
 \\
\text{(2a)} &
\symb{sum2}(arr,n) & \arrz & \symb{u}(arr,n,ret,\nul) \\
& \multicolumn{3}{l}{\phantom{ABC} \symb{u}\ \text{rules as copied from
    above}}\\
\\
\text{(3a)} &
\symb{sum3}(arr,len) & \arrz & \symb{v}(arr,len,\nul) \\
\text{(3b)} &
\symb{v}(arr,len,i) & \arrz &
  \symb{error} \hfill \constraint{i < len - \one \wedge (i < \nul \vee i + 1 \geq \symb{size}(arr))} \\
\text{(3c)} &
\symb{v}(arr,len,i) & \arrz & \symb{v}(\symb{store}(arr,i+1,\symb{select}(arr,i+\one)+\symb{select}(arr,i)), len, i+\one)\ \  \\
  & & & \hfill \constraint{i < len -\one \wedge \nul \leq i \wedge i + 1 < \symb{size}(arr)} \\
\end{array}
\]
\[
\begin{array}{crcll}
\text{(3d)} &
\symb{v}(arr,len,i) & \arrz &
  \symb{return}(arr,\symb{select}(arr,len-\one)) \\
  & & & \hfill \constraint{i \geq len-\one \wedge \nul \leq len -\one < \symb{size}(arr)} \\
\text{(3e)} &
\symb{v}(arr,len,i) & \arrz & \symb{error}
  \hfill \constraint{i \geq len-\one \wedge (len - \one < \nul \vee len
  - \one \geq \symb{size}(arr))} \\
 \\
\text{(4a)} &
\symb{sum4}(arr,k) & \arrz & \symb{return}(arr,\nul) \hfill \constraint{k \leq \nul} \\
\text{(4b)} &
\symb{sum4}(arr,k) & \arrz & \symb{error} \hfill \constraint{k - \one \geq \symb{size}(arr)} \\
\text{(4c)} &
\symb{sum4}(arr,k) & \arrz & \symb{w}(\symb{select}(arr,k-\one), \symb{sum4}(arr,k-\one))\ \  \hfill \constraint{\nul \leq k - \one < \symb{size}(arr)} \\
\text{(4d)} &
\symb{w}(n,\symb{error}) & \arrz & \symb{error} \\
\text{(4e)} &
\symb{w}(n,\symb{return}(a,r)) & \arrz & \symb{return}(a, n + r) \\
\end{array}
\]
Note that arrays
carry an implicit size (their allocated memory)
which is queried
to model the runtime behavior of the
C program
and test for out-of-bound errors.
The fresh variable in the right-hand side of $(\text{2a})$
models
that the third parameter of $\symb{u}$
is assigned an \emph{arbitrary} integer.
The details of this transformation are discussed in
\secshort\ref{sec:transformations}.

Using inductive theorem proving, we can now prove that
\begin{itemize}
\item 
  $\forall arr \in \symb{array}(\Int).\ 
  \forall len \in \Int.\ 
  ~~
  \symb{sum1}(arr,len) \leftrightarrow^* \symb{sum4}(arr,len)\
  \text{if}\ \nul \leq len \leq \symb{size}(arr)$
\item
  $\exists arr \in \symb{array}(\Int).\ 
  \exists len \in \Int.\ 
  ~~
  \symb{sum3}(arr,len) \not\leftrightarrow^* \symb{sum4}(arr,len)\ 
  \text{with}\ \nul \leq len \leq \symb{size}(arr)$
\end{itemize}
So $\symb{sum1}$ and $\symb{sum4}$ return the same result on any input
such that the given length does not cause out-of-bound errors, but
$\symb{sum3}$ and $\symb{sum4}$ do not.  (It seems likely that the
disproof obtained from inductive theorem proving could be used to
extract counterexample inputs, but at present we have not studied a
systematic way of doing so.)

For $\symb{sum2}$, we \emph{do} have $\symb{sum2}(arr,len)
\leftrightarrow^* \symb{sum4}(arr,len)$, since
we \emph{can} always choose to instantiate $ret$ with $\nul$.
The system is not \emph{confluent}; we
can also prove that there exist
$a,n$ such that
$\symb{sum2}(a,n) \rightarrow^* \aterm \neq \bterm \leftarrow^*
\symb{sum4}(a,n)$ for terms $\aterm,\bterm$ in normal form.  As
explained in \secshort\ref{sec:implementation}, we use a
proof strategy which typically proves only the
``$\neq$'' statement.
\end{example}

\subsection{Practical Use}\label{subsec:practical}

The primary application that we see for our technique is the following:

\subsubsection{Comparing a function to a specification}
As in \exshort\ref{exa:motivating},
we can verify correctness of a C function
$\symb{f}$ against a reference implementation $\symb{g}$
by translating both functions to LCTRS rules
(\secshort\ref{sec:transformations}) and
proving that
$\symb{f}(\avar_1,\dots,\avar_n) \approx \symb{g}(\avar_1,\dots,
\avar_n)\ \constraint{\strue}$ is an inductive theorem.
If we only need equivalence under given preconditions
on the input variables---such as
  $\nul \leq \mathit{len}
  \leq \symb{size}(\mathit{arr})$ in
  \exshort\ref{exa:motivating}---we formulate this as a constraint
  $\varphi$ and
  analyze whether
$\symb{f}(\avar_1,\dots,\avar_n) \approx \symb{g}(\avar_1,\dots,
\avar_n)\ \constraint{\varphi}$
is an inductive theorem.

Note that
we do not require a separate specification language---although
if desirable, it is of course possible to specify the
reference implementation directly as an LCTRS.

\medskip
Further possible applications of our technique include:

\subsubsection{Code optimization (or other improvement)}
Sometimes the
``reference implementation'' $\symb{g}$ suggested above can simply be
an existing---and inefficient, or inelegant---version of a function.
Thus, inductive theorem proving can be used to prove that it is safe
to replace a function in a large real-life program by an optimized
alternative.

\subsubsection{Error checking}
As the transformation from C to LCTRSs
includes error checking (as seen for memory safety violations
in \exshort\ref{exa:motivating}), we can use inductive theorem proving
to verify the absence of such errors.  This is done by adding
error-checking rules,
e.g.,
\[
\begin{array}{rclcrcl}
\symb{errorfree}(\symb{return}(a,n)) & \arrz & \strue &
\phantom{ABCDE} &
\symb{errorfree}(\symb{error}) & \arrz & \sfalse
\end{array}
\]
and proving that $\symb{errorfree}(\symb{sum4}(a,n)) \approx \strue\ \constraint{\varphi}$ is an inductive theorem, where $\varphi$ is
the precondition on the input.
Aside from memory safety, this approach can be used to certify the
absence of
for instance divisions
by zero or
integer
overflow.  The key is in the transformation, where we can
choose which constructions result in an error.

\subsubsection{Classical correctness checks}
\label{subsubsec:classical}
Aside from comparisons to
an example implementation, we can also specify a correctness property
directly in SMT.  For instance, given an implementation of the
$\symb{strlen}$ function, its correctness could be verified by proving
that
\[
\symb{strlen}(x) \approx \symb{return}(n)\ \constraint{\nul \leq
n < \symb{size}(x) \wedge \symb{select}(x,n) = \nul \wedge
\bquant{\forall}{i \in \intint{\nul}{n-\one}}{\symb{select}(x,i) \neq \nul}}
\]
is an inductive theorem.  Alternatively, we can use extra rules
to test properties in SMT.

\begin{example}\label{ex:strcpytest}
To analyze correctness of an implementation of $\symb{strcpy}$,
we may use
\[
\begin{array}{rcl}
\symb{test}(\avar,n,\symb{error}) & \arrz & \sfalse \\
\symb{test}(\avar,n,\return(\bvar)) & \arrz & b\ 
  \constraint{b \Leftrightarrow
  \bquant{\forall}{i \in \intint{\nul}{n}}{\symb{select}(\avar,i) =
  \symb{select}(\bvar,i)}} \\
\end{array}
\]
and prove that the following equation is an inductive theorem:
\[
\begin{array}{c}
\symb{test}(\avar,n,\symb{strcpy}(\bvar,\avar)) \approx \strue \\
\aaa\nul \leq n < \symb{size}(\avar) \wedge n < \symb{size}(\bvar)
  \wedge \symb{select}(\avar,n) = \nul \wedge
\bquant{\forall}{i \in \intint{\nul}{n-\one}}{\symb{select}(\avar,i) \neq \nul}]
\end{array}
\]
Note that this more sophisticated test is \emph{needed} in this
case, since correctness of $\symb{strcpy}$ does not require that
$x = y$ if $\symb{strcpy}(x) \arrz^* \symb{return}(y)$ (the sizes
of $x$ and $y$ may differ).
\end{example}

\section{Preliminaries}
\label{sec:prelim}

In this section, we briefly recall \emph{Logically Constrained Term
Rewriting Systems} (usually abbreviated as \emph{LCTRSs}),
following the definitions in~\cite{kop:nis:13}.

\subsection{Logically Constrained Term Rewriting Systems}

\myparagraph{Many-sorted terms}
We introduce terms, typing, substitutions, contexts, and subterms
(with corresponding terminology) in the usual way for many-sorted
term rewriting.

\begin{definition}
We assume given a set $\setsorts$ of \emph{sorts} and an infinite set
$\setvars$ of \emph{variables}, each variable equipped with a sort.
A \emph{signature} $\Sigma$ is a set of \emph{function symbols}
$\afun$, disjoint from $\setvars$, each
equipped with a
\emph{sort declaration} $[\asort_1 \times \cdots \times \asort_n]
\arrtype \bsort$, with all $\asort_i$ and $\bsort$ sorts.
For readability, we often write $\bsort$ instead of
$[] \arrtype \bsort$.
The set $\Terms(\Sigma,\setvars)$ of \emph{terms} over $\Sigma$ and
$\setvars$ contains any expression $\aterm$ such that $\vdash \aterm
: \asort$ can be derived for some sort $\asort$, using:

\vspace{-4pt}
\noindent
\begin{minipage}[b]{0.28\linewidth}
\[
\frac{}{\vdash \avar : \asort}
\ (\avar : \asort \in \setvars)
\]
\end{minipage}
\begin{minipage}[b]{0.71\linewidth}
\[
\frac{\vdash \aterm_1 : \asort_1\ \ \ldots\ \ \vdash
\aterm_n : \asort_n
}{\vdash \afun(\aterm_1,\ldots,\aterm_n) : \bsort}
\ (\afun : [\asort_1 \times \cdots \times \asort_n] \arrtype \bsort 
\in \Sigma)
\]
\end{minipage}
\end{definition}

We fix $\Sigma$ and $\setvars$.
Note that for every term $\aterm$, there is a unique sort $\asort$
with $\vdash \aterm : \asort$.

\begin{definition}
Let $\vdash \aterm : \asort$.
We call $\asort$ the \emph{sort of} $\aterm$.
Let $\FV(\aterm)$ be the set of variables occurring in $\aterm$;
we say that $\aterm$ is \emph{ground} if $\FV(\aterm) = \emptyset$.
\end{definition}

\begin{definition}
A \emph{substitution} $\gamma$ is a sort-preserving total mapping
from $\setvars$ to $\Terms(\Sigma,\setvars)$.
The result $\subst{\aterm}{\gamma}$ of applying a substitution
$\gamma$ to a term $\aterm$ is $\aterm$ with all occurrences of
a variable $\avar$ replaced by $\gamma(\avar)$.
The \emph{domain} of $\gamma$, $\domain(\gamma)$, is the set
of variables $x$ with $\gamma(x) \neq x$.
The notation $[\avar_1:=\aterm_1,\ldots,\avar_k:=\aterm_k]$ denotes
a substitution $\gamma$ with $\gamma(\avar_i) = \aterm_i$ for
$1 \leq i \leq n$, and $\gamma(y) = y$ for $y \notin \{\avar_1,
\dots,\avar_n\}$.
For two substitutions $\gamma$ and $\delta$, their composition
$\gamma \circ \delta$ is given by $(\gamma \circ \delta)(\avar) =
\gamma(\delta(\avar)) = (\avar \delta) \gamma$ for all variables $\avar$.

Two terms $\aterm$ and $\bterm$ are \emph{unifiable}
if there exists a substitution $\gamma$ such that
$\subst{\aterm}{\gamma} = \subst{\bterm}{\gamma}$.
Then $\gamma$ is called a \emph{unifier}
for $\aterm$ and $\bterm$. If moreover for all unifiers $\gamma'$
for $\aterm$ and $\bterm$ there is a substitution
$\delta$ such that $\gamma' = \delta \circ \gamma$, we call $\gamma$
a \emph{most general unifier (mgu)}
for $s$ and $t$.
\end{definition}

\begin{definition}
Given a term $\aterm$, a \emph{position} in $\aterm$ is a sequence
$p$ of positive
integers such that $\subpos{\aterm}{p}$ is defined,
where $\subpos{\aterm}{\posroot} = \aterm$
and $\subpos{\afun(\aterm_1,\ldots,\aterm_n)}{i
\cdot p} = \subpos{(\aterm_i)}{p}$.  We call $\subpos{\aterm}{p}$
a \emph{subterm}~of
$\aterm$.
If $\vdash \subpos{\aterm}{p} : \asort$ and
$\vdash \bterm : \asort$, then $\subreplace{\aterm}{\bterm}{p}$
denotes $\aterm$ with the subterm at position $p$ replaced by
$\bterm$.
A \emph{context} $C$ is a term containing one or more typed
\emph{holes} $\Box_i : \asort_i$.
If $\aterm_1 : \asort_i,\ldots,\aterm_n : \asort_n$, we
define $C[\aterm_1,\ldots,\aterm_n]$ as $C$ with each $\Box_i$
replaced by $\aterm_i$.
\end{definition}

\myparagraph{Logical terms}
Specific to LCTRSs, we consider different kinds of symbols and terms.

\begin{definition}
We assume given:
\begin{itemize}
\item signatures $\Sigmaterms$ and $\Sigmalogic$ such that
  $\Sigma = \Sigmaterms \cup \Sigmalogic$;
\item a mapping $\T$ which assigns to each sort $\asort$ occurring in
  $\Sigmalogic$ a set $\T_\asort$;
\item a mapping $\J$ which assigns to each $\afun : [\asort_1 \times
  \cdots \times \asort_n] \arrtype \bsort \in \Sigmalogic$ a function
  in $\T_{\asort_1} \times \cdots \times \T_{\asort_n} \arrfunc
  \T_\bsort$;
\item for all sorts $\asort$ occurring in $\Sigmalogic$ a set
  $\Values_\asort \subseteq \Sigmalogic$ of \emph{values}: function
  symbols $a : [] \arrtype \asort$ such that $\J$ gives a bijective
  mapping from $\Values_\asort$ to $\T_\asort$.
\end{itemize}
We require that $\Sigmaterms \cap \Sigmalogic \subseteq \Values
= \bigcup_\asort \Values_\asort$.
The sorts occurring in $\Sigmalogic$ are called \emph{theory sorts},
and the symbols \emph{theory symbols}.
Symbols in $\Sigmalogic \setminus \Values$ are \emph{calculation
  symbols}.
A term in $\Terms(\Sigmalogic,\setvars)$ is called a \emph{logical
term}.
\end{definition}

\begin{definition}
For ground logical terms, let
$\interpret{\afun(\aterm_1,\ldots,\aterm_n)} :=
\J_\afun(\interpret{\aterm_1},\ldots,\interpret{\aterm_n})$.
For every ground logical term $\aterm$ there is a unique value $c$
such that $\interpret{\aterm} = \interpret{c}$; we say that $c$ is the
value of $\aterm$.
A \emph{constraint} is a logical term $\varphi$ of some sort $\BOOL$
with $\T_\BOOL = \Bool = \{ \top,\bot \}$, the set of
\emph{booleans}.
A constraint $\varphi$ is \emph{valid} if
$\interpret{\subst{\varphi}{\gamma}} = \top$ for \emph{all}
substitutions
$\gamma$ which map $\FV(\varphi)$ to values,
and \emph{satisfiable} if
$\interpret{\subst{\varphi}{\gamma}} = \top$ for \emph{some}
such substitutions.
A substitution $\gamma$ \emph{respects}
$\varphi$ if $\gamma(\avar)$ is a value
for all $\avar \in \FV(\varphi)$ and
$\interpret{\varphi\gamma} = \top$.
\end{definition}

Terms in $\Terms(\Sigmaterms,\emptyset)$ can be thought of
as the primary objects of rewriting:
a reduction typically begins and ends with such terms,
with elements of $\Sigmalogic \setminus \Values$ (also called
\emph{calculation symbols})
to perform calculations in the underlying theory.

\medskip
We typically choose a theory signature with $\Sigmalogic \supseteq
\Sigmacore$, where $\Sigmacore$ contains
$\strue,\sfalse : \BOOL,\wedge,\vee,\Rightarrow : [\BOOL \times
\BOOL] \arrtype \BOOL$,\ $\neg\!\!: [\BOOL] \arrtype \BOOL$,
and, for all theory sorts $\asort$, symbols $=_\asort, \neq_\asort :
[\asort \times \asort] \arrtype \BOOL$, and an evaluation function
$\J$ that interprets these symbols as expected.  We omit the
sort subscripts from $=$ and $\neq$ when
clear from context.

\begin{definition}
The standard integer signature $\Sigmaint$ is $\Sigmacore \cup \{ +,
-,*,\symb{exp},\symb{div},\linebreak
\symb{mod} : [\Int \times \Int] \arrtype
\Int; \leq, < : [\Int \times \Int] \arrtype \BOOL \} \cup \{\symb{n}
: \Int \mid n \in \Z\}$
with values $\strue,\ \sfalse$ and $\symb{n}$ for all $n \in \Z$.
Thus, we use $\symb{n}$ (in $\symb{sans}\text{-}\symb{serif}$ font)
as the function symbol for $n \in \Z$ (in $\mathit{math}$ font).
We define $\J$ in the natural way, except: since all
$\J_\afun$ must be total functions, we set
$\J_{\symb{div}}(n,0) = \J_{\symb{mod}}(n,0) = \J_{\symb{exp}}(n,k) = 0$ for all
$n$ and all $k < 0$.  Of course, when constructing LCTRSs, we normally
add explicit error checks to prevent such calls.
\end{definition}

\begin{example}
\label{exa:factsignature}
Let $\setsorts = \{ \Int,\BOOL \}$, and
$\Sigma = \Sigmaterms \cup \Sigmaint$, where
\[
\Sigmaterms = \{\ \fact : [\Int] \arrtype \Int\ \} \cup \{\ \symb{n} : 
  \Int \mid n \in \Z\ \}
\]
Then both $\Int$ and $\BOOL$ are theory sorts.
We also define set and function
interpretations, i.e., $\T_\Int = \Z$,\ $\T_\BOOL = \Bool$, and $\J$
is defined as above. With $=$ for $=_\Int$ and infix notation,
examples of logical terms are $\nul = \nul+-\one$ and
$\avar+\symb{3} \geq \bvar + -\symb{42}$.  Both are
constraints.  $\symb{5}+\symb{9}$ is also a (ground) logical term, but
not a constraint.
Expected starting terms are, e.g.,
$\fact(\symb{42})$ or $\fact(\fact(\symb{-4}))$:
ground terms fully built using symbols in $\Sigmaterms$.
\end{example}

\myparagraph{Rules and rewriting}
We adapt the standard notions of rewriting (see, e.g.,
\cite{baa:nip:98}) by including constraints
and adding rules to perform calculations.

\begin{definition}
A \emph{rule} is a triple $\ell \arrz r\ \constraint{\varphi}$
with $\ell$ and $r$ terms of the same sort and $\varphi$
a constraint.  Here, $\ell$ has the form $f(\ell_1,\dots,
\ell_n)$ and contains at least one symbol in $\Sigmaterms \setminus
  \Sigmatheory$ (so $\ell$ is not a logical term).
If $\varphi = \symb{true}$ with $\J(\symb{true}) = \top$,
we may write
$\ell \arrz r$.
We define $\LVar(\ell \arrz r\ \constraint{\varphi})$ as $\FV(\varphi)
\cup (\FV(r) \setminus \FV(\ell))$.
A substitution $\gamma$ \emph{respects} $\ell \arrz r\ 
\constraint{\varphi}$\linebreak if
$\gamma(\avar) \in\Values$
for all $\avar \in
\LVar(\ell \arrz r\ \constraint{\varphi})$, and 
$\interpret{\varphi\gamma} = \top$.
The rule is \emph{left-linear} if $\ell$ is linear, i.e.,
all variables occur at most once in $\ell$, and \emph{irregular}
if $\FV(\varphi) \setminus \FV(\ell) \neq \emptyset$.
\end{definition}

Note that it is allowed to have $\FV(r) \not \subseteq \FV(\ell)$, but
fresh variables in the right-hand side may only be instantiated
with \emph{values}.  This is done to model user input or random
choice.
Otherwise, variables outside the constraint may be instantiated by
any term; we do not impose strategies like innermost or call-by-value 
reduction.

\begin{definition}
We assume given a set of rules $\Rules$ and let
$\Rulescalc$ be the set
$\{ \afun(\avar_1,\ldots,\avar_n) \arrz \bvar\ 
\constraint{\bvar = \afun(\seq{\avar})} \mid \afun : [\asort_1
\times \cdots \times \asort_n] \arrtype \bsort \in \Sigmalogic
\setminus \Values \}$
(writing $\seq{\avar}$ for $\avar_1,\ldots,\avar_n$).
The \emph{rewrite relation} $\arrz_{\Rules}$ is a binary relation
on terms, defined by:
\[
\begin{array}{rcll}
C[\ell\gamma] & \arr{\Rules} & C[r\gamma]\ &
\text{if}\ 
\ell \arrz r\ \constraint{\varphi} \in \Rules \cup \Rulescalc\ 
\text{and}\ 
\gamma\ \text{respects}\ \ell \arrz r\ \constraint{\varphi} \\
\end{array}
\]
Here, $C$ is a context with exactly one hole.
We say that the reduction occurs at position $p$ if $C =
\subreplace{C}{\Box}{p}$.  Let $\aterm \leftrightarrow_\Rules
\bterm$ if $\aterm \arr{\Rules} \bterm$ or $\bterm \arr{\Rules}
\aterm$.
A reduction step with $\Rulescalc$ is called a \emph{calculation}.
A term is in \emph{normal form} if it cannot be reduced with
$\arr{\Rules}$.
We say that $t$ is a \emph{normal form of $s$} if $s \arrr{\Rules} t$ and
$t$ is a normal form.
The relation $\arr{\Rules}$ is \emph{confluent} if whenever $s \arrr{\Rules}
t$ and $s \arrr{\Rules} t'$, there exists also some $u$ with
$t \arrr{\Rules} u$ and $t' \arrr{\Rules} u$.
\end{definition}

We usually call the elements of $\Rulescalc$ rules---or
\emph{calculation rules}--even though their left-hand side is a
logical term.
Note that if $\arr{\Rules}$ is confluent, every term has at most
one normal form (intuitively, then $\Rules$ is deterministic
with respect to big-step semantics).

\begin{definition}
For $\afun(\ell_1,\ldots,\ell_n) \arrz r\ \constraint{\varphi}
\in \Rules$ we call $\afun$ a \emph{defined symbol}; non-defined
elements of $\Sigmaterms$ and all values are \emph{constructors}.
Let $\Defineds$ be the set of all defined symbols
and $\Constructors$ the set of constructors.
A term in $\Terms(\Constructors,\setvars)$ is a
\emph{constructor term}.
\end{definition}

Now we may define a \emph{logically constrained term rewriting
system} (LCTRS) as the abstract rewriting system $(\Terms(\Sigma,
\setvars),\arr{\Rules})$.  An LCTRS is usually given by
supplying $\Sigma$,\ $\Rules$, and an informal description of
$\T$ and $\J$ if these are not clear from context.

\begin{example}\label{exa:factlctrs}
To implement an LCTRS calculating the \emph{factorial} function, we
use the signature $\Sigma$ from \exshort\ref{exa:factsignature}
and the following rules:
\[
\Rules_{\fact}
 = \{\ 
\fact(x) \to \one\ \constraint{x \leq \nul}\ \ ,\ \ 
\fact(x) \to x * \fact(x-\one)\ \constraint{\neg (x \leq \nul)}
\ \}
\]
Using calculation steps, a term $\symb{3}-\symb{1}$ reduces to
$\symb{2}$ in one step (using the calculation rule
$\avar-\bvar\arrz\cvar\ \constraint{\cvar = \avar-\bvar}$), and
$\symb{3} * (\symb{2} * (\symb{1} * \symb{1}))$ reduces to
$\symb{6}$ in three steps.  Using also the rules in $\Rules_{\fact}$,
$\symb{fact}(\symb{3})$ reduces in ten steps to $\symb{6}$.
\end{example}

\begin{example}\label{exa:arraysumlctrs}
To implement an LCTRS calculating the sum of elements in an array,
let $\T_\BOOL = \Bool,\ \T_\Int = \Z,\ \T_\IArr = \Z^*$, so
$\IArr$ is mapped to finite-length integer sequences.
Let $\Sigmalogic = \Sigmaint \cup \{ \symb{size} : [\IArr] \arrtype
\Int,\ \symb{select} : [\IArr \times \Int] \arrtype \Int \}\ \cup\ 
\{ \symb{a} \mid a \in \Z^* \}$.
(We do \emph{not} encode arrays as lists: every ``array''---integer
sequence---$a$ corresponds to a unique symbol $\symb{a}$.)
The interpretation function $\J$ behaves on $\Sigmaint$ as usual,
maps the values $\symb{a}$ to the corresponding integer sequence, and
has:
\[
\begin{array}{rcll@{\!\:}rcll}
\J_{\symb{size}}(a) & = & k & \text{if}\ a = \langle n_0,\ldots,n_{k-1} \rangle\quad &
\J_{\symb{select}}(a,i) & = & n_i & \text{if}\ a = \langle n_0,\ldots, n_{k-1} \rangle\ \text{and}\ 0 \leq i < k \\
& & & & 
  & & 0 & \text{otherwise} \\
\end{array}
\]
In addition, let
  $\Sigmaterms = \{\ \summ, \symb{sum0} : [\IArr] \arrtype \Int
  \ \} \cup
    \{\ \symb{n} : \Int \mid n \in \Z\ \} \cup \{\ \symb{a} \mid a
    \in \Z^*\ \}$
  and let $\Rules$ consist of
  \[
  \begin{array}{rclrcll}
  \summ(\avar) & \arrz & \symb{sum0}(\avar,\symb{size}(\avar)-\one)
  \quad &
  \symb{sum0}(\avar,k) & \arrz & \symb{select}(\avar,k) +
    \symb{sum0}(\avar,k-\one) & \constraint{k \geq \nul} \\
  & & & 
  \symb{sum0}(\avar,k) & \arrz & \nul & \constraint{k < \nul} \\
  \end{array}
  \]
Note that this implementation differs from the ones in
\exshort\ref{exa:motivating}, because there we analyzed encodings
of imperative programs;
on C level there is no functionality for the programmer to explicitly
query the size of an array. Here, we avoided boundary checks.
\end{example}

Values are new in LCTRSs
compared to older styles of constrained rewriting.
These representatives of the underlying theory are always
\emph{constants}
(constructor symbols which do not take arguments),
even if they represent complex structures, as seen in
\exshort\ref{exa:arraysumlctrs}.
Note that variables in a rule's constraint \emph{must} be
instantiated by values; for instance in \exshort\ref{exa:factlctrs},
a term $\fact(\symb{1}+\symb{2})$ must be reduced by a
calculation first.  We also do
not match modulo theories, e.g., we do not equate
$\nul + (\avar + \bvar)$ with $\bvar + \avar$ for matching.

\myparagraph{Differences to~\cite{kop:nis:13}}
In the original definition of LCTRSs,
variables in $\setvars$ are unsorted, and a separate
\emph{variable environment} is used for typing.
Also, $\arr{\Rules}$ is there defined as the union of two
relations $\arrzrule$ and $\arrzcalc$ rather than including
$\Rulescalc$.
These changes give equivalent results, but the current definitions
cause less bookkeeping.
A larger difference is the restriction on rules:
in~\cite{kop:nis:13} left-hand sides must have a root symbol in
$\Sigmaterms \setminus \Sigmalogic$.  We follow~\citeN{kop:13}
and \citeN{kop:nis:14} in weakening this (only asking that
they are not logical terms).

\subsection{Quantification}
\label{subsec:quantify}

The definition of LCTRSs does not permit
constraints with quantifiers (constraints are terms, and
first-order rewriting does not allow quantifiers in terms).  In,
for instance, an LCTRS over integers and arrays, which has
$\symb{addtoend} : [\Int \times \IArr] \arrtype \IArr \in
\Sigmalogic$ and $\symb{extend} : [\IArr \times \Int] \arrtype
\IArr \in \Sigmaterms$, we cannot specify a rule like:
\[
\symb{extend}(\mathit{arr},\avar) \arrz \symb{addtoend}(\avar,
\mathit{arr})\ \constraint{\quant{\forall}{\bvar \in
  \intint{\nul}{\symb{size}(\mathit{arr})-\one}}{\avar \neq
     \symb{select}(\mathit{arr},\bvar)}}
\]

However, one of the key features of LCTRSs is that theory symbols,
including predicates, are not confined to a fixed list.  Therefore,
we \emph{can} add a new symbol to $\Sigmatheory$ (and $\J$).  For
the $\symb{extend}$ rule, we might
introduce a symbol $\symb{notin} : [\Int \times \IArr]
\arrtype \BOOL$ with $\J_{\symb{notin}}(u,\langle a_0,\ldots,a_{n-1}
\rangle) = \top$ iff for all $i$: $u \neq a_i$, and replace the
constraint by $\symb{notin}(\avar,\mathit{arr})$.  This generates
exactly the same reduction relation as the original rule.

Thus, we can permit quantifiers in the constraints of rules and also
on right-hand sides of rules, as an intuitive notation for fresh
predicates. However, an \emph{unbounded} quantification would
likely not be useful, as it would give an undecidable relation
$\arr{\Rules}$.

\medskip
\commentbox{One might argue that adding symbols like this
is problematic in practice:
no SMT solver will support new symbols like $\symb{notin}$.
However, for the \emph{technique} this makes no difference.  In
an implementation, we might allow quantifiers as syntactic sugar
(and pass the same sugar to the SMT solver), or add a layer on top of the
SMT solver which translates the new symbol(s), replacing for instance
\texttt{(notin \emph{u} \emph{a})}
by
\texttt{(forall ((x Int)) (distinct \emph{u} (select \emph{a} x)))}.}

\subsection{Rewriting Constrained Terms}
\label{subsec:constrained-term}

In LCTRSs, the objects of study are \emph{terms}, with
$\arr{\Rules}$ defining the relation between them.  However, for
analysis it is often useful to consider \emph{constrained terms}:

\begin{definition}
\label{def:equivalent-constrained-terms}
A constrained term is a pair $\coterm{\aterm}{\varphi}$ of a
term $\aterm$ and a constraint $\varphi$.
We say $\coterm{\aterm}{\varphi}$ and $\coterm{\bterm}{\psi}$
are \emph{equivalent}, notation $\coterm{\aterm}{\varphi} \equalgen
\coterm{\bterm}{\psi}$, if for all substitutions $\gamma$ which
respect $\varphi$ there is a substitution $\delta$ which respects
$\psi$ such that $\aterm \gamma = \bterm\delta$, and vice versa.
\end{definition}

Intuitively, a constrained term $\coterm{\aterm}{\varphi}$
represents all terms $\aterm\gamma$ where $\gamma$ respects
$\varphi$, and can be used to reason about such terms.
Equivalent constrained terms represent the same set of terms;
for example $\coterm{\symb{f}(\nul)}{\strue} \equalgen
\coterm{\symb{f}(\avar)}{\avar = \nul}$, and $\coterm{\symb{g}(
\avar,\bvar)}{\avar > \bvar} \equalgen \coterm{\symb{g}(\cvar
,\dvar)}{\dvar \leq \cvar - \one}$.
Note that $\coterm{\aterm}{\varphi} \equalgen \coterm{\aterm}{\psi}$
if and only if $\quant{\forall}{\seq{\avar}}{\bquant{\exists}{\seq{\bvar}}{\varphi}
\leftrightarrow \bquant{\exists}{\seq{\cvar}}{\psi}}$ holds, where
$\FV(\aterm) = \{\seq{\avar}\}$,\ $\FV(\varphi)\setminus\FV(\aterm)
= \{\seq{\bvar}\}$ and $\FV(\psi)\setminus\FV(\aterm) =
\{\seq{\cvar}\}$.

\begin{definition}
For a rule $\rho := \ell \arrz r\ \constraint{\psi}\in \Rules \cup
\Rulescalc$ and position $q$,
we let $\coterm{\aterm}{\varphi} \arr{\rho,q}
\coterm{\bterm}{\varphi}$ if
there exists a substitution $\gamma$ such that
$\subpos{\aterm}{q} = \ell\gamma$,
$\bterm = \subreplace{\aterm}{r\gamma}{q}$,
$\gamma(\avar)$ is a value or variable in $\FV(\varphi)$ for
all $\avar \in \LVar(\ell \arrz r\ \constraint{\psi})$, and
$\varphi \Rightarrow (\psi\gamma)$ is valid.
Let $\coterm{\aterm}{\varphi} \arrzbase \coterm{\bterm}{\varphi}$ if
$\coterm{\aterm}{\varphi} \arr{\rho,q} \coterm{\bterm}{\varphi}$ for
some $\rho,q$.
The relation $\arr{\Rules}$ on constrained terms is defined as
$\equalgen \cdot \arrzbase \cdot \equalgen$. We say that
$\coterm{\aterm}{\varphi} \arr{\Rules} \coterm{\bterm}{\psi}$
\emph{at position $q$ by rule $\rho$}
if $\coterm{\aterm}{\varphi}
\equalgen \cdot \arr{\rho,q} \cdot \equalgen \coterm{\bterm}{\psi}$.
\end{definition}

\begin{example}\label{ex:factconstrained}
In the
LCTRS from \exshort\ref{exa:factlctrs}, we have
$\coterm{\fact(\avar)}{\avar > \symb{3}} \arr{\Rules}
\coterm{\avar * \fact(\avar-\one)}{\avar > \symb{3}}$.
Now we can use a calculation
rule $\avar - \bvar \arrz \cvar\ \constraint{\cvar = \avar - \bvar}$,
with a non-empty $\equalgen$-step, as follows:
$\coterm{\avar * \fact(\avar-\one)}{\avar > \symb{3}} \equalgen
\coterm{\avar * \fact(\avar-\one)}{\avar > \symb{3} \wedge \cvar =
\avar - \one} \arrzbase \coterm{\avar * \fact(\cvar)}{\avar >
\symb{3} \wedge \cvar = \avar - \one}$.  The $\equalgen$-relation
holds because indeed $\quant{\forall}{\avar}{\avar >
\symb{3} \leftrightarrow \bquant{\exists}{\cvar}{\avar > \symb{3} \wedge \cvar = \avar -
\one}}$.
\end{example}

\begin{example}\label{ex:constrainedirregular}
The $\equalgen$-relation also allows us to reformulate the constraint
after a reduction.
For example, with the rule
$\symb{f}(\avar) \arrz \symb{g}(\bvar)\ \constraint{\bvar > \avar}$, we
have: $\coterm{\symb{f}(\avar)}{\avar > \symb{3}} \equalgen
\coterm{\symb{f}(\avar)}{\avar > \symb{3} \wedge \bvar > \avar}
\arrzbase \coterm{\symb{g}(\bvar)}{\avar > \symb{3} \wedge \bvar >
\avar} \equalgen \coterm{\symb{g}(\bvar)}{\bvar > \symb{4}}$.
We do \emph{not} have
that $\symb{f}(\avar)\ \constraint{\strue} \arr{\Rules} \symb{g}(\avar+\one)\ 
\constraint{\strue}$, as $\avar+\one$ cannot be instantiated
to a value.
\end{example}

\begin{example}\label{ex:constrainedredchoice}
A constrained term does not always need to be reduced in the most
general way.  With the rule $\symb{f}(\avar) \arrz \symb{g}(\bvar)\ 
\constraint{\bvar > \avar}$, we have $\symb{f}(\nul)\ \constraint{\strue}
\equalgen \symb{f}(\nul)\ \constraint{\bvar > \nul} \arrzbase
\symb{g}(\bvar)\ \constraint{\bvar > \nul}$, but we also have
$\symb{f}(\nul)\ \constraint{\strue} \equalgen
\symb{f}(\nul)\ \constraint{\one > \nul} \arrzbase
\symb{g}(\one)\ \constraint{\one > \nul} \equalgen
\symb{g}(\one)\ \constraint{\strue}$.
\end{example}

\noindent
As intended, constrained reductions give information about usual
reductions:

\begin{theorem}
\label{thm:constrainedterm}
If $\coterm{\aterm}{\varphi} \arr{\Rules} \coterm{\bterm}{\psi}$,
then for all substitutions $\gamma$ which respect $\varphi$ there
exists $\delta$ which respects $\psi$ such that $\aterm\gamma
\arr{\Rules} \bterm \delta$.
Both steps use the same rule and position.
\end{theorem}

\begin{proof}
We first observe (**):
\emph{
  If $\coterm{\cterm}{\xi} \arrzbase
  \coterm{\eterm}{\xi}$, then
  for any substitution $\gamma$
  which respects $\xi$ also $\cterm\gamma \arr{\Rules}
  \eterm\gamma$.
}
Proof: if $\coterm{\cterm}{\xi} \arrzbase
\coterm{\eterm}{\xi}$, then
there are $p, \ell \arrz r\ \constraint{c}$ and $\delta$
such that
$\subpos{\cterm}{p} = \ell\delta$, $\eterm =
\subreplace{\cterm}{r\delta}{p}$, $\delta(x) \in \FV(\xi)
\cup \Values$ for all $x \in \LVar(\ell \arrz r~[c])$ and $\xi
\Rightarrow (c\delta)$ is valid.
With
$\eta = \gamma \circ \delta$,
we have
$\subpos{(\cterm\gamma)}{p} =
\subpos{\cterm}{p}\gamma = \ell\delta\gamma = \ell\eta$ and
$\eterm\gamma = \subreplace{\cterm}{r\delta}{p}\gamma =
\subreplace{(\cterm\gamma)}{r\delta\gamma}{p} =
\subreplace{(\cterm\gamma)}{r\eta}{p}$.
We also have $\eta(x) = \delta(x)\gamma \in \Values$ for $x
\in \LVar(\ell \arrz r~[c])$ because $\gamma$ respects $\xi$ and,
since $\interpret{\xi\gamma} = \top$ and $\xi \Rightarrow (c\delta)$
is valid, also $\interpret{(c\delta)\gamma} = \interpret{c\eta} =
\top$.  So indeed $\cterm\gamma \arr{\Rules} \eterm\gamma$.

Now, suppose $\coterm{\aterm}{\varphi}
\arr{\Rules} \coterm{\bterm}{\varphi}$, so $\coterm{\aterm}{\varphi}
\equalgen \coterm{\aterm'}{\xi} \arrzbase
\coterm{\bterm'}{\xi} \equalgen \coterm{\bterm}{\psi}$, and let $\gamma$
respect $\varphi$.  By definition of $\equalgen$, there is some
substitution $\eta$ which respects $\xi$ such that $\aterm\gamma =
\aterm'\eta$.  By (**)
$\aterm'\eta \arr{\Rules} \bterm'\eta$.  Again by definition
of $\equalgen$, we find $\delta$ which respects $\psi$ such that
$\bterm'\eta = \bterm\delta$.
\qed
\end{proof}

\begin{theorem}
\label{thm:constrainedterm:reverse}
If $\coterm{\aterm}{\varphi} \arr{\Rules} \coterm{\bterm}{\psi}$,
then for all substitutions $\delta$ which respect $\psi$ there
exists $\gamma$ which respects $\varphi$ such that $\aterm\gamma
\arr{\Rules} \bterm \delta$.
Both steps use the same rule and position.
\end{theorem}

\begin{proof}
Parallel to the proof of \thmshort\ref{thm:constrainedterm}: if
$\coterm{\aterm}{\varphi} \equalgen \coterm{\aterm'}{\xi}
\arrzbase \coterm{\bterm'}{\xi} \equalgen \coterm{\bterm}{\psi}$,
then by definition of $\equalgen$ there are suitable $\eta,
\gamma$ such that $\bterm\delta = \bterm'\eta \leftarrow_\Rules
\aterm'\eta = \aterm\gamma$.
\qed
\end{proof}

\commentbox{The relation $\arr{\Rules}$
on constrained terms is \emph{not} stable: for instance, in the
system from \exshort\ref{ex:constrainedredchoice}, we can derive
$\coterm{\symb{f}(x)}{\strue} \arr{\Rules} \coterm{\symb{g}(x)}{\strue}$ even
though $\coterm{\symb{f}(\nul)}{\strue} \not\arr{\Rules}
\coterm{\symb{g}(\nul)}{\strue}$.  This is because the variables in a
constrained term $\aterm\: [\varphi]$ are fully changeable;
one can see variables in $\FV(\aterm)$ as universal and
the others as existential. This is not problematic, as
we do not instantiate constrained terms; to reason with
constrained reduction we only use Theorems~\ref{thm:constrainedterm}
and~\ref{thm:constrainedterm:reverse}.}

\section{Transforming Imperative Programs into LCTRSs}
\label{sec:transformations}

Equivalence-preserving
transformations of imperative programs into constrained rewriting
systems operating on integers have been investigated in
e.g.~\cite{fal:kap:09,fal:kap:sin:11,fur:nis:sak:kus:sak:08};
more generally, such translations from imperative to functional programs
have been investigated at least since \cite{mcc:60}.  Although
these papers use different definitions of constrained rewriting, the
proposed transformations can be adapted to produce LCTRSs that operate
on integers, i.e., use $\Sigmalogic$ as in \exshort\ref{exa:factlctrs}.
What is more, we can extend the ideas to also handle more advanced
programming structures, such as arrays and exceptions.

\smallskip
In this section, we will discuss a number of ideas towards a
translation from C to LCTRS.  A more detailed and formal treatment of
the limitation to integers and one-dimensional integer arrays is
available online along with an implementation, at:
\begin{center}
\url{http://www.trs.css.i.nagoya-u.ac.jp/c2lctrs/}
\end{center}
Given the extensiveness of the C specification, we will not
attempt to prove that the result of our transformation
corresponds to the origin.  Instead, we shall
rely on an appeal to intuition.  An advantage is that the same ideas
apply to other programming languages; we should be
able to use similar translations for, e.g., Python or Java.

\subsection{Transforming Simple Integer Functions}
\label{subsec:translate:basic}

The base form of the transformation---limited to integer functions
with no global variables or function calls---is very similar to the
transformations for integer TRSs
in \cite{fal:kap:09,fal:kap:sin:11,fur:nis:sak:kus:sak:08}.
Each function is transformed separately.  We introduce a function
symbol for every statement (including declarations), which
operates on the variables in scope.  The transition from one statement
to another is encoded as a rule, with assignments reflected by
argument updates in the right-hand side, and conditions by the
constraint.  Return statements are encoded by reducing to an
expression $\ret{f}(e)$, where $\ret{f} : [\Int] \arrtype \result{f}$
is a constructor.

\begin{example}
\label{exa:fact1 cnt 1}
Consider the following C function and its translation:

\begin{minipage}[h]{0.53\textwidth}
\begin{verbatim}
int fact(int x) {
  int z = 1;
  for (int i = 1; i <= x; i++) z *= i;
  return z;
}
\end{verbatim}
\end{minipage}
\begin{minipage}[h]{0.45\textwidth}
\vspace{-2pt}
\[
\begin{array}{rl@{~~}l}
\symb{fact}(x) & \to \symb{u}_1(x,\one) \\
\symb{u}_1(x, z) & \to \symb{u}_2(x, z, \one) \\
\symb{u}_2(x, z, i) & \to \symb{u}_3(x, z, i) & \constraint{i \leq x} \\
\symb{u}_2(x, z, i) & \to \symb{u}_5(x, z) &
  \constraint{\neg (i \leq x)} \\
\symb{u}_3(x, z, i) & \to \symb{u}_4(x, z * i, i) \\
\symb{u}_4(x, z, i) & \to \symb{u}_2(x, z, i + \one) \\
\symb{u}_5(x, z) & \to \ret{\symb{fact}}(z) \\
\end{array}
\]
\vspace{-2pt}
\end{minipage}
For $\Sigmalogic$ we assume the standard integer signature;
$\Sigmaterms$ contains $\symb{fact}$, all $\symb{u}_i$ and the
constructor $\ret{f}$, all of which have output sort $\result{f}$ and
argument sorts $\Int$.
\end{example}

A realistic translation of C code must also handle the
absence of a boolean data type,
operator precedence, and expressions with side effects (e.g.,
a loop condition \verb+--+$x$).  All this is
easily doable\footnote{This is discussed in the formal treatment at
\url{http://www.trs.css.i.nagoya-u.ac.jp/c2lctrs/formal.pdf}} (and
included in our implementation), but for the sake of brevity we
will not go into detail here.

\medskip
Finally, the generated system is optimized to make it more amenable to
analysis:\footnote{Variations of such preprocessing steps preserving the
properties of interest to simplify
the output of an automatic translation are
fairly standard in program analysis, see
e.g.\ \cite{alb:are:gen:pue:zan:08,alp:esc:luc:07,bey:cim:gri:ker:seb:09,fal:kap:sin:11,gie:asc:bro:emm:fro:fuh:hen:ott:plu:sch:str:swi:thi:16,spo:lu:mes:09}.}

  \begin{itemize}
  \item rules are combined where possible, e.g., replacing a pair of
    rules $\ell \arrz \uu(r_1,\dots,r_n)\ \constraint{\varphi}$ and
    $\uu(x_1,\dots,x_n) \arrz s\ \constraint{\strue}$ by $\ell \arrz
    s[x_1:=r_1,\dots,x_n:=r_n]$ if $\uu$ is not used elsewhere;
  \item unused arguments of function symbols are removed, such as the
    second (but not the first!) argument of $\uu$ in an LCTRS with
    rules $\uu(x,y,z) \arrz \uu(x-\one,y+\one,z*\symb{2})\ 
    \constraint{x > \nul}$ and $\uu(x,y,z) \arrz \return(z)\ 
    \constraint{\neg (x > \nul)}$;
  \item constraints are simplified, for instance replacing
    $\neg (x > \nul)$ by $x \leq \nul$ in the rules above.
  \end{itemize}

We will use these optimizations also for the extended transformations of
\secsshort\ref{subsec:transition:floats}--\ref{subsec:transition:array}.

\commentbox{When
\emph{time complexity}---defined as, e.g.,
the number of certain calculation
steps---is considered, the
argument removal step is dangerous, as it may remove calculations.
In such cases we would use a different simplification method.}

\begin{example}
\label{exa:fact1 cnt 2}
Optimizing the LCTRS from \exshort\ref{exa:fact1 cnt 1}, we obtain:
\[
\begin{array}{rl@{~~}rl@{~~}l}
\symb{fact}(x) & \to \symb{u}_2(x,\one,\one) \hspace{40pt} &
\symb{u}_2(x, z, i) & \to \symb{u}_2(x, z * i, i + \one) &
  \constraint{i \leq x} \\
&&\symb{u}_2(x, z, i) & \to \ret{\symb{fact}}(z) &
  \constraint{i > x}
\end{array}
\]
\end{example}

\myparagraph{Differences to older work} In contrast to existing
transformations to integer TRSs
(e.g.~\cite{fal:kap:09,fal:kap:sin:11,fur:nis:sak:kus:sak:08}),
we do not consider \emph{basic blocks}, but simply create rules for
every statement; this gives no substan\-tial difference after
optimization.
Additionally, $\ret{f}$ is new here: in
the work by Falke et al, the return statement is omitted, as
they focus on \emph{termination}, while
in~\cite{fur:nis:sak:kus:sak:08} the final term reduces directly to
the return-value, e.g.~$\symb{u}_4(x,z) \to z + x~\constraint{x \leq
\nul}$.

\subsection{Non-Integer Data Types}\label{subsec:transition:floats}

Integers are not special: as the definition of LCTRSs permits
arbitrary theories, we can handle any data type in C.  We might
for instance interpret \texttt{double} as either real numbers or
double-precision floating point numbers; this choice is left to the
user and may vary by application.
The only requirement is that a suitable theory signature---with
corresponding SMT solver if the system is to be analyzed
automatically---is available.
The translation is straightforward, with the only difficulty that
type casts must be made explicit, and we need to use separate symbols
such as $\symb{+.}$ for \texttt{double} addition.

\begin{example}
Consider the following C function and its translation.

\vspace{3pt}
\hspace{-12pt}
\begin{minipage}[h]{0.41\textwidth}
\begin{Verbatim}[commandchars=\\\[\]]
double halfsum(double thold) {
  double ret = 0.0;
  for (int d = 2; d < 100;
                  d *= 2) {
    ret += 1.0 / d;
    if (ret > thold) return ret;
} }
\end{Verbatim}
\vspace{2pt}
\end{minipage}
\begin{minipage}[h]{0.6\textwidth}
\vspace{-5pt}
\[
\begin{array}{rcll}
\symb{halfsum}(t) & \arrz & \symb{u}_2(t,\nul.\nul,\two) \\
\symb{u}_2(t,r,d) & \arrz & \symb{u}_4(t,r +\!\!.~\symb{1.0}/\symb{todouble}(d),d)\!\! &
  \constraint{d < \symb{100}} \\
\symb{u}_2(t,r,d) & \arrz & \ret{\symb{halfsum}}(rnd) &
  \constraint{d \geq \symb{100}} \\
\symb{u}_4(t,r,d) & \arrz & \ret{\symb{halfsum}}(r) & \constraint{r >\!\!.~t} \\
\symb{u}_4(t,r,d) & \arrz & \symb{u}_2(t,r,d * \two) &
  \constraint{r \leq\!\!.~t} \\
\\
\\
\end{array}
\]
\end{minipage}
This demonstrates both an explicit cast and one possible way to handle an
undefined return value (by a fresh variable, which may be
instantiated with a random value).
\end{example}

\subsection{Error Handling}
\label{subsec:transition:error}

The transformation of \secshort\ref{subsec:translate:basic}
does not fully reflect the original C program: as computers have
limited memory, integers are internally represented as
\emph{bitvectors}.
To address this, we could change the theory.  Rather than using $\mathbb{Z}$,
we let $\Values_\Int = \{\mathtt{MININT},\ldots,\mathtt{MAXINT}\}$ and
make $\J_+$, $\J_-$,
and $\J_*$ wrap around (e.g., $\J_-(\mathtt{MININT},\one) =
\mathtt{MAXINT}$).  The resulting LCTRS
has the same rules, but acts more closely to
the real program behavior.

However, integer overflow is often indicative of an \emph{error}.
Indeed, in C an overflow for the type \texttt{int} leads to undefined behavior
(which also surfaces in optimizing compilers such as
\texttt{gcc} or \texttt{clang}).
In order to model this (or other instances of undefined behavior in C,
such as a missing \texttt{return} statement), we will reduce to a
special $\symb{error}$ state.

Thus, for every rule $\symb{u}_i(\avar_1,\ldots,\avar_n) \arrz r\ 
\constraint{\varphi}$: if this rule represents a transition where an
error may occur under condition $\tau$, then we split it in two:
\[\begin{array}{cc}
  \symb{u}_i(\avar_1,\ldots,\avar_n) \arrz r\ \constraint{\varphi
  \wedge \neg \tau}
\ \ \ \ \ \ \ \ &
  \symb{u}_i(\avar_1,\ldots,\avar_n) \arrz \symb{error}_f\ 
  \constraint{\varphi \wedge \tau}
\end{array}\vspace{-1pt}\]
As usual, we simplify the resulting constraint
(writing, e.g., $x < \symb{0}$ instead of $\neg(x \geq \symb{0})$).

\begin{example}
\label{exa:fact1 cnt 3}
Continuing \exshort\ref{exa:fact1 cnt 2}, we generate the following
rewrite rules:
\[
\begin{array}{rl@{~~}l}
\symb{fact}(x) & \to \symb{u}_2(x,\one,\one) \\
\symb{u}_2(x, z, i) & \to \symb{u}_2(x, z * i, i + \one) &
  \constraint{i \leq x \land z * i \leq \mathtt{MAXINT} \land z * i
  \geq \mathtt{MININT} \land i + \one \leq \mathtt{MAXINT}} \\
\symb{u}_2(x, z, i) & \to \symb{error}_{\symb{fact}} &
  \constraint{i \leq x \land (z * i > \mathtt{MAXINT} \lor z * i <
  \mathtt{MININT} \lor i + \one > \mathtt{MAXINT})} \\
\symb{u}_2(x, z, i) & \to \ret{\symb{fact}}(z) & \constraint{i > x}
\end{array}
\]
\end{example}
Note that we could easily model assertions and \texttt{throw}
statements for exceptions in the same way.
Division by zero is handled in a similar way.

We can choose whether to add error transitions before or after the
simplification step.  The distinction is important: when simplifying,
calculations which do not contribute to the final result are thrown
away.  In the case of overflow errors, it may seem reasonable to
consider the post-simplification rules, as we did in
\exshort\ref{exa:fact1 cnt 3}.  In the
case of for instance division by zero, we should add the errors to the
pre-simplification rules.

\commentbox{When transforming a function into an LCTRS, we can
\emph{choose} what errors to model.  For instance, we could ignore
overflows (effectively assuming unbounded integers), but still
test for division by zero.  We could also let $\symb{error}_f$ be a
constructor which takes an argument, i.e., $\symb{error}_f :
[\symb{Errors}] \arrtype \symb{result}_f \in \Sigmaterms$, where
$\symb{Errors}$ is a sort with constructors $\symb{IntegerOverflow}$,
$\symb{DivisionByZero}$, and so on.}

\subsection{Global Variables}\label{subsec:global}

Thus far, we have considered very \emph{local} code: a function never
calls other functions or modifies global variables.  By altering the
$\symb{return}$ constructors, we easily change the latter:
we assume
that a function symbol is given all global variables that it uses as
input, and that it
returns those global variables it alters as output, along
with its return value.  This change also allows for non-redundant
\texttt{void} functions.

\begin{example}\label{ex:best}
Consider the following short program and its (simplified) translation:

\medskip
\begin{minipage}[h]{0.55\textwidth}
\begin{verbatim}
int best;
int up(int x) {
  if (x > best) { best = x; return 1; }
  return 0;
}
\end{verbatim}
\end{minipage}
\begin{minipage}[h]{0.42\textwidth}
\[
\begin{array}{rcll}
\symb{up}(b,x) & \arrz & \symb{return}_{\symb{up}}(x,\one) & \constraint{x > b} \\
\symb{up}(b,x) & \arrz & \symb{return}_{\symb{up}}(b,\nul) & \constraint{x \leq b} \\
\\
\\
\end{array}
\]
\end{minipage}
\end{example}

\subsection{Function Calls}
\label{subsec:calls}

Next, let us consider \emph{function calls}.  A difficulty is that
they may occur in an expression, e.g.,
$\fact(\symb{3}) + \symb{5}$, which is not well sorted in the
corresponding LCTRS: $\fact(\symb{3})$
has sort $\symb{result}_\fact$, not $\Int$.  To avoid this issue,
and to propagate errors, we split off function calls occurring
inside expressions other than $\mathit{var} =
\mathit{func}(\mathit{arg}_1,\ldots,\mathit{arg}_n)$ and store their
return value into a temporary variable.  For example:

\begin{minipage}[h]{0.45\textwidth}
\begin{verbatim}
int ncr(int x, int y) {
  int a = fact(x);
  int b = fact(y) * fact(x - y);
  return a / b;
}
\end{verbatim}
\end{minipage}
\begin{minipage}[h]{0.165\textwidth}
\[ \mbox{\huge $\Longrightarrow$ } \]
\end{minipage}
\begin{minipage}[h]{0.4\textwidth}
\begin{verbatim}
int ncr(int x, int y) {
  int a = fact(x);
  int tmp1 = fact(y);
  int tmp2 = fact(x - y);
  int b = tmp1 * tmp2;
  return a / b;
}
\end{verbatim}
\end{minipage}

This change may cause declarations at places in the function where a
C compiler would not accept them, but for the translation, this is no
issue.  We translate the resulting function
by executing function
calls in a separate parameter and using a separate step to examine the
outcome of a function call and assign it to the relevant variable(s).

\begin{example}
The \texttt{ncr} program above is transformed to the following
optimized LCTRS (where we test for division by zero but not
integer overflow for simplicity):
\[
\begin{array}{cc}
\begin{array}{rclrcl}
\symb{ncr}(x,y) & \arrz & \symb{u}_2(x,y,\fact(x)) &
\symb{u}_2(x,y,\symb{error}_\fact) & \arrz & \symb{error}_{\symb{ncr}} \\
\symb{u}_2(x,y,\symb{return}_\fact(k)) & \arrz & \symb{u}_3(x,y,k,
  \fact(y)) &
\symb{u}_3(x,y,a,\symb{error}_\fact) & \arrz & \symb{error}_{\symb{ncr}} \\
\symb{u}_3(x,y,a,\symb{return}_\fact(k)) & \arrz & \symb{u}_4(x,y,a,k,
  \fact(x-y)) &
\ \ \ \symb{u}_4(x,y,a,t_1,\symb{error}_\fact) & \arrz &
  \symb{error}_{\symb{ncr}} \\
\end{array} \\
\begin{array}{rcll}
\symb{u}_4(x,y,a,t_1,\symb{return}_\fact(k)) & \arrz &
  \symb{error}_{\symb{ncr}} & \constraint{t_1 * k = \nul} \\
\symb{u}_4(x,y,a,t_1,\symb{return}_\fact(k)) & \arrz &
  \symb{return}_{\symb{ncr}}(a\ \symb{div}\ (t_1 * k)) &
  \constraint{t_1 * k \neq \nul} \\
\end{array}
\end{array}
\]
\end{example}

\subsection{Statically Allocated Arrays}
\label{subsec:transition:array}

Finally, let us consider arrays.
After we have seen \exshort\ref{exa:motivating} and the way side
effects were handled in \secshort\ref{subsec:global}, this is
largely as expected.  For now, we will not consider aliasing.

To start, we must fix a theory signature and corresponding
interpretations.
For a given theory sort $\asort$ which admits at least one value,
say $\nul_\asort$, let $\symb{array}(\asort)$ be a new sort and
$\T_{\symb{array}(\asort)} = \T_\asort^*$---so each value
corresponds to a finite sequence.
We introduce the following theory
symbols (in addition to $\Sigmaint$ and other desired theories):
\begin{itemize}
\item $\symb{size}_\asort : [\symb{array}(\asort)] \arrtype \Int$:
  we define $\J_{\symb{size}_\asort}(a)$ as the length of the sequence
  $a$.
\item $\symb{select}_\asort : [\symb{array}(\asort) \times \Int]
  \arrtype \asort$: if $a = \langle a_0,\ldots,a_{n-1}\rangle$, we
  define $\J_{\symb{select}_\asort}(a,k) = a_k$ if $0 \leq k < n$ and
  $\J_{\symb{select}_\asort}(a,k) = \nul_\asort$ otherwise.
\item $\symb{store}_\asort : [\symb{array}(\asort) \times \Int \times
  \asort] \arrtype \symb{array}(\asort)$: if $a = \langle a_0,\ldots,
  a_{n-1} \rangle$, we define $\J_{\symb{store}_\asort}(a,k,v) =
  \langle a_0,\ldots,a_{k-1},v,a_{k+1},\ldots,a_{n-1}\rangle$ if $0
  \leq k < n$ and $\J_{\symb{store}_\asort}(a,k,v) = a$ otherwise.
\end{itemize}
We will usually omit the subscript $\asort$ when the sort is clear
from context.

Our arrays are different from
SMT-LIB (cf.\ \url{http://www.smt-lib.org/}), where arrays are
functions from one (possibly infinite) domain to another.
For program analysis, finite-length sequences seem practical instead.
SMT problems on our arrays can be translated to SMT-LIB
  format using an additional integer variable
  $a_{\symb{size}}$ for the size of an array $a$ and universal
  quantification to set entries outside the array to a fixed value.

We encode \emph{lookups} $a[i]$ as $\symb{select}(a,i)$;
for \emph{assignments} $a[i] = e$, we
replace $a$ by $\symb{store}(a,\linebreak i,e)$.
To ensure correctness here, we add boundary checks to the constraint
and reduce to $\symb{error}_\afun$ if such a check is not satisfied.
After an assignment, the updated variable
is included in the return value since the underlying memory of
the array was altered.

\begin{example}\label{ex:strcpy}
Consider the following C implementation of the $\symb{strcpy}$
function, which copies the contents of \texttt{original} into the
array \texttt{goal}, until a $0$ is reached.

\vspace{-4pt}
\begin{verbatim}
void strcpy(char goal[], char original[]) {
  int i = 0;
  for (; original[i] != 0; i++) goal[i] = original[i];
  goal[i] = 0;
}
\end{verbatim}
\vspace{-4pt}

For simplicity, we think of strings as integer arrays
(although alternative choices for $\T_{\symb{char}}$ make little
difference).
The function never updates \texttt{original}, but may update
\texttt{goal}, so the return value must include the latter.
We obtain the following LCTRS:
\[
\begin{array}{rcll}
\symb{strcpy}(gl,org) & \arrz & \symb{v}(gl,org,\nul) \\
\symb{v}(gl,org,i) & \arrz & \error_{\symb{strcpy}} &
  \constraint{i < \nul \vee i \geq \symb{size}(org)} \\
\symb{v}(gl,org,i) & \arrz & \symb{w}(gl,org,i) &
  \constraint{\nul \leq i < \symb{size}(org) \wedge
              \symb{select}(org,i) = \nul} \\
\symb{v}(gl,org,i) & \arrz & \error_{\symb{strcpy}} &
  \constraint{\nul \leq i < \symb{size}(org) \wedge
              \symb{select}(org,i) \neq \nul \wedge i \geq \symb{size}(gl)} \\
\symb{v}(gl,org,i) & \arrz & \multicolumn{2}{l}{\symb{v}(\symb{store}(gl, i,
  \symb{select}(org,i)),org,i+\one)} \\
& & &
 \constraint{\nul \leq i < \symb{size}(org) \wedge
 \symb{select}(org,i) \neq \nul \wedge
 i < \symb{size}(gl)} \\
\symb{w}(gl,org,i) & \arrz & \error_{\symb{strcpy}} &
  \constraint{i < \nul \vee i \geq \symb{size}(gl)} \\
\symb{w}(gl,org,i) & \arrz &
  \multicolumn{2}{l}{
    \return_{\symb{strcpy}}(\symb{store}(gl,i,\nul))\ 
    \constraint{\nul \leq i < \symb{size}(gl)}
  }\\
\end{array}
\]
Here, the notation $\nul \leq i < \symb{size}(org)$ is shorthand
for $\nul \leq i \land i < \symb{size}(org)$.  Note that this LCTRS
could be further simplified by combining the third rule
with the last two rules.
\end{example}

\commentbox{It should now be clear how the systems from
\secshort\ref{subsec:motiv} have been translated from C code to
LCTRSs.  The only deviation is that there we have included the array
$\mathit{arr}$ in the return value of $\symb{sum1}$, $\symb{sum2}$, and
$\symb{sum4}$, which is not necessary as it is not modified in these
cases.  This was done to allow for a direct comparison with
$\symb{sum3}$,
where the array \emph{is} modified.  In addition, the $\symb{return}$
and $\symb{error}$ symbols in these examples are not indexed, for
the same reason.}

\subsection{Dynamically Allocated Arrays and Aliasing}\label{subsec:transition:pointer}

The transformation in \secshort\ref{subsec:transition:array}
allows us to abstract from the underlying memory model when encoding
arrays.  This makes analysis easier, but does not allow for aliasing
or pointer arithmetic beyond accessing an array element.
As a result, properties
we prove about $\symb{strcpy}$ from \exshort\ref{ex:strcpy} might
fail to hold for a call like $\symb{strcpy}(a,a)$.

As we seek to handle only \emph{part} of the language, this does
not need to be an issue; in practice,
a fair number of programs are written
without explicit pointer use and with easily removable
aliasing only. For example, we might replace
$\symb{strcpy}(a,a)$ by $\symb{strcpy}'(a)$, and create new rules
for $\symb{strcpy}'$ by collapsing the variables in the rules for
$\symb{strcpy}$.
To handle programs with more sophisticated pointer use,
including dynamically allocated arrays, we can
encode the memory as a list of arrays and pass this along as a
variable.  This is somewhat beyond the scope of this paper, but is
explored in \appshort\ref{app:pointers}.

\subsection{Remarks}
\label{subsec:trafo_remarks}

The treatment in this section is both informal and incomplete:
we have discussed only
a fraction of the C language---albeit an important fraction
for verification.  We believe
that these ideas easily extend
further, with for instance the \texttt{switch} statement,
user-defined data structures,
or standard library functions, as
well as compiler-specific choices.
Important to note is that the translation gives several
\emph{choices}.  Most pertinently, we saw the choices what
sort interpretations to use (e.g.,
whether \texttt{int} should be mapped to
the set of integers or bitvectors) and what errors to consider.

In this paper, and in line with our automatic translation at
\url{http://www.trs.css.i.nagoya-u.ac.jp/c2lctrs/}, we have chosen
to work with real integers
and not test
for overflows.
We also do not permit aliasing.
By avoiding the more sophisticated translation steps, we obtain
LCTRSs which are correspondingly easier to analyze.

The LCTRSs
from this transformation are well behaved: all rules are left-linear
and non-overlapping,\footnote{Non-overlappingness means that
  for every term $s$ and rule $\rho\colon\ell \arrz r\ 
  \constraint{\varphi}$ such that $s$ reduces with $\rho$ at the root
  position: (a) there are no other rules $\rho'$ such that $s$ reduces
  with $\rho'$ at the root position, and (b) if $s$ reduces with any
  rule at a non-root position $q$, then $q$ is not a position of
  $\ell$.  For our translations, this holds because (a) rules with the
  same defined symbol have either incompatible constraints or
  non-unifiable arguments, and (b) in a rule $\afun(\ell_1,\dots,
  \ell_n) \arrz r\ \constraint{\varphi}$, the terms $\ell_i$ do not
  contain defined or calculation symbols.}
and have the property that all ground terms
can be reduced or are constructor terms.
Rules $\ell \to r\ \constraint{\varphi}$ \emph{can}
have variables in $r$ or $\varphi$ which do not occur
in $\ell$:  this is mostly due to unspecified values in the
C code.  Where such variables do not
occur---or are removed in the optimization step---the resulting
LCTRSs are confluent.

\section{Rewriting Induction for LCTRSs}
\label{sec:ri}

In this section, we adapt the inference rules
from~\cite{red:90,fal:kap:12,sak:nis:sak:sak:kus:09} to
inductive theorem proving with LCTRSs.  This provides the core
theory for rewriting induction,
strengthened with
two generalization techniques in \secshort\ref{sec:lemma-gen}.

We start by listing some restrictions we need to impose on LCTRSs
for the method to work (\secshort\ref{subsec:ri:restrictions}).
Then, we provide the theory for the technique
(\secshort\ref{subsec:ri:main})
and some illustrative examples
(\secshort\ref{subsec:ri-examples}).
Compared to older definitions of rewriting induction,
we make several changes to best handle the new formalism.
We complete by proving correctness
(\secshort\ref{subsec:ri-correctness-proof}).

\subsection{Restrictions}\label{subsec:ri:restrictions}

In order for rewriting induction to be successful, we need to impose
certain restrictions.

\begin{definition}\label{def:rirestrictions}
In the following, we limit interest to LCTRSs which satisfy
(\ref{req:corepresent})--(\ref{red:groundinstances}):

\begin{enumerate}
\item\label{req:corepresent}
  all core theory symbols are present in $\Sigmalogic$:
  $\Sigmalogic \supseteq \Sigmacore$;
\item\label{req:termination}
  the LCTRS is \emph{terminating}: there is no infinite reduction
  $\aterm_1 \arr{\Rules} \aterm_2 \arr{\Rules} \cdots$;
\item\label{req:completereducible}
  the system is \emph{quasi-reductive}: i.e., for every ground
  term $\aterm$
  either $\aterm \in \Terms(\Constructors,\emptyset)$ (we say $\aterm$
  is a \emph{ground constructor term}), or there is some $\bterm$ such
  that $\aterm \arr{\Rules} \bterm$;
\item\label{red:groundinstances}
  there are ground terms of every sort occurring in
  $\Sigma$.
\end{enumerate}
\end{definition}

Property~\ref{req:corepresent} is the standard assumption
from \secshort\ref{sec:prelim}.  We will need symbols such as
$=$, $\wedge$ and $\Rightarrow$ to add new information to a constraint.
Termination (property~\ref{req:termination})
essentially indicates that a program cannot run indefinitely; this
is crucial for our inductive reasoning, as the method uses induction
on an extension of $\arr{\Rules}$ on terms.

Property~\ref{req:completereducible} indicates that an evaluation
cannot get ``stuck''; roughly, that pattern matching and case
analysis are exhaustive.
Termination and quasi-reductivity together ensure that every ground
term reduces to a constructor term.  This makes it
possible to do an exhaustive case analysis on the rules applicable to
an equation, and lets us assume that variables are always
instantiated by ground constructor terms.

The last property is natural, since inductive theorem proving
makes a statement on \emph{ground} terms; there is no point
in regarding empty sorts.
Together with quasi-reductivity and termination, this
implies that all sorts admit ground \emph{constructor} terms.

Methods to prove both quasi-reductivity and termination have
previously been published for different styles of constrained
rewriting; see e.g.\ \cite{fal:kap:12} for quasi-reductivity
and~\cite{fal:09,sak:nis:sak:11} for termination.  These methods are
easily adapted to LCTRSs.
Quasi-reductivity is handled in~\cite{quasireductive} and is
moreover always satisfied by systems obtained from the
transformations in \secshort\ref{sec:transformations}.
Some basics of termination analysis for LCTRSs are discussed
in~\cite{kop:13}.

\begin{example}\label{ex:running}
As a running example in this section, we will consider
$\Rules_\fact$, which combines the factorial function from
\exshort\ref{exa:fact1 cnt 2} with a recursive variant obtained
from \texttt{int fact(int x) \{ if (x <= 1) return 1; else return
x * fact(x - 1); \}}.
\[
\begin{array}{lrclllrcll}
\text{(1)} & \symb{factiter}(x) & \arrz & \symb{iter}(x,\one,\one) & &
\text{(4)} & \symb{factrec}(x) & \arrz & \return(\one) &
  \constraint{x \leq \one} \\
\text{(2)} & \symb{iter}(x, z, i) & \arrz & \symb{iter}(x, z * i, i + \one) &
  \constraint{i \leq x} &
\text{(5)} & \symb{factrec}(x) & \arrz & \\
\text{(3)} & \symb{iter}(x, z, i) & \arrz & \return(z) & \constraint{i > x} &
  \multicolumn{4}{r}{\symb{mul}(x,\symb{factrec}(x-\one))} & \constraint{x > 1} \\
& & & & & 
\text{(6)} & \symb{mul}(x,\return(y)) & \arrz & \return(x * y) \\
\end{array}
\]
(Function symbols were renamed for readability.)
We can choose a signature which includes $\Sigmacore$, and
each of the sorts---$\Int,\BOOL,\symb{result}$---clearly admits
ground terms (e.g., $\nul,\sfalse,\return(\nul)$).
The system was obtained using \secshort\ref{sec:transformations},
so is quasi-reductive.
Termination follows because in the recursive rule (2), the value $x-i$
is decreased, while bounded from below by $0$, and in the
recursion in rule (5), $x$ decreases against the bound $1$.
This could be proved using, e.g.,
interpretations with support for built-in integers and non-theory
symbols~\cite{fuh:gie:plu:sch:fal:09}, and is automatically handled by
our tool \ctrl.
\end{example}

\subsection{Rewriting Induction}\label{subsec:ri:main}

We now introduce the notions of \emph{constrained equations} and
\emph{inductive theorems}.

\begin{definition}
\label{def:constrained_equation}
A \emph{(constrained) equation} is a triple $\aterm \approx
\bterm\ \constraint{\varphi}$ with $\aterm$ and $\bterm$ terms
and $\varphi$ a constraint.
We write $\aterm \simeq \bterm\ \constraint{\varphi}$ to denote either
$\aterm \approx \bterm\ \constraint{\varphi}$ or
$\bterm \approx \aterm\ \constraint{\varphi}$.
A substitution $\gamma$ \emph{respects} $s \approx t\ 
\constraint{\varphi}$ if $\gamma$ respects $\varphi$ and
$\FV(\aterm) \cup \FV(\bterm) \subseteq \domain(\gamma)$; it is called
a \emph{ground constructor substitution} if all $\gamma(\avar)$ with
$x\in\domain(\gamma)$ are ground constructor terms.

An equation $\aterm \approx \bterm\ \constraint{\varphi}$ is an
\emph{inductive theorem} of an LCTRS $\Rules$
if $\aterm\gamma \leftrightarrow_\Rules^* \bterm\gamma$ for any ground
constructor substitution $\gamma$ that respects this equation.
\end{definition}

Intuitively, if an equation $\afun(\seq{\avar}) \approx
\bfun(\seq{\avar})\ \constraint{\varphi}$ is an inductive theorem,
then $\afun$ and $\bfun$ define the same function (conditional on
$\varphi$, and assuming confluence). 
As we require termination,
we thus consider \emph{total equivalence} in the categorization of
\citeN{god:str:08}:
on all inputs, both programs terminate and return the same values.

To prove that an equation is an inductive theorem, we consider nine
inference rules,
  in \secsshort\ref{subsubsec:simplification}--\ref{subsubsec:completeness}.
Four originate in~\cite{red:90}; three are
based on extensions \cite{bou:97,fal:kap:12,sak:nis:sak:sak:kus:09};
two are new.
All these rules modify a triple $(\E,\H,b)$, called a
\emph{proof state}. Here, $\E$ is a set of equations, $\H$ a set
of rules with $\arr{\Rules \cup \H}$
terminating, and $b \in \{$\complflag, \incomplflag$\}.$
A rule in $\H$ plays the role of an \emph{induction hypothesis} for
``proving'' the equations in $\E$ and is called an \emph{induction
rule}. The flag $b$ indicates whether we can use the current proof
state to \emph{refute} that the initial equation is an inductive
theorem; we can do so if $b = \complflag$.

The definition of these rules is used in the following
result, proved in \secshort\ref{subsec:ri-correctness-proof}.

\begin{theorem}\label{thm:ri-correctness}
Let an LCTRS with rules $\Rules$ and signature $\Sigma$, satisfying
the restrictions from \defshort\ref{def:rirestrictions}, be given.
Let $\E$ be a finite set of equations and let $\flag = \complflag$
if we can confirm that $\Rules$ is confluent
and $\flag = \incomplflag$ otherwise.
If $(\E,\emptyset,\flag) \vdashristar (\emptyset,\H, \flag')$
for some $\H,\flag'$, then every equation in $\E$ is an
inductive theorem of $\Rules$.
If $(\E,\emptyset,\flag) \vdashristar \bot$,
then there is some equation in $\E$
that is not an inductive theorem of $\Rules$.
\end{theorem}

\begin{example}\label{ex:running:start}
We will illustrate the various rules
by proving that
$\symb{factrec}$ and $\symb{factiter}$
are equivalent on positive input,\footnote{We limit interest to
positive input for demonstration purposes only: these functions give
the same result on \emph{all} input, but considering only $n \geq
\one$ allows us to apply the inference rules in a convenient order.}
by showing
that (FCT.A) is an inductive theorem:
\[
\mbox{(FCT.A)}\ \ \symb{factrec}(n) \approx \symb{factiter}(n)\ 
  \constraint{n \geq \one}
\]
$\Rules_{\fact}$ is confluent: as
seen in \secshort\ref{subsec:trafo_remarks},
it is left-linear and non-overlapping, and the right-hand sides
do not introduce fresh
variables, so confluence is given by \cite[\thmshort{}4]{kop:nis:13}.
Thus, we will start with the proof state
$(\ \{\ \mbox{(FCT.A)}\ \},\ \emptyset,\ \complflag\ )$.
\end{example}

Let us now define the nine inference rules to reduce proof states.

\subsubsection{\normalfont\textsc{Simplification}}
\label{subsubsec:simplification}
Our first inference rule originates in~\cite{red:90} and can be
considered one of the core rules of rewriting induction.

\begin{definition}
If $\coterm{\aterm \approx \bterm}{\varphi} \arr{\Rules \cup \H}
  \coterm{\cterm \approx \bterm}{\psi}$, where $\approx$ is seen as
  a fresh constructor for the purpose of constrained term
  reduction,\footnote{It does not suffice if
$\coterm{\aterm}{\varphi} \arrz_\Rules \coterm{\cterm}{\psi}$:
when reducing constrained terms, unused variables may be manipulated
at will, which causes problems if they are used in
$\bterm$.  For example,
\[\coterm{f(\avar+\nul)}{\avar > \bvar}
\equalgen \coterm{f(\avar+\nul)}{\cvar = \avar+\nul}
\arrzbase \coterm{f(\cvar)}{\cvar = \avar+\nul}
\equalgen \coterm{f(\avar)}{\avar < \bvar}\]
but we should certainly not replace
an equation $f(\avar+\nul) \approx g(\bvar)\ \constraint{\avar >
\bvar}$ by $f(\avar) \approx g(\bvar)\ \constraint{\avar < \bvar}$.}
then we may derive:
\[
(\E \uplus \{ (\aterm \simeq \bterm\ \constraint{\varphi}) \}, \H, b)
\vdashri (\E \cup \{ (\cterm \approx \bterm\ \constraint{\psi})\}, \H, b)
\]
\end{definition}

This inference rule allows us to reduce one side of an equation.
This is altered from Reddy's definition by using constrained
rather than normal reduction.

\begin{example}\label{ex:running:simplify}
Following \exshort\ref{ex:running:start}, we observe that
$\symb{factiter}(n)$ can be reduced by the unconstrained rule (1).
Thus, using \textsc{Simplification} we obtain the proof state:
\[
\left(\ 
\left\{
\begin{array}{cc}
  \mbox{(FCT.B)}: &
    \symb{iter}(n,\one,\one) \approx \symb{factrec}(n)\ 
    \constraint{n \geq \one}
  \\
\end{array}
\right\},\ 
\emptyset,\ 
\complflag\ 
\right)
\]
Here we reduce the right-hand side of the equation
(recall that $s \simeq t$ in the rule means
$s \approx t$ or $t \approx s$);
the reduced term
moves to
the left-hand side of the new equation.
Next, observe that $\symb{iter}(n,\one,\one)$ can be reduced by rule
(2) if $n \geq \one$; \textsc{Simplification} then gives:
\[
\left(\ 
\left\{
\begin{array}{cc}
  \mbox{(FCT.C)}: &
    \symb{iter}(n,\one * \one,\one + \one) \approx
    \symb{factrec}(n)\ \constraint{n \geq \one}
  \\
\end{array}
\right\},\ 
\emptyset,\ 
\complflag\ 
\right)
\]
\end{example}

Recall that constrained reduction also allows for steps with
calculation rules; see, e.g., \exshort\ref{ex:factconstrained}.  The
added complexity is that we must decide how to handle the fresh
variable these rules introduce.  In this paper we will use the
following strategy:
\begin{itemize}
\item if $\aterm \arrzcalc \cterm$ then $\coterm{\aterm \approx
  \bterm}{\varphi}$ is simplified
  to $\coterm{\cterm \approx \bterm}{\varphi}$, e.g.\ $\symb{f}(\nul +
  \one) \approx r\ \constraint{\varphi}$ reduces to $\symb{f}(\one)
  \approx r\ \constraint{\varphi}$;
\item a calculation containing variables can be replaced by a
  fresh variable, which is defined in the (updated) constraint,
  e.g.\ $\symb{f}(\avar + \one) \approx r\ \constraint{\varphi}$
  reduces to $\symb{f}(\bvar) \approx r\ \constraint{\varphi \wedge
  \bvar = \avar + \one}$;
  if such a definition already occurs in the constraint, the
  relevant variable is used instead, e.g.\ 
  $\symb{f}(\avar + \one) \approx r\ \constraint{\varphi \wedge \bvar
  = \avar + \one}$ reduces to $\symb{f}(\bvar) \approx r\ 
  \constraint{\varphi \wedge \bvar = \avar + \one}$.
\end{itemize}

\begin{example}\label{ex:simple:calcsimplify}
The proof state from \exshort\ref{ex:running:simplify} is further
simplified to:
\[
\left(\ 
\left\{
\begin{array}{cc}
  \mbox{(FCT.D)}: &
    \symb{iter}(n,\one,\symb{2}) \approx
    \symb{factrec}(n)\ \constraint{n \geq \one}
  \\
\end{array}
\right\},\ 
\emptyset,\ 
\complflag\ 
\right)
\]
\end{example}

\subsubsection{\normalfont\textsc{Expansion}}
Our second core rule also originates from~\cite{red:90}, but has
been more heavily adapted to support irregular rules.

\begin{definition}\label{def:expansion}
Let $s,t$ be terms and $\varphi$ a constraint, all with variables
distinct from those in $\Rules$ (we can always rename the variables in
the rules to support this), and $p$ a position of $s$.
Let $\Expdapp{\aterm}{\bterm}{\varphi}{p}$ be a set of equations containing, for all
rules $\ell \arrz r\ \constraint{\psi} \in \Rules$
such that $\ell$ is unifiable with $\subpos{\aterm}{p}$ with most
general unifier $\gamma$ and $\gamma(x) \in \Values \cup \setvars$ for all $x \in
\FV(\varphi) \cup \FV(\psi)$, an equation $\aterm' \approx \bterm'\ \constraint{\varphi'}$ where
$\aterm \gamma \approx \bterm\gamma\ 
\constraint{(\varphi\gamma) \wedge (\psi\gamma)} \arr{\Rules}
\aterm' \approx \bterm'\ \constraint{\varphi'}$ with rule $\ell \arrz
r\ \constraint{\psi}$ at position $1 \cdot p$.  Here, as in
\textsc{Simplification}, $\approx$ is seen as a fresh constructor for
the
reduction.
If
$\subpos{\aterm}{p}$ is \emph{basic} 
(i.e., $\subpos{\aterm}{p} = \afun(\aterm_1,\ldots,\aterm_n)$ with $\afun
\in \Defineds$
and all $\aterm_i$ constructor terms),
we may derive:
\[
(\E \uplus \{ \aterm \simeq \bterm\ \constraint{\varphi} \}, \H, b)
\vdashri (\E \cup \Expdapp{\aterm}{\bterm}{\varphi}{p}, \H, b)
\]
If, moreover, $\Rules \cup \H \cup \{ \aterm \arrz \bterm\ 
\constraint{\varphi} \}$ is terminating, we may even derive:
\[
(\E \uplus \{ \aterm \simeq \bterm\ \constraint{\varphi} \}, \H, b)
\vdashri (\E \cup \Expdapp{\aterm}{\bterm}{\varphi}{p}, \H \cup \{ \aterm
\arrz \bterm\ \constraint{\varphi} \}, b)
\]
\end{definition}

Intuitively, this inference rule
uses narrowing for
a \emph{case analysis}:
$\Expd$ generates all resulting equations if a ground
constructor instance of $\aterm \approx \bterm\ \constraint{\varphi}$
is reduced at position $p$ of $\aterm$.  In addition,
we save the current equation as a rule to take an induction step.

\begin{example}\label{ex:running:expand}
Following \exshort\ref{ex:simple:calcsimplify}, we consider which
rules may apply to an instance of $\symb{factrec}(n)$ with $n \geq
\one$.
For $\Expdapp{\symb{factrec}(n)}{\symb{iter}(n,\one,\symb{2})}{n \geq \one}{\epsilon}$, we choose:
\[
\left\{
\begin{array}{lrcll}
\mbox{(FCT.E):} & \return(\one) & \approx & \symb{iter}(n,\one,
  \symb{2}) & \constraint{n \geq \one \wedge n \leq \one}, \\
\mbox{(FCT.F):} & \symb{mul}(n,\symb{factrec}(n-\one)) & \approx &
  \symb{iter}(n,\one,\symb{2}) & \constraint{n \geq \one \wedge n >
  \one} \\
\end{array}
\right\}
\]
In both cases we used the unifier $\gamma = [x:=n]$.
If we write $(\mbox{FCT.D}^{-1})$ for the rule
generated from the inverse of \mbox{(FCT.D)}---so
$\symb{factrec}(n) \arrz \symb{iter}(n,
\one,\symb{2})\ \constraint{n \geq \one}$---$\Rules
\cup \{ (\mbox{FCT.D}^{-1}) \}$ is terminating
as the new rule does not cause mutual recursion between
$\symb{iter}$ and $\symb{factrec}$.
We continue with
$(\ \{ \mbox{(FCT.E)}, \mbox{(FCT.F)} \},\ 
\{ (\mbox{FCT.D}^{-1}) \},\ \complflag\ )$.
Now we can show the second kind of calculation step, using
\textsc{Simplification} on (FCT.F), which gives:
\[
\left(
\left\{
\begin{array}{crcll}
\mbox{(FCT.E):} & \return(\one) & \approx & \symb{iter}(n,\one,
  \symb{2}) & \constraint{n \geq \one \wedge n \leq \one}, \\
\mbox{(FCT.G):} & \symb{mul}(n,\symb{factrec}(m)) & \approx &
  \symb{iter}(n,\one,\symb{2}) & \constraint{n > \one \wedge m = n -
  \one} \\
\end{array}
\right\},
\begin{array}{c}
\{ \mbox{(FCT.D$^{-1}$)} \}, \\
\complflag
\end{array}
\right)
\]
Here, we also removed the redundant clause $n \geq \one$, which is
allowed by definition of $\arr{\Rules}$ on constrained terms.
As $n \geq \one \wedge n \leq \one$
implies $n = \one$, we may use \textsc{Simplification} with rule
(3) on (FCT.E), and with rule (2) followed by calculations on
(FCT.G), to get:
\[
\left(
\left\{
\begin{array}{crcll}
\mbox{(FCT.H):} & \return(\one) & \approx & \return(\one) &
  \constraint{n = \one}, \\
\mbox{(FCT.I):} & \symb{iter}(n,\symb{2},\symb{3}) & \approx &
  \symb{mul}(n,\symb{factrec}(m)) & \constraint{n > \one \wedge m = n -
  \one} \\
\end{array}
\right\},
\begin{array}{c}
\{ \mbox{(FCT.D$^{-1}$)} \}, \\
\complflag
\end{array}
\right)
\]
Now we can use
``induction'':
we
eliminate the occurrence of $\symb{factrec}$ with a
\textsc{Simplifi\-cation} step using the induction rule
\mbox{(FCT.D$^{-1}$)} and substitution $[n:=m]$.  This gives:
\[
\left(
\left\{
\begin{array}{crcll}
\mbox{(FCT.H):} & \return(\one) & \approx & \return(\one) &
  \constraint{n = \one}, \\
\mbox{(FCT.J):} & \symb{mul}(n,\symb{iter}(m,\symb{1},\symb{2})) &
  \approx & \symb{iter}(n,\symb{2},\symb{3}) & \constraint{n > \one
  \wedge m = n - \one} \\
\end{array}
\right\},
\begin{array}{c}
\{ \mbox{(FCT.D$^{-1}$)} \}, \\
\complflag
\end{array}
\right)
\]
\end{example}

Note that the choice of $\Expd$ is non-deterministic, as it uses
reduction of constrained terms.
The most natural choice for $\Expdapp{s}{t}{\varphi}{p}$---which we
use in examples---is
\begin{gather*}
\{\;\; \subreplace{\aterm}{r}{p}\gamma \approx \bterm\gamma\ 
  \constraint{(\varphi\gamma) \wedge (\psi\gamma)} \; \mid \;
          \ell \arrz r\ \constraint{\psi} \in \Rules,\;
          \aterm|_p\ \text{unifies with}\ \ell\ \text{with mgu}\ 
          \gamma
   \;\; \}
 \end{gather*}
However, for \emph{irregular} rules in particular, it may be strategic
to choose a different set.  Consider for example a (non-confluent)
LCTRS with rules $\afun(x) \arrz \bfun(y)\ \constraint{x > \nul \wedge
x > y}$ and $\afun(x) \arrz \bfun(y)\ \constraint{x \leq \nul \wedge
x \leq y}$.  With the choice for $\Expdapp{s}{t}{\varphi}{p}$
above, an equation $\afun(x) \approx \bfun(\nul)\ \constraint{\strue}$
results in $\{\ \bfun(y) \approx \bfun(\nul)\ \constraint{x > \nul
\wedge x > y},\ \bfun(y) \approx \bfun(\nul)\ \constraint{x \leq \nul
\wedge x \leq y}\ \}$.  If $\bfun$ is a constructor, neither of these
equations can be handled.  Using the full definition of
\textsc{Expansion}, we can choose $\bfun(\nul) \approx \bfun(\nul)\ 
\constraint{\strue}$ for both equations.

Also note that there is no choice in the orientation of the rule
added to $\H$: this is determined by the side of the equation on which
the expansion was applied.  Thus, in \exshort\ref{ex:running:expand}
we were not allowed to add (FCT.D) instead of
(FCT.D$^{-1}$).

Our definition of \textsc{Expansion} differs from both its
original and existing work on constrained rewriting induction.
To start, those works define $\Expdapp{s}{t}{\varphi}{p}$ simply
as the ``natural choice'' given above.
Second, we included a case where no rule is added, to allow for
progress when adding the rule might cause non-termination.
Forms of this case appear as a separate rule in other work,
e.g., \textsc{Case Analysis} in~\cite{bou:97} and
\textsc{Rewrite/Partial Splitting}
\journalOrReport{in~\cite{bou:jac:08,bou:jac:08:corr}}{by Bouhoula and Jacquemard \citeNN{bou:jac:08,bou:jac:08:corr}}.
A weaker form with constraints is given
in~\cite{fal:kap:12} (\textsc{Case-Simplify}).

\subsubsection{\normalfont\textsc{Deletion}} The last of the
core rules serves to remove solved equations from $\E$.

\begin{definition}
If $\aterm = \bterm$ or $\varphi$ is not satisfiable, we can delete
$\aterm \approx \bterm\ \constraint{\varphi}$ from $\E$:
\[
(\E \uplus \{ \aterm \approx \bterm\ \constraint{\varphi} \}, \H, b)
\vdashri (\E,\H, b)
\]
\end{definition}

Compared to the corresponding rule in~\cite{red:90}, the
unsatisfiability case is new; it is similar to the corresponding
rules in~\cite{sak:nis:sak:sak:kus:09,fal:kap:12}.

\begin{example}\label{ex:running:delete}
Following \exshort\ref{ex:running:expand}, the left- and
right-hand side of (FCT.H) are the same, so we may remove the
equation with \textsc{Deletion}, obtaining
$
(\ \{\ \mbox{(FCT.J)}\ \},\linebreak \{\ \mbox{(FCT.D$^{-1}$)}\ \},\ 
\complflag\ )
$.
We will see the other form of \textsc{Deletion} in
\exshort\ref{ex:running:eqdelete}.
\end{example}

\subsubsection{\normalfont\textsc{Postulate}}
Sometimes it is useful to make the problem \emph{seemingly} harder.
To this
end, we consider the last inference rule from~\cite{red:90}.

\begin{definition} For any set of equations $\E'$, we can derive:
\[
(\E,\H,b) \vdashri (\E \cup \E', \H, \incomplflag)
\]
\end{definition}

The \textsc{Postulate} rule allows us to add additional equations to
$\E$ (although at a price: we cannot conclude non-equivalence after
adding a potentially unsound equation).  The reason to do so is that
in proving the equations in $\E'$ to be inductive theorems,
we may derive new induction rules.
These can then be used to simplify the elements of $\E$.

\begin{example}\label{ex:running:postulate}
Following \exshort\ref{ex:running:delete},
\textsc{Expansion} followed by \textsc{Simplification} gives:
\[
\begin{array}{crcll}
\mbox{(FCT.K):} &
  \symb{mul}(n,\symb{iter}(m,\symb{2},\symb{3})) &
  \approx &
  \symb{iter}(n,\symb{6},\symb{4}) &
  \constraint{n \geq \symb{3} \wedge m = n - \one}
\end{array}
\]
But now a pattern starts to arise.  Expanding and fully simplifying
again, we obtain:
\[
\begin{array}{crcll}
\mbox{(FCT.L):} &
  \symb{mul}(n,\symb{iter}(m,\symb{6},\symb{4})) &
  \approx &
  \symb{iter}(n,\symb{24},\symb{5}) &
  \constraint{n \geq \symb{4} \wedge m = n - \one}
\end{array}
\]
And so on.  Here,
(FCT.K) cannot be handled by the induction rule
(FCT.J$^{-1}$), nor can (FCT.L) be handled by
(FCT.K$^{-1}$).
We have a \emph{divergence}: a sequence of increasingly
complex equations, each generated from the same leg in an
\textsc{Expansion} (see also the \emph{divergence critic}
in~\cite{wal:96}).  Yet the previous induction rules never apply
to the new equation.
This suggests we need a lemma equation.  We use
\textsc{Postulate}
to \pagebreak get:
\[
\left(
\left\{
\begin{array}{cc}
\mbox{(FCT.J):} & \symb{mul}(n,\symb{iter}(m,\symb{1},\symb{2}))
  \approx \symb{iter}(n,\symb{2},\symb{3}) \\
  & \constraint{n > \one \wedge m = n - \one} \\
\mbox{(FCT.M):} & \symb{mul}(n,\symb{iter}(m,x,y)) \approx
  \symb{iter}(n,x',y') \\
  & \constraint{n \geq y \wedge m = n - \one \wedge y' = y + 1 \wedge
    x' = x * y} \\
\end{array}
\right\},
\begin{array}{c}
\{ \mbox{(FCT.D$^{-1}$)} \}, \\
\incomplflag
\end{array}
\right)
\]
Using \textsc{Expansion} on the right-hand of (FCT.M), we have:
\[
\left(
\left\{
\begin{array}{cc}
\mbox{(FCT.J):} & \symb{mul}(n,\symb{iter}(m,\symb{1},\symb{2}))
  \approx \symb{iter}(n,\symb{2},\symb{3}) \\
  & \constraint{n > \one \wedge m = n - \one} \\
\mbox{(FCT.N):} & \symb{iter}(n,x' * y',y' + \one) \approx
  \symb{mul}(n,\symb{iter}(m,x,y)) \\
  \multicolumn{2}{r}{
    \phantom{ABC}\constraint{n \geq y \wedge m = n - \one \wedge y' = y + 1 \wedge
    x' = x * y \wedge y' \leq n}} \\
\mbox{(FCT.O):} & \return(x') \approx
  \symb{mul}(n,\symb{iter}(m,x,y)) \\
  \multicolumn{2}{r}{
    \constraint{n \geq y \wedge m = n - \one \wedge y' = y + 1 \wedge
    x' = x * y \wedge y' > n}} \\
\end{array}
\right\},
\begin{array}{c}
\left\{
\begin{array}{c}
  \mbox{(FCT.D$^{-1}$)} \\
  \mbox{(FCT.M$^{-1}$)} \\
\end{array}
\right\}, \\
\incomplflag
\end{array}
\right)
\]
But now we have added (FCT.M$^{-1}$) as an induction rule.  As a
result---since $n > \one$ clearly implies $n \geq \symb{2}$---we can
use \textsc{Simplification} with a substitution
$[n:=n,x:=\one,y:=\symb{2},x':=\symb{2},y':=\symb{3}]$ to reduce
(FCT.J) to the equation $\symb{mul}(n,\symb{iter}(m,\one,\symb{2}))
\approx \symb{mul}(n,\symb{iter}(m,\one,\symb{2}))\ \constraint{\dots
}$, which we may immediately remove by \textsc{Deletion}.  We continue
with the proof state
$(\{\ \mbox{(FCT.N)}, \mbox{(FCT.O)}\ \},\ \{\ \mbox{(FCT.D$^{-1}$)}
,\ \mbox{(FCT.M$^{-1}$)}\ \},\ \incomplflag)$.
\end{example}

Although
the need
to choose arbitrary new equations for use in \textsc{Postulate} may
seem somewhat problematic, this is actually a key step.  Complex
theorems typically require more than straight induction, both in
our
setting and in mathematical proofs in general.
Thus,
generation of suitable \emph{lemma equations} $\E'$ is not only
part, but even at the heart, of inductive theorem proving.
Hence,
this subject has been extensively investigated
\cite{bun:bas:hut:ire:05,kap:sak:03,kap:sub:96,nak:nis:kus:sak:sak:10,urs:kou:04,wal:96},
and a large variety of lemma generation techniques exist, at
least in the setting without constraints.

\subsubsection{\normalfont\textsc{Generalization}}
A very typical use of \textsc{Postulate} is to \emph{generalize}
a problematic equation.
For simplicity, we add a shortcut to do this in one step.

\begin{definition}\label{def:generalization}
If for all substitutions $\gamma$ which respect $\varphi$ there is a
substitution $\delta$ which respects $\psi$ with $\aterm\gamma =
\aterm'\delta$ and $\bterm\gamma = \bterm'\delta$, then we can derive:
\[
(\E \uplus \{ \aterm \approx \bterm\ \constraint{\varphi} \},\H,b)
\vdashri
(\E \cup \{ \aterm' \approx \bterm'\ \constraint{\psi} \},\H,\incomplflag)
\]
\end{definition}

This inference rule is rarely \emph{necessary}: we could
usually add $\aterm' \approx \bterm'\ \constraint{\psi}$ using
\textsc{Postulate}, and use the resulting induction rules to eliminate
$\aterm \approx \bterm\ \constraint{\varphi}$, as we did in
\exshort\ref{ex:running:postulate}.
By generalizing instead, we avoid extra steps, and intuitively, we
strengthen an induction statement rather than
add a separate lemma.
Without constraints,
\textsc{Generalization} can be seen as a combination of
\textsc{Postulate} and the \textsc{Subsumption} rule in~\cite{bou:97}.
As there are several results for
generalizing equations in the literature
\cite{bun:ste:vanh:ire:sma:93,bun:bas:hut:ire:05,bas:wal:92,wal:96,urs:kou:04},
the combination is useful beyond just this paper.

\begin{example}\label{ex:running:generalize}
In \exshort\ref{ex:running:postulate}, we could have used
\textsc{Generalization} immediately to move from the proof state
$(\{\ \mbox{(FCT.J)}\ \},\ \{\ \mbox{(FCT.D$^{-1}$)}\ \},\ 
\incomplflag)$ to $(\{\ \mbox{(FCT.M)}\ \},\ 
\{\ \mbox{(FCT.D$^{-1}$)}\ \},\ \incomplflag)$.
\end{example}

\subsubsection{\normalfont\textsc{EQ-deletion}}
The following rule, which was adapted
from~\cite{sak:nis:sak:sak:kus:09}, provides a link between the
equation part $\aterm \approx \bterm$ and the constraint.

\begin{definition}
Let $C$ be an arbitrary context with $n$ holes ($C$ may contain
symbols in $\Sigmalogic$). If all $\aterm_i,\bterm_i \in
\Terms(\Sigmalogic,\FV(\varphi))$, then we can derive:
\[
\begin{array}{c}
(\E \uplus \{ C[\aterm_1,\ldots,\aterm_n] \simeq C[\bterm_1,\ldots,
\bterm_n]\ \constraint{\varphi} \}, \H, b) \vdashri \\
(\E \cup \{
C[\aterm_1,\ldots,\aterm_n] \approx C[\bterm_1,\ldots,\bterm_n]\ 
\constraint{\varphi \wedge \neg (\bigwedge_{i = 1}^n \aterm_i =
\bterm_i)} \}, \H, b)
\end{array}
\]
\end{definition}

Intuitively, if $\bigwedge_{i=1}^n \aterm_i = \bterm_i$ holds, then
$C[\aterm_1,\ldots,\aterm_n]\gamma \leftrightarrow_{\Rulescalc}^*
C[\bterm_1,\ldots,\bterm_n]\gamma$, so we are done.
\textsc{EQ-deletion} excludes this case from the equation.
In combination with \textsc{Deletion}, this rule gives a more
general variation of
\textsc{Theory${}_\top$} in~\cite{fal:kap:12}.

\begin{example}\label{ex:running:eqdelete}
Continuing from \exshort\ref{ex:running:postulate} (or
\exshort\ref{ex:running:generalize}), we observe that $n \geq y$,
$y' = y + \one$ and $y' > n$ together imply $n = y$, and with $m = n
- \one$ we thus have $y > m$ as well.  Therefore,
\textsc{Simplification} on (FCT.O) by rule (3) followed by (6) gives:
\[
\mbox{(FCT.P):}\quad \return(n * x) \approx \return(x')\ 
    \constraint{n = y \wedge m = n - \one \wedge y' = y + 1 \wedge
    x' = x * y}
\]
We can use \textsc{EQ-deletion} with the context $C[\Box] = \return(
\Box)$ to replace (FCT.P) by:
\[
\mbox{(FCT.Q):}\ \ 
\return(n * x) \approx \return(x')\ 
\constraint{n = y \wedge m = n - \one \wedge y' = y + 1 \wedge
 x' = x * y \wedge \neg (n * x = x')}
\]
As $n = y$ and $x' = x * y$ together imply that $n * x = x'$, the
constraint of this equation is not satisfiable.  We may remove it
using \textsc{Deletion}, giving the proof state
$(\{\ \mbox{(FCT.N)}\ \},\ \{\ \mbox{(FCT.D$^{-1}$)},
\mbox{(FCT.M$^{-1}$)}\ \},\ \incomplflag)$.
\end{example}

\textsc{EQ-deletion} is among the core rules
for constrained rewriting induction: almost all inductive proofs use
it, in contrast to the remaining three inference rules.

\begin{example}\label{ex:running:complete}
To complete our example, consider (FCT.N).  As $y + \one =
y' \leq n \wedge m = n - \one$ implies $y \leq m$, we may apply
\textsc{Simplification} with rule (2) to replace it by:
\[
\begin{array}{cc}
\mbox{(FCT.R):} & \symb{mul}(n,\symb{iter}(m,x * y,y + \one)) \approx
  \symb{iter}(n,x' * y',y' + \one)\phantom{A} \\
  \multicolumn{2}{r}{
    \constraint{n \geq y \wedge m = n - \one \wedge y' = y + 1 \wedge
    x' = x * y \wedge y' \leq n}} \\
\end{array}
\]
Then, using \textsc{Simplification} with calculations (and observing
that both $x * y$ and $y+\one$ are ``defined'' in the constraint, as
discussed in \secshort\ref{subsubsec:simplification}), we get:
\[
\begin{array}{cc}
\mbox{(FCT.S):} & \symb{mul}(n,\symb{iter}(m,x',y')) \approx
  \symb{iter}(n,x'',y'') \\
  \multicolumn{2}{r}{
    \phantom{Ab}\constraint{n \geq y' \wedge m = n - \one \wedge
    x' = x * y \wedge x'' = x' * y' \wedge y'' = y' + \one}} \\
\end{array}
\]
(We removed the clauses with $y$ from the constraint, as $y$ does not
occur in the equation part.)
But now the induction rule (FCT.M$^{-1}$) applies!
As this rule is irregular, we must be careful.
We use the substitution $\gamma = [n:=n,m:=m,x'':=x',
y':=y'',x:=x',y:=y']$, which also affects variables not occurring in
the left-hand side.  The substituted
constraint for the rule is $n \geq y' \wedge m = n - \one \wedge y'' =
 y' + \one \wedge x'' = x' * y'$, which is indeed implied by the
constraint of (FCT.S).  Using \textsc{Simplification}, we thus obtain:
\[
\left(
\left\{
\begin{array}{cc}
\mbox{(FCT.T):} & \symb{mul}(n,\symb{iter}(m,x',y')) \approx
  \symb{mul}(n,\symb{iter}(m,x',y')) \\
  \multicolumn{2}{r}{
    \phantom{Ab}\constraint{n \geq y' \wedge m = n - \one \wedge
    x' = x * y \wedge x'' = x' * y' \wedge y'' = y' + \one}} \\
\end{array}
\right\},
\begin{array}{c}
\left\{
\cdots
\right\}, \\
\incomplflag
\end{array}
\right)
\]
As the left- and right-hand side of the remaining equation are the
same, we may remove it using \textsc{Deletion}.  This leaves a proof
state of the form $(\emptyset,\H,\incomplflag)$, so by
\thmshort\ref{thm:ri-correctness}, the equation
$\symb{factrec}(n) \approx \symb{factiter}(n)\ \constraint{n \geq
\one}$ is an inductive theorem.
\end{example}

\subsubsection{\normalfont\textsc{Constructor}}
Where~\citeN{fal:kap:12} and \citeN{sak:nis:sak:sak:kus:09} focus on
systems with only theory symbols and defined symbols,
here we are also interested in non-theory constructors, such as
$\symb{error}_\afun$ and $\symb{return}_\afun$.
To support this, we add:

\begin{definition}
If $\afun$ is a constructor, we can derive:
\[
(\E \uplus \{ \afun(\aterm_1,\ldots,\aterm_n) \approx
\afun(\bterm_1,\ldots,\bterm_n) \constraint{\varphi} \}, \H, b) \vdashri
(\E \cup \{ \aterm_i \approx \bterm_i\ \constraint{\varphi} \mid 1
\leq i \leq n \}, \H, b)
\]
\end{definition}

The \textsc{Constructor} rule originates in~\cite{bou:97}, where it is
called \textsc{Positive Decomposition},
although variations occur in earlier work on
implicit induction, e.g., \cite{hue:hul:82}.
It is used to split up a large equation into smaller problems.
This inference rule is particularly useful in applications where a
recursive structure, such as a list, is inductively built up, but will
also be invaluable as part of a disproof.

\begin{example}\label{ex:running:constructor}
Suppose that, in \exshort\ref{ex:running:start}, we had
started with
$
\mbox{(BAD.A):}\ \ 
\symb{factiter}(x) \approx \symb{factrec}(x-\one)\ \constraint{\strue}
$.
Following some expansions and simplifications, we arrive at
\[
\left(
\left\{
\begin{array}{rrcll}
\mbox{(BAD.B):} & \return(\symb{2}) & \approx & \return(\one) &
  \constraint{x = \symb{2}} \\
\mbox{(BAD.C):} & \symb{iter}(x,\one,\one) &
  \approx &
  \symb{factrec}(y) &
  \constraint{y = x - \one \wedge y > \one} \\
\end{array}
\right\},
\H,
\complflag
\right)
\]
(for some $\H$).
We can use \textsc{Constructor} to replace (BAD.B) by (BAD.D):
$\symb{2} \approx \one\ \constraint{x = \symb{2}}$.
\end{example}

\subsubsection{\normalfont\textsc{Disprove}}
Recall that, to
show that an equation is \emph{not} an inductive theorem, we must
derive $\bot$ from a $\complflag$ proof state.  For this, we use
\textsc{Disprove}.

\begin{definition}
Suppose $\vdash \aterm : \asort$ and one of the following holds:
\begin{itemize}
\item $\aterm,\bterm \in \Terms(\Sigmalogic,\setvars)$,\ $\asort$ is
  a theory sort, and
  $\varphi \wedge \aterm \neq \bterm$ is satisfiable;
\item $\aterm = \afun(\seq{\aterm})$ and $\bterm = \bfun(\seq{\bterm})$
  with $\afun,\bfun$ distinct constructors and $\varphi$
  satisfiable;
\item $\aterm \in \setvars \setminus \FV(\varphi)$,\ 
  $\varphi$ is satisfiable, at least two different constructors have
  output sort $\asort$, and either
  $\bterm$ is a variable distinct from $\aterm$ or $\bterm$ has the form
  $\bfun(\seq{\bterm})$ with $\bfun \in \Constructors$;
\end{itemize}
Then we may derive:
\[
(\E \uplus \{ \aterm \simeq \bterm\ \constraint{\varphi} \}, \H, \complflag)
\vdashri \bot
\]
\end{definition}

The first case of this rule
corresponds to \textsc{Theory${}_\top$}
in~\cite{fal:kap:12} and \thmshort{}7.2 in~\cite{sak:nis:sak:sak:kus:09};
note that the restriction to theory sorts only excludes the case
where $s$ and $t$ are non-logical variables.
The second case corresponds
to \textsc{Positive Clash} in~\cite{bou:97}.  The third
case is new in rewriting induction,
but appears in~\cite{hue:hul:82}, an implicit induction method based on
completion.

\begin{example}\label{ex:running:disprove}
Following \exshort\ref{ex:running:constructor}, we observe that
$x = \symb{2} \wedge \symb{2} \neq \symb{1}$ is satisfiable.  Thus,
by \textsc{Disprove} we reduce
$(\{\ 
\mbox{(BAD.D)},
\mbox{(BAD.C)}\ \},\ \H,\ \complflag)$ to $\bot$.
By confluence of $\Rules_{\symb{fact}}$, we see that
$\symb{factiter}(x)$ and $\symb{factrec}(x-1)$ have different normal
forms for some $x$.
\end{example}

\subsubsection{\normalfont\textsc{Completeness}}\label{subsubsec:completeness}
A downside of \textsc{Postulate} and \textsc{Generalization} is the
potential loss of the completeness flag.
To weaken this
problem---and empower automatic tools
to combine the search for a proof and a disproof---we add
our final inference rule.

\begin{definition}
For any set of equations $\E$ and $\E' \;\subseteq\; \E$ we can derive:
\[
\begin{array}{rcl}
\text{If}\ (\E,\H,\complflag) & \vdashristar & (\E',\H',\incomplflag) \\
\text{then}\ (\E,\H,\complflag) & \vdashri & (\E', \H', \complflag) \\
\end{array}
\]
\end{definition}

Essentially, \textsc{Completeness} allows us to return the
completeness flag that was lost due to a \textsc{Postulate} or
\textsc{Generalization} step, once we have managed to remove all the
added / generalized lemma equations.  In practice, a tool or human
prover might have a derivation that could be denoted $(\E,\H,
\complflag) {\vdashri}_{\!\mathrm{(\textsc{Postulate})}}\ 
(\E \cup \E',\H,\incomplflag) \vdashri \cdots \vdashri
(\E,\H \cup \H',\incomplflag)
{\vdashri}_{\!\mathrm{(\textsc{Completeness})}}\ 
(\E,\H \cup \H',\complflag)$ by remembering the set $\E$ where the
completeness flag was lost.

\begin{example}
Recall \exshort\ref{ex:running:postulate}.  Starting in
$(\{\ \mbox{(FCT.J)}\ \},\ \{\ \mbox{(FCT.D$^{-1}$)}\ \},\ 
\complflag)$, we lost completeness by adding a lemma equation.
Then, after using \textsc{Expansion}, we arrived at
$(\{\ \mbox{(FCT.J)},\mbox{(FCT.N)}, \mbox{(FCT.O)}\ \},\ 
\{\ \mbox{(FCT.D$^{-1}$)},\mbox{(FCT.M$^{-1}$)}\ \},\ \incomplflag)$.
Applying the proof steps of \exsshort\ref{ex:running:eqdelete}
and~\ref{ex:running:complete} without touching (FCT.J), we could
reduce this state to $(\{\ \mbox{(FCT.J)}\ \},\ \{\ 
\mbox{(FCT.D$^{-1}$)}, \mbox{(FCT.M$^{-1}$)}\ \},\ \incomplflag)$.
But the only equation (FCT.J) in this set is the one we started with.
Thus, we may restore the completeness flag, resulting in
$(\{\ \mbox{(FCT.J)}\ \},\ \{\ \mbox{(FCT.D$^{-1}$)},
\mbox{(FCT.M$^{-1}$)}\ \},\ \complflag)$.
\end{example}

There are many other potential inference rules we could consider,
as various extensions of the base method have been studied in the
literature (see e.g.~\cite{bou:97}).  For now, we stick to these
nine rules and leave the remainder to future work.

\subsection{Examples}\label{subsec:ri-examples}

The running example in \secshort\ref{subsec:ri:main} gives a good
general idea of the power of the method
and the way it is applied.
In this section we present some further examples.  For brevity, we
only list the equations $\E$ in each step, not the completeness flag
or induction rules $\H$.  Unless stated otherwise, these induction
rules are not applicable to new equations.

\begin{example}\label{exa:strlen}
Let us look at an
assignment to implement \texttt{strlen}, a string function which
operates on 0-terminated \texttt{char} arrays.
As \texttt{char} is a numeric data type, we use integer arrays in the
LCTRS translation (although another underlying sort
$\T_{\symb{char}}$ would make little difference).
The example function and its LCTRS translation are as follows:

\noindent
\begin{tabular}{cc}
\begin{minipage}[h]{0.275\textwidth}
\smallskip
\begin{verbatim}
int strlen(char *s) {
  for(int i = 0;;i++){
    if(s[i] == 0)
      return i;
  }
}
\end{verbatim}
\vspace{0pt}
\end{minipage}
&
\begin{minipage}[h]{0.65\textwidth}
\vspace{-8pt}
\[
\begin{array}{rrcll}
\text{(1)} & \symb{strlen}(\avar) & \arrz & \symb{u}(\avar,\nul) \\
\text{(2)} & \symb{u}(\avar,i) & \arrz & \symb{error} &
  \constraint{i < \nul \vee i \geq \symb{size}(\avar)} \\
\text{(3)} & \symb{u}(\avar,i) & \arrz & \symb{return}(i) &
  \constraint{\nul \leq i < \symb{size}(\avar) \wedge
  \symb{select}(\avar,i) = \nul} \\
\text{(4)} & \symb{u}(\avar,i) & \arrz & \symb{u}(\avar,i+1) &
  \constraint{\nul \leq i < \symb{size}(\avar) \wedge
  \symb{select}(\avar,i) \neq \nul} \\
\end{array}
\]
\end{minipage}
\end{tabular}

\noindent
Note that the
bounds checks guarantee termination.
To see that $\symb{strlen}$ does what we would expect it to do, we
want to know that for \emph{valid C strings}, $\symb{strlen}(a)$
returns the first integer $i$ such that $a[i] = 0$.
Following \secshort\ref{subsubsec:classical}, this corresponds to
the equation:
\[
\begin{array}{cc}
\mbox{(LEN.A)} &
\symb{strlen}(\avar) \approx \return(n) \\
&  \constraint{\nul \leq n < \symb{size}(\avar) \wedge
  \bquant{\forall}{i \in \intint{\nul}{n-\one}}{\symb{select}(\avar,i)
  \neq \nul} \wedge \symb{select}(\avar,n) = \nul}
\end{array}
\]
Here, we use bounded quantification, which, as described in
\secshort\ref{subsec:quantify}, can be seen as syntactic sugar for an
additional predicate; the underlying LCTRS could, e.g.,
use a symbol
$\symb{nonzero}$
and replace $\quant{\forall}{i \in
\intint{\nul}{n-\one}}{\symb{select}(\avar,i) \neq \nul}$ by
$\symb{nonzero}(\avar,n)$ in the constraint.

We first use \textsc{Simplification} with rule (1), which gives
(LEN.B):
\[
\begin{array}{cc}
\uu(\avar,\nul) \approx \return(n)\ 
 \constraint{\nul \leq n < \symb{size}(\avar) \wedge
  \bquant{\forall}{i \in \intint{\nul}{n-\one}}{\symb{select}(\avar,i)
  \neq \nul} \wedge \symb{select}(\avar,n) = \nul} \\
\end{array}
\]
We continue with \textsc{Expansion}, again on the left-hand side.
Since the constraint implies that $\nul < \symb{size}(\avar)$, the
error case (2) is unsatisfiable, so we delete it, which leaves:
\[
\begin{array}{clr}
\mbox{(LEN.C)} &
\symb{return}(\nul) \approx \return(n) &
 [\nul \leq n < \symb{size}(\avar) \wedge
  \bquant{\forall}{i \in \intint{\nul}{n-\one}}{\symb{select}(\avar,
  i) \neq \nul}\ \wedge\  \\
& & \symb{select}(\avar,n) = \nul \wedge \nul \leq \nul <
  \symb{size}(\avar) \wedge \symb{select}(\avar,\nul) = \nul] \\
\mbox{(LEN.D)} &
\uu(\avar,\nul+\one) \approx \return(n) &
  [\nul \leq n < \symb{size}(\avar) \wedge
  \bquant{\forall}{i \in \intint{\nul}{n-\one}}{\symb{select}(\avar,i)
  \neq \nul}\ \wedge\  \\
& & \symb{select}(\avar,n) = \nul \wedge \nul \leq \nul <
  \symb{size}(\avar) \wedge \symb{select}(\avar,\nul) \neq \nul] \\
\end{array}
\]

As the constraint of (LEN.C) implies that $n = \nul$,
we can remove (LEN.C) using \textsc{EQ-deletion} and
\textsc{Deletion}. (LEN.D) is simplified with a calculation:
\[
\begin{array}{clr}
\mbox{(LEN.E)} &
\uu(\avar,\one) \approx \return(n) &
  [\nul \leq n < \symb{size}(\avar) \wedge
  \bquant{\forall}{i \in \intint{\nul}{n-\one}}{\symb{select}(\avar,i)
  \neq\nul}\ \wedge\  \\
& & \symb{select}(\avar,n) = \nul \wedge \nul <
  \symb{size}(\avar) \wedge \symb{select}(\avar,\nul) \neq \nul] \\
\end{array}
\]
Which we expand again (once more skipping the $\symb{error}$ case due
to unsatisfiability):
\[
\begin{array}{cl}
\mbox{(LEN.F)} &
\symb{return}(\one) \approx \return(n)\ \ 
  [\nul \leq n < \symb{size}(\avar) \wedge
  \bquant{\forall}{i \in \intint{\nul}{n-\one}}{\symb{select}(\avar,i)
  \neq\nul}\ \wedge\  \\
\multicolumn{2}{r}{
  \symb{select}(\avar,n) = \nul \wedge
  \nul < \symb{size}(\avar) \wedge \symb{select}(\avar,\nul) \neq \nul
  \wedge \nul \leq \one < \symb{size}(\avar) \wedge
  \symb{select}(\avar,\one) = \nul]
} \\
\mbox{(LEN.G)} &
\uu(\avar,\one+\one) \approx \return(n)\ \ 
  [\nul \leq n < \symb{size}(\avar) \wedge
  \bquant{\forall}{i \in \intint{\nul}{n-\one}}{\symb{select}(\avar,i)
  \neq\nul}\ \wedge \\
\multicolumn{2}{r}{
  \symb{select}(\avar,n) = \nul \wedge
  \nul < \symb{size}(\avar) \wedge \symb{select}(\avar,\nul) \neq \nul
  \wedge \nul \leq \one < \symb{size}(\avar) \wedge
  \symb{select}(\avar,\one) \neq \nul]
}\\
\end{array}
\]
The constraint of (LEN.F) implies that $n = \one$, so we easily
remove this equation.  (LEN.G) is simplified using a calculation 
and then expanded again:
\[
\begin{array}{cll}
\mbox{(LEN.H)} &
\symb{return}(\symb{2}) \approx \return(n) &
  \constraint{\dots \wedge \symb{2} < \symb{size}(x) \wedge \symb{select}(x,\symb{2}) = \nul} \\
\mbox{(LEN.I)} &
\uu(\avar,\symb{2}+\one) \approx \return(n) &
  [\nul \leq n < \symb{size}(\avar) \wedge
  \bquant{\forall}{i \in \intint{\nul}{n-\one}}{\symb{select}(\avar,i)
  \neq\nul} \\
& & \multicolumn{1}{r}{\wedge\ \symb{select}(\avar,n) = \nul \wedge
  \nul < \symb{size}(\avar) \wedge \symb{select}(\avar,\nul) \neq \nul
  \wedge \one <\ }\\ & &
  \symb{size}(\avar) \wedge \symb{select}(\avar,\one) \neq \nul
  \wedge
  \symb{2} < \symb{size}(\avar) \wedge \symb{select}(\avar,\symb{2})
  \neq \nul
  ] \\
\end{array}
\]
We drop (LEN.H) easily.  Simplifying (LEN.I) and
reformulating its constraint gives:
\[
\begin{array}{ccr}
\mbox{(LEN.J)}
&
  \uu(\avar,\symb{3}) \approx \return(n)
& \aaa\nul \leq n < \symb{size}(\avar) \wedge
  \bquant{\forall}{i \in \intint{\nul}{n-\one}}{\symb{select}(\avar,i)
  \neq\nul}\ \wedge\ \ \ \  \\
& \multicolumn{2}{r}{
  \symb{select}(\avar,n) = \nul \wedge
    \nul \leq \symb{2} < \symb{size}(\avar) \wedge
    \bquant{\forall}{j \in \intint{\nul}{\symb{2}}}{\symb{select}(
    \avar,j) \neq \nul}
    ]
  } \\
\end{array}
\]
Note that we grouped together the $\neq \nul$ statements into a
quantification, which looks a lot like the other quantification in
the constraint.  Now, let us generalize!  We will use the generalized
equation (LEN.K): $\uu(\avar,k) \approx \return(n)\ 
\constraint{\varphi}$, where:
\[
\begin{array}{cc}
\varphi: &
  k = m + \one \wedge
  \nul \leq n < \symb{size}(\avar) \wedge
  \bquant{\forall}{i \in \intint{\nul}{n-\one}}{\symb{select}(\avar,i)
  \neq\nul}\ \wedge \\
& \symb{select}(\avar,n) = \nul \wedge
  \nul \leq m < \symb{size}(\avar) \wedge
  \bquant{\forall}{j \in \intint{\nul}{m}}{\symb{select}(\avar,j)
  \neq \nul}
\end{array}
\]
Obviously, (LEN.J) is an instance of (LEN.K); we
use \textsc{Expansion} to obtain:
\[
\begin{array}{cll}
\mbox{(LEN.L)} &
  \symb{error} \approx \return(n) &
  \constraint{\varphi \wedge (k < \nul \vee k \geq \symb{size}(\avar))} \\
\mbox{(LEN.M)} &
  \symb{return}(k) \approx \return(n) &
  \constraint{\varphi \wedge \nul \leq k < \symb{size}(\avar) \wedge
  \symb{select}(\avar,k) = \nul} \\
\mbox{(LEN.N)} &
  \uu(\avar,k+\one) \approx \return(n) &
  \constraint{\varphi \wedge \nul \leq k < \symb{size}(\avar) \wedge
  \symb{select}(\avar,k) \neq \nul} \\
\end{array}
\]
The two $\forall$ statements in $\varphi$, together with
$\symb{select}(\avar,n) = \nul$, imply that $m < n$, so $k \leq n$.
Consequently, (LEN.L) has an unsatisfiable constraint and may be
deleted:
$k < \nul$ cannot hold because $k = m+\one$ and $\nul \leq m$, nor
$k \geq \symb{size}(x)$ because $k \leq n$ and $n < \symb{size}(x)$.

For (LEN.M), the two $\forall$ statements together with
$\symb{select}(\avar,k) =
\nul$ imply that $n - \one < k$, so $n \leq k$.  Thus, $n
= k$.  \textsc{EQ-deletion} gives an equation with an unsatisfiable
constraint, which we remove using \textsc{Deletion}.
As for (LEN.N), we use \textsc{Simplification} with a calculation
and reformulate the constraint to obtain:
\[
\begin{array}{cr}
\mbox{(LEN.O)} &
  \uu(\avar,p) \approx \return(n)\ \ \aaa p = k + \one \wedge
  \symb{select}(\avar,n) = \nul \wedge \nul \leq n <
  \symb{size}(\avar)\ \wedge\ \ \ \\
& \bquant{\forall}{i \in \intint{\nul}{n-\one}}{\symb{select}(\avar,i)
  \neq\nul} \wedge \nul \leq k < \symb{size}(\avar)\ \wedge\ \ \ \\
& \bquant{\forall}{j \in \intint{\nul}{k}}{\symb{select}(\avar,j) \neq
  \nul} \wedge
 \text{some constraints on $m$}] \\
\end{array}
\]
This equation is simplified to an equation of the form $\return(n)
\approx \return(n)\ \constraint{\ldots}$ using the induction rule
obtained from (LEN.K); we complete with \textsc{Deletion}.
\end{example}

\begin{example}\label{exa:sum14}
We consider $\Rules_{\summ}$, the LCTRS
with the two correct implementations of the motivating
\exshort\ref{exa:motivating}; that is, rules $\text{(1a)}$--$\text{(1d)}$
and $\text{(4a)}$--$\text{(4e)}$.
The rules are terminating because in the recursive rule
$\text{(1c)}$, $n-i$ decreases in every step and is bounded from
below by $0$, and in
rule $\text{(4c)}$, the value $k$ decreases against the bound $0$.

To prove equivalence of these implementations when the given
length is within the array bounds, we must show that (ARR.A) is an
inductive theorem:
\[
\mbox{(ARR.A)}\ \ \symb{sum1}(a,k) \approx \symb{sum4}(a,k)\ 
  \constraint{\nul \leq k \leq \symb{size}(a)}
\]
The derivation follows a similar pattern as with factorial: we
first simplify the left hand using rule (1a), then expand on the
right
and use the induction rule,
$\symb{sum4}(a,k) \arrz \symb{u}(a,k,\nul,\nul)\ \constraint{\nul
\leq k \leq \symb{size}(a)}$, to eliminate the remaining occurrence
of $\symb{sum4}$.  This gives:
\[
\begin{array}{c}
  \symb{w}(n,\symb{u}(a,k',\nul,\nul)) \approx
  \symb{u}(a,k,r,\one) \\
  \constraint{k' = k - \one \wedge \nul \leq k' < \symb{size}(a)
  \wedge n = \symb{select}(a,k') \wedge r = \nul + \symb{select}(a,
  \nul)} \\
\end{array}
\]
Continuing to expand and simplify,
we easily remove the equations resulting from
rules (1b) and (1d) in every step,
but the recursive rule (1c) causes a divergence.
\[
\begin{array}{r@{\>}c@{\>}ll}
  \symb{u}(a,k,r_2,\symb{3}) & \approx & \symb{w}(n,\symb{u}(a,k',
    r_1,\symb{2})) &
  \constraint{k' = k - \one \wedge \symb{2} < k \leq \symb{size}(a)
    \wedge r_2 = r_1 + \symb{select}(a,\symb{1}) \wedge \ldots} \\
  \symb{u}(a,k,r_3,\symb{4}) & \approx & \symb{w}(n,\symb{u}(a,k',
    r_2,\symb{3}))\!\! &
  \constraint{k' = k - \one \wedge \symb{3} < k \leq \symb{size}(a)
    \wedge r_3 = r_2 + \symb{select}(a,\symb{2}) \wedge \ldots} \\
  \symb{u}(a,k,r_4,\symb{5}) & \approx & \symb{w}(n,\symb{u}(a,k',
      r_3,\symb{4}))\!\! &
  \constraint{k' = k - \one \wedge \symb{4} < k \leq \symb{size}(a)
    \wedge r_4 = r_3 + \symb{select}(a,\symb{3}) \wedge \ldots} \\
\end{array}
\]
We can easily complete after generalizing any of these equations to:
\[
\begin{array}{cr@{\>}c@{\>}ll}
  \mbox{(ARR.GEN):} &
    \symb{u}(a,k,r,i) & \approx & \symb{w}(n,\symb{u}(a,k',r',i')) &
    [k' = k - \one \wedge \nul \leq i' < k \leq \symb{size}(a)\ \wedge\ \  \\
  & \multicolumn{4}{r}{i' = i - \one \wedge r = r' +
    \symb{select}(a,i') \wedge n = \symb{select}(a,k')]} \\
\end{array}
\]
\end{example}

\begin{example}\label{exa:strcpy_proof}
Recall $\symb{strcpy}$ from \exshort\ref{ex:strcpy}
and the analysis rules and equation from
\exshort\ref{ex:strcpytest}.
The inductive proof follows roughly the same lines as the one for
$\symb{strlen}$ and is found automatically by our tool (see
\secshort\ref{sec:experiments}).
We reach a divergence in equations such as:
\[
\begin{array}{cr}
\bullet &
\symb{test}(\avar,n,\symb{v}(a,\avar,\one)) \approx \strue\ 
  [\nul \leq n < \symb{size}(\avar) \wedge n < \symb{size}(a)
  \wedge \symb{select}(\avar,n) = \nul\ \wedge\  \\
& \bquant{\forall}{i \in \intint{\nul}{n-\one}}{\symb{select}(\avar,
  i) = \nul} \wedge \symb{select}(\avar,\nul) \neq \nul \wedge
  \symb{select}(\avar,\nul) = \symb{select}(a,\nul)] \\
\bullet &
\symb{test}(\avar,n,\symb{v}(b,\avar,\symb{2})) \approx \strue\ 
  [\nul \leq n < \symb{size}(\avar) \wedge n < \symb{size}(b)
  \wedge \symb{select}(\avar,n) = \nul\ \wedge\  \\
& \bquant{\forall}{i \in \intint{\nul}{n-\one}}{\symb{select}(\avar,
  i) = \nul} \wedge \symb{select}(\avar,\nul) \neq \nul \wedge
  \symb{select}(\avar,\nul) = \symb{select}(b,\nul)\ \wedge\ \\
& 
  \symb{select}(\avar,\one) \neq \nul \wedge
  \symb{select}(b,\one) = \symb{select}(\avar,\one)
  ] \\
\bullet &
\symb{test}(\avar,n,\symb{v}(c,\avar,\symb{3})) \approx \strue\ 
  \constraint{\dots \wedge
  \symb{select}(c,\symb{2}) \neq \nul \wedge
  \symb{select}(c,\symb{2}) = \symb{select}(\avar,\symb{2})}
\end{array}
\]
To generalize, we abstract $\one,\symb{2},\symb{3}$ by $k \geq \nul$,
collect similar statements into quantifications and remove the
endpoint.  We quickly complete after this \textsc{Generalization} to:
\[
\begin{array}{r}
\symb{test}(\avar,n,\symb{v}(c,\avar,k)) \approx \strue\ [
\nul \leq n < \symb{size}(\avar) \wedge n < \symb{size}(c)
  \wedge \symb{select}(\avar,n) = \nul \wedge \nul \leq k\ \wedge\ \\
\bquant{\forall}{i \in \intint{\nul}{n-\one}}{\symb{select}(\avar,i)
  \neq \nul} \wedge \bquant{\forall}{i \in \intint{\nul}{k-\one}}{
  \symb{select}(\avar,i) \neq \nul}\ \wedge\ \\
  \bquant{\forall}{i \in \intint{\nul}{k-\one}}{\symb{select}(c,i) =
  \symb{select}(\avar,i)}
  ]
\end{array}
\]
\end{example}

\begin{example}\label{ex:fib}
Let us compare two implementations of the Fibonacci function:
\[
\begin{array}{crcll}
\text{(1)} & \symb{fibrec}(x) & \arrz & \nul & \constraint{x \leq \nul} \\
\text{(2)} & \symb{fibrec}(\one) & \arrz & \one &  \\
\text{(3)} & \symb{fibrec}(x) & \arrz & \symb{plus}(\symb{fibrec}(x-\one),
  \symb{fibrec}(x-\symb{2})) & \constraint{x \geq \symb{2}} \\
\text{(4)} & \symb{plus}(\return(x),\return(y)) & \arrz &
  \return(x+y) & \\
\text{(5)} & \symb{fibiter}(x) & \arrz & \symb{iter}(x,\one,\nul,\one) \\
\text{(6)} & \symb{iter}(x,i,y,z) & \arrz & \symb{iter}(x, i + \one, z,
  y + z) & \constraint{x \geq i} \\
\text{(7)} & \symb{iter}(x,i,y,z) & \arrz & \return(y) & \constraint{x < i} \\
\end{array}
\]
Starting with the equation $\symb{fibrec}(x) \approx
\symb{fibiter}(x)\ \constraint{\strue}$ eventually results in a
divergence:
\[
\begin{array}{rcll}
\symb{iter}(n,\symb{3},\symb{1},\symb{2}) & \approx &
  \symb{plus}(\symb{iter}(m,\symb{iter}(m,\symb{2},\symb{1},\symb{1})),
              \symb{iter}(k,\symb{iter}(k,\symb{1},\symb{0},\symb{1})))\ 
  \constraint{m = n - \one \wedge k = n - \symb{2}} \\
\symb{iter}(n,\symb{4},\symb{2},\symb{3}) & \approx &
  \symb{plus}(\symb{iter}(m,\symb{iter}(m,\symb{3},\symb{1},\symb{2})),
              \symb{iter}(k,\symb{iter}(k,\symb{2},\symb{1},\symb{1})))\ 
  \constraint{m = n - \one \wedge k = n - \symb{2}} \\
\symb{iter}(n,\symb{5},\symb{3},\symb{5}) & \approx &
  \symb{plus}(\symb{iter}(m,\symb{iter}(m,\symb{4},\symb{2},\symb{3})),
              \symb{iter}(k,\symb{iter}(k,\symb{3},\symb{1},\symb{2})))\ 
  \constraint{m = n - \one \wedge k = n - \symb{2}} \\
\end{array}
\]
The proof is easily finished by using the following generalization:
\[
\begin{array}{c}
\symb{iter}(n_3,i_3,z_3,z_4) \approx \symb{plus}(\symb{iter}(n_2,
  i_2,z_2,z_3),\symb{iter}(n_1,i_1,z_1,z_2)) \\
\constraint{n_2 = n_3 - \one \wedge n_1 = n_2 - \symb{2} \wedge
i_3 = i_2 + \one \wedge i_2 = i_1 + \one \wedge z_3 = z_1 + z_2 \wedge
z_4 = z_2 + z_3}
\end{array}
\]
Thus, we can show equivalence of functions with wildly different
time complexities ($\symb{fibrec}$'s running time is exponential in
the input
value, whereas that of $\symb{fibiter}$ is linear).
\end{example}

\begin{example}\label{ex:CR}
Finally, we consider
an example~which \citeN[\secshort{}6, item 2]{god:str:08}
describe as beyond their method. Here two recursive
imperative programs calculating $\sum_{i = 1}^n i$ are compared.
The methods from \secshort\ref{sec:transformations}
yield the following LCTRS.
\[
\begin{array}{crcllcrcll}
(1) & \symb{f}(n) & \arrz & \return(n) & \constraint{n \leq \nul} &
(4) & \symb{g}(n) & \arrz & \return(n) & \constraint{n \leq \one} \\
(2) & \symb{f}(n) & \arrz & \symb{u}(n,\symb{f}(n\!-\!1)) &
  \constraint{n > \nul} &
(5) & \symb{g}(n) & \arrz & \symb{v}(n,\symb{g}(n\!-\!1)) &
  \constraint{n > \one} \\
(3) & \symb{u}(n,\return(m)) & \arrz & \return(n\!+\!m) & &
(6) & \symb{v}(n,\return(m)) & \arrz & \return(n\!+\!m) \\
\end{array}
\]
Starting with the equation $\symb{f}(x) \approx \symb{g}(x)\ 
\constraint{\strue}$ eventually results in a divergence:
\[
\begin{array}{crcl}
\mbox{(CR.A):} &
\symb{u}(x,\symb{u}(y_1,\symb{g}(y_2))) & \approx &
  \symb{v}(x,\symb{u}(z_1,\symb{g}(z_2))) \\
& \multicolumn{3}{c}{
  \constraint{x > \one \wedge y_1 = x - \one \wedge z_1 = x - \one
  \wedge y_2 = y_1 - \one \wedge z_2 = z_1 - \one}} \\
\mbox{(CR.B):} &
  \symb{u}(x,\symb{u}(y_1,\symb{u}(y_2,\symb{g}(y_3)))) & \approx &
  \symb{v}(x,\symb{u}(z_1,\symb{u}(z_2,\symb{g}(z_3)))) \\
  \multicolumn{4}{c}{\ \ \ 
  \constraint{x > \one \wedge y_1 = x - \one \wedge z_1 = x - \one
  \wedge y_2 = y_1 - \one \wedge z_2 = z_1 - \one \wedge
  y_3 = y_2 - \one \wedge z_3 = z_2 - \one}} \\
\mbox{(CR.C):} &
  \phantom{ABCDG}
  \symb{u}(x,\symb{u}(y_1,\symb{u}(y_2,\symb{u}(y_3,\symb{g}(y_4))))) &
  \approx & \symb{v}(x,\symb{u}(z_1,\symb{u}(z_2,\symb{u}(z_3,
  \symb{g}(z_4)))))\ \constraint{\dots} \\
\end{array}
\]
As the constraints imply that each $y_i = z_i$, these
equations can all be generalized to $\symb{u}(x,\symb{u}(y,z))
\approx \symb{v}(x,\symb{u}(y,z))\ \constraint{x > \one}$.  Again, the
proof is quickly completed.
\end{example}

\subsection{Soundness and Completeness of Rewriting Induction}
\label{subsec:ri-correctness-proof}

We now give an intuition on how to prove
\thmshort\ref{thm:ri-correctness}.
The complete proof can be found in \appshort\ref{sec:correctness}.
We follow the proof method
of~\cite{sak:nis:sak:sak:kus:09}, which builds on the original proof
idea in~\cite{red:90}.
This uses the relation $\leftrightarrow_\E$, defined by
\[
\begin{array}{rcll}
\subreplace{C}{s\gamma}{p} & \leftrightarrow_\E & \subreplace{C}{t\gamma}{p}\ &
\text{if}\ s \approx t\ \constraint{\varphi} \in \E\ \text{or}\ t \approx
s\ \constraint{\varphi} \in \E,\ \text{and}\ 
\gamma\ \text{respects}\ \varphi \\
\end{array}
\]
for $\E$ a set of equations.
The proof is split up into several auxiliary lemmas.
To start:

\newcounter{key-lemma-section-Counter}
\newcounter{key-lemma-lemma-Counter}
\setcounter{key-lemma-section-Counter}{\value{section}}
\setcounter{key-lemma-lemma-Counter}{\value{theorem}}
\begin{lemma}\label{lem:theorem:alternative}
All equations in $\E$ are inductive theorems \emph{if and only if}
$\lrarr{\E}\ \subseteq\ \lrarrr{\Rules}$ on ground terms (so
  if $\aterm,\bterm$ are ground and $\aterm \lrarr{\E} \bterm$, then
  also $\aterm \lrarrr{\Rules} \bterm$).
\end{lemma}

This is obvious from the definitions.
The next lemma originates in \cite{sak:nis:sak:sak:kus:09}, which
is adapted from \cite{koi:toy:00} and is key to our method.

\begin{lemma}[\cite{sak:nis:sak:sak:kus:09}]
\label{lem:principle}
Let $\arrz_1$ and $\arrz_2$ be binary relations.
We have $\leftrightarrow^*_1$ $=$ $\leftrightarrow^*_2$ if
  (a)
    $\arrz_1$ $\subseteq$ $\arrz_2$,
  (b)
    $\arrz_2$ is well founded, and
  (c)
    $\arrz_2$ $\subseteq$ $\left(\arrz_1 \cdot \arrz^*_2 \cdot
    \leftrightarrow_1^* \cdot \gets^*_2 \right)$.
\end{lemma}
\begin{proof}
It follows from $\arrz_1$ $\subseteq$ $\arrz_2$ that
$\leftrightarrow^*_1$ $\subseteq$ $\leftrightarrow^*_2$.
To show that $\leftrightarrow^*_2$ $\subseteq$ $\leftrightarrow^*_1$,
we prove $\arrz^*_2$ $\subseteq$ $\leftrightarrow^*_1$ by well-founded
induction on $\arrz_2$.
Since the base case $s$ $\arrz^*_2$ $s$ is clear, we suppose $s$
$\arrz_2$ $t$ $\arrz^*_2$ $u$.
As $\arrz_2$ $\subseteq$ $\left(\arrz_1 \cdot \arrz^*_2 \cdot
\leftrightarrow_1^* \cdot \gets^*_2 \right)$ there
must be some $a,b,c$ such that $s$ $\arrz_1$ $a$ $\arrz_2^*$ $b$
$\leftrightarrow_1^*$ $c$ $\gets^*_2$ $t$.
Since $\arrz_1$ $\subseteq$ $\arrz_2$ (i.e., $s$ $\arrz_2$ $a$), we
can apply the induction hypothesis both on $a$ and on $t$, so $a
\leftrightarrow_1^* b \leftrightarrow_1^* c \leftrightarrow_1^* t$
and $t \leftrightarrow_1^* u$.
Therefore, $s$ $\leftrightarrow^*_1$ $u$.
\end{proof}

We will use Lemma~\ref{lem:principle} with $\arrz_\Rules$ for
$\arrz_1$, and $\arrz_{\Rules \cup \H}$ for $\arrz_2$.
Soundness of the algorithm then
follows if $\leftrightarrow_\E$ is included
in $\leftrightarrow_\H^*$ whenever $(\E,\emptyset,\flag) \vdashristar
(\emptyset,\H,\flag')$.

\thmshort\ref{thm:ri-correctness} is the combination of
\lemshort\ref{lem:theorem:alternative} with
Lemmas~\ref{lem:ri-correctness:sound}
and~\ref{lem:ri-correctness:complete} below.

\begin{lemma}\label{lem:ri-correctness:sound}
If $(\E,\emptyset,\flag) \:\vdashristar\: (\emptyset,\H,
\flag')$, then $\lrarr{\E}\;\subseteq\;\lrarr{\Rules}$
holds on ground terms.
\end{lemma}

\begin{proof}[idea]
Let $\parlr{\E}$ denote a \emph{parallel} application of zero or
more $\leftrightarrow_{\E}$ steps.
We first show that $(\E,\H,\flag) \vdashri (\E',\H',\flag')$ by any
rule other than \textsc{Completeness} implies both (a)
$
 \parlr{\E}
 {\subseteq\:}
 \left(\arrz^*_{\Rules\cup\H'} \cdot
 \parlr{\E'}
 \cdot \gets^*_{\Rules\cup\H'}\right)
$
on ground terms, and (b) if
$
  \arrz_{\Rules \cup \H} {\subseteq\:} 
 (\arrz_{\Rules}\!\cdot\!\arrz^*_{\Rules\cup\H}\!\cdot\!
 \parlr{\E}\!
 \cdot\!\gets^*_{\Rules\cup\H}
 )
$
on ground terms, then
$
 \arrz_{\Rules \cup \H'} {\subseteq\:} 
 (\arrz_{\Rules}\!\cdot\!\arrz^*_{\Rules\cup\H'}\!\cdot\!
 \parlr{\E'}\linebreak
 \cdot\!\gets^*_{\Rules\cup\H'}
 )
$
on ground terms.
We show this by considering how each step alters $\E$ and
$\H$, which we use to see that
$(\E,\H,\flag) \vdashristar (\E',\H',\flag')$
implies (a) and (b), by induction on the total
number of $\vdashrinoblank$-steps in the derivation (counting also the
hidden steps inside \textsc{Completeness}).
Thus, if $(\E,\emptyset,\flag) \vdashristar (\emptyset,\H,\flag')$
then
$\arrz_{\Rules \cup \H}\ \subseteq\ \arr{
\Rules} \cdot \arrz^*_{\Rules \cup \H} \cdot \leftarrow^*_{\Rules \cup \H}$,
so we can apply Lemma~\ref{lem:principle} to conclude that
$\arrz_\Rules$ and $\arrz_{\Rules \cup \H}$ are the same (on ground
terms).
Therefore, and by property (a),
$\leftrightarrow_\E\ \subseteq \parlr{\E} \subseteq\
\arrz^*_{\Rules \cup \H} \cdot
\leftarrow^*_{\Rules \cup \H}\ \subseteq\ 
\leftrightarrow^*_\Rules$.
\end{proof}

\begin{lemma}\label{lem:ri-correctness:complete}
If $\Rules$ is confluent
and $(\E,\emptyset,\complflag) \vdashristar
\bot$, then $\lrarr{\E}\;\not\subseteq\;\lrarr{\Rules}$
holds
on ground terms.
\end{lemma}

\begin{proof}[idea]
By confluence and termination together, we can speak of \emph{the}
normal form $\cterm\!\downarrow_\Rules$ of any term $\cterm$; if
$\cterm$ is ground, then by quasi-reductivity its normal form is a
ground constructor term.
A property of confluence is that if $\dterm \leftrightarrow_\Rules^*
\eterm$, then
$\dterm\!\downarrow_\Rules = \eterm\!\downarrow_\Rules$.
So, it suffices to prove that for some $\aterm \approx \bterm\ 
\constraint{\varphi} \in \E$ there is a ground constructor
substitution $\gamma$ which respects this equation, such that
$\aterm\gamma \neq \bterm\gamma$.
We first note that if $(\E,\H,\complflag) \vdashri \bot$, then
this can only be a \textsc{Disprove} step; in all cases the equation
that causes the disproof has this property.
We also see, by examining the various inference rules, that if
$(\E_1,\H_1,\complflag) \vdashri (\E_2,\H_2,\complflag)$ and
both (a)
$\rightarrow_{\Rules \cup \H_1}\ \subseteq\ \rightarrow_\Rules
\cdot \rightarrow_{\Rules \cup \H_1}^* \cdot \parlr{\E}
\cdot \leftarrow_{\Rules \cup \H_1}^*$ and (b)
$\leftrightarrow_{\E_1} \cup \leftrightarrow_{\H_1}\ \subseteq\ 
\leftrightarrow_\Rules^*$ on ground terms, then also
$\leftrightarrow_{\E_2} \cup \leftrightarrow_{\H_2}\ \subseteq\ 
\leftrightarrow_\Rules^*$ on ground terms.
In a reduction $(\E,\emptyset,\complflag) = (\E_1,\H_1,\flag_1)
\vdashri \cdots \vdashri (\E_n,\H_n,\flag_n) \vdashri \bot$, we may
assume (a) by the observations in the proof
of Lemma~\ref{lem:ri-correctness:sound}, and (b) is inductively
preserved.
As $\leftrightarrow_{\E_n \cup \H_n}$ cannot be included in
$\leftrightarrow_\Rules^*$, therefore
neither can $\leftrightarrow_\E\ =\ \leftrightarrow_{\E_1 \cup \H_1}$.
We complete by Lemma~\ref{lem:theorem:alternative}.
\end{proof}

\section{Generalizing Equations}\label{sec:lemma-gen}

\emph{Divergence},
as encountered in all examples in
\secshort\ref{sec:ri}, is very common in inductive theorem proving:
we often
need a more general claim to obtain a stronger induction hypothesis.
As it is not always easy to find a suitable generalization, the
(automatic) generation of suitable generalizations, and lemma
equations for \textsc{Postulate}, has been extensively investigated
\cite{bun:bas:hut:ire:05,kap:sak:03,kap:sub:96,nak:nis:kus:sak:sak:10,urs:kou:04,wal:96}.

Also for transformed procedural programs, we will certainly need
a large variety of lemma generation techniques to handle most
practical cases.  We start the work by proposing two methods to
generalize equations, specialized to deal with constraints.

\subsection{Generalizing Initializations}\label{subsec:generalise}

Our first and most important technique fundamentally relies on
the constrained setting.  Although it may appear deceptively simple
(at its core, the generalization just drops a part of the constraint),
it is particularly effective for dealing with loops.

\begin{example}\label{exa:generalize}
Let us state the rules of $\Rules_{\fact}$ from
\exshort\ref{ex:running}
in an alternative way: we replace rule (1)
$
\symb{factiter}(x) \arrz \symb{iterm}(x,\one,\one)
$
by (1\Prime):
$
\symb{factiter}(x) \arrz \symb{iter}(x,
  v_1,v_2)\ \constraint{v_1 = \one \wedge v_2 = \one}
$.
That is, the values corresponding to \emph{initializations}
\texttt{int z = 1; int i = 1;} are moved into the constraint.
Evidently, this change does not alter the relation $\arr{\Rules}$.

Now consider what happens if we use the same steps as in
\exshort\ref{ex:running}--\ref{ex:running:postulate}.
The resulting proof has the same shape, but with more complex
equations.  Some instances:
\[
\begin{array}{cl}
  \mbox{(FCT.B\Prime)}: & \symb{iter}(n,v_1,v_2) \approx
    \symb{factrec}(n)\ \ 
      \constraint{n \geq \one \wedge v_1 = \one \wedge v_2 = \one}
  \\
  \mbox{(FCT.D\Prime)}: & \symb{iter}(n,z_1,i_1) \approx
    \symb{factrec}(n)\ \ 
      \constraint{n \geq \one \wedge v_1 = \one \wedge v_2 = \one
        \wedge z_1 = v_1 * v_2 \wedge i_1 = v_2 + \one}
    \\
  \mbox{(FCT.J\Prime)}: & \symb{mul}(n,\symb{iter}(m,z_1,i_1))
    \approx \symb{iter}(n,z_2,i_2)\ 
    [n > \one \wedge v_1 = \one \wedge v_2 = \one \wedge m = n -
    \one\ \wedge \\
  & \multicolumn{1}{r}{
      z_1 = v_1 * v_2 \wedge i_1 = v_2 + \one \wedge
      z_2 = z_1 * i_1 \wedge i_2 = i_1 + \one]
    } \\
  \mbox{(FCT.K\Prime)}: & \symb{mul}(n,\symb{iter}(m,z_2,i_2))
    \approx \symb{iter}(n,z_3,i_3)\ 
      [n > \one \wedge v_1 = \one \wedge v_2 = \one
        \wedge m = n - \one\ \wedge \\
  \multicolumn{2}{r}{
      z_1 = v_1 * v_2 \wedge i_1 = v_2 + \one \wedge
      z_2 = z_1 * i_1 \wedge i_2 = i_1 + \one \wedge
      z_3 = z_2 * i_2 \wedge i_3 = i_2 + \one]} \\
\end{array}
\]
Here the left- and right-hand side of the divergent equations
(FCT.J\Prime) and (FCT.K\Prime) are the same modulo variable
renaming, while the constraint grows.
Essentially, we keep track of parts of the history of an equation in
its constraint.
We generalize (FCT.J\Prime) by dropping all clauses
$v_i = q_i$
where  $v_i$ is an
initialization variable and $q_i$ a value.
We rename
the variables $v_i$ (as they no longer play a special
role) and obtain:
\[
 \begin{array}{lc}
  \mbox{(FCT.M\Prime)}: & \symb{mul}(n,\symb{iter}(m,z_1,i_1))
    \approx \symb{iter}(n,z_2,i_2) \\
  & \constraint{n > \one \wedge m = n - \one \wedge
  z_1 = x_1 * x_2 \wedge i_1 = x_2 + \one \wedge
    z_2 = z_1 * i_1 \wedge i_2 = i_1 + \one} \\
 \end{array}
\]
We can complete the derivation with (FCT.M\Prime) as we did with
(FCT.M) before.
\end{example}

Formally, what we do here is threefold.  First, we alter the set of rules
we work from.

\begin{definition}[Initialization-free Rules]
Given $\Rules$, fix a set $\setivars \subsetneq \setvars$ of variables
not occurring in $\Rules$.  The \emph{initialization-free} counterpart
$\Rules'$ of $\Rules$ is obtained by stepwise replacing any rule
$\ell \to C[f(r_1,\dots,r_i,\dots,r_n)]\ \constraint{\varphi}$ with
$f \in \Defineds$ and
$r_i$ a value by $\ell \to C[f(r_1,\dots,v,\dots,r_n)]\ \constraint{
\varphi \wedge v = r_i}$ for some fresh $v \in \setivars$, until no
such rules remain.
\end{definition}

Then, to apply \textsc{Generalization} to an equation $s \approx t\ 
\constraint{\varphi_1 \wedge \dots \wedge \varphi_n}$
we choose
\[s \approx t\ \constraint{\bigwedge \{\varphi_i \mid 1 \leq
i \leq n \wedge \varphi_i\ \text{does not have the form}\ v = u\ 
\text{with}\ v \in \setivars\ \text{and}\ u \in \Values\}}\]
as the generalized equation
and rename its variables in $\setivars$
to variables in $\setvars$.

Finally, we
restrict the \textsc{Simplification} and \textsc{Expansion} steps to
preserve initialization
constraints
throughout the proof.  The strategy we use in \ctrl---which includes
an approach to handle the
$v \in \setivars$---is described in \secshort\ref{sec:strategy}, but
in particular:
\begin{itemize}
\item When
  we rename rules for use in \textsc{Simplification} or
  \textsc{Expansion}, the renaming must respect membership in
  $\setivars$, i.e., if $x$ is renamed to $y$, then $y \in \setivars$
  iff $x \in \setivars$.
\item In $\sim$-steps,
  any conjuncts $v = n$ are ignored: to simplify
  $s \approx t\ \constraint{\varphi \wedge v_1 = n_1 \wedge \dots
  \wedge v_k = n_k}$, we modify $s \approx t\ \constraint{\varphi}$,
  obtaining $s' \approx t'\ \constraint{\varphi'}$, and continue with
  $s' \approx t'\ \constraint{\varphi \wedge v_1 = n_1 \wedge \dots
  \wedge v_k = n_k}$.  Thus we avoid, e.g., translating
  $\afun(v_i)\ \constraint{v_i = \nul}$ back to $\afun(\nul)\ 
  \constraint{\strue}$.
\end{itemize}

\subsection{Abstracting Equivalent Recursive Calls}\label{subsec:abstract}

Our second generalization technique aims to remove
\emph{recursive} symbols where possible.

\begin{definition}\label{def:recursive}
For symbols $f,g$, let $f \leadsto g$ if there is a rule
$f(\seq{\ell}) \arrz r\ \constraint{\varphi}$ with $g$ a symbol in
$r$.  A symbol $f$ is \emph{recursive} if it is a defined symbol with
$f \leadsto^+ f$.
\end{definition}

The key idea is to identify equivalent occurrences of a recursive
call on both sides of an equation
and to replace them by a variable.
For example, $\symb{g}(x) + \symb{f}(y) \approx \symb{f}(z) +
\symb{g}(x)\ \constraint{y \geq z \land y \leq z}$ is replaced
by $a + b \approx b + a\ \constraint{\strue}$ because
for values $k,n,m$: if $n \geq m \land m \leq n$ holds, then both
$\symb{g}(k)$ and $\symb{g}(k)$,
as well as
$\symb{f}(n)$ and $\symb{f}(m)$, are syntactically equal.

\begin{definition}
\label{def:rec-abs}
A \emph{recursion-abstraction} of
$s \approx t\ 
\constraint{\varphi}$ is any equation of the form $C[x_1,\dots,x_n]
\approx D[x_{i_1},\dots,x_{i_n}]$ such that
(a) $s = C[s_1,\dots,s_n]$ and $t = D[t_{i_1},\dots,t_{i_n}]$ for some
  $\vec{s},\vec{t}$;
(b) $\{i_j \:|\: 1 \leq j \leq n \} = \{ 1,\dots,n\}$;
(c) neither $C$ nor $D$ contain recursive symbols;
(d) each $s_j$ and $t_j$ has a recursive symbol as root symbol;
(e) for $1 \leq i \leq n$ and all ground substitutions $\gamma$ which
  respect $s \approx t\ \constraint{\varphi}$: $s_i\gamma = t_i
  \gamma$;
(f) $x_1,\dots,x_n$ are fresh w.r.t.~$s,t$.
\end{definition}

For a given equation, at most one choice of $C,D$ is possible, and
there are only finitely many permutations $i_1,\dots,i_n$.
Requirement (e) can be
checked
by confirming that an equation $s_j \approx t_j\ 
\constraint{\varphi}$ is removed by the combination of
\textsc{EQ-deletion} and \textsc{Deletion}.

\begin{example}
In \exshort\ref{ex:CR}, we find an abstraction for (CR.A) by choosing
$C = \symb{u}(x,\symb{u}(y_1,\Box))$, $D = \symb{v}(x,\symb{u}(
z_1,\Box))$, $s_1 = \symb{g}(y_2)$ and $t_1 = \symb{g}(z_2)$.
Requirement (e) holds:
if we write $\varphi$ for
the constraint of (CR.A), \textsc{EQ-deletion} on $\symb{g}(y_2)
\approx \symb{g}(z_2)\ \constraint{\varphi}$ produces the
unsatisfiable constraint $\varphi \wedge y_2 \neq z_2$.  Thus,
we generalize the equation to $\symb{u}(x,\symb{u}(y_1,a)) \approx
\symb{v}(x,\symb{u}(z_1,a))\ \constraint{\varphi}$, which is
$\equalgen$-equivalent to the equation
used in \exshort\ref{ex:CR}.
\end{example}

\begin{example}
Given $\symb{g}(x) + \symb{f}(y) \approx \symb{f}(z) +
\symb{g}(x)\ \constraint{y \geq z \land y \leq z}$,
let $C$ and $D$ be $\Box +
\Box$, $s_1 = \symb{g}(x), s_2 = \symb{f}(y), t_1 = \symb{g}(x), t_2 =
\symb{f}(z)$, $i_1 = 2$ and $i_2 = 1$.
We must see
that for all $\gamma$ which respect $y \geq z \wedge y \leq z$:
$\symb{g}(x)\gamma = \symb{g}(x)\gamma$ and
$\symb{f}(y)\gamma = \symb{f}(z)\gamma$.
Both are easily confirmed, so we generalize to $x_1 + x_2 \approx x_2 + x_1\ 
\constraint{y \geq z \land y \leq z} \equalgen a + b \approx b + a\ \constraint{\strue}$
as suggested.
\end{example}

One can see this generalization heuristic as an instance
of the inference rule \textsc{Specialization} by \citeN{aub:79} for
unconstrained explicit induction; restricted to recursive function calls and
combined with \textsc{Substitutivity of Equality} from the same paper.
Here we lift equality from syntactic level to semantic level
in SMT-theories.

\subsection{Discussion}\label{subsec:discussion}

The first method to generalize
equations is strong (\secshort\ref{subsec:generalise}), but only
for equations of a specific form: we can only use the method
if the equation part of the divergence has the same shape every time.
This is the case for $\fact$, because the rule that causes the
divergence has the form $\symb{iter}(\avar_1,\ldots,\avar_n) \arrz
\symb{iter}(r_1,\ldots,r_n)\ \constraint{\varphi}$, preserving its
outer shape.

In general, the method is most likely to be successful
for the analysis of
tail-recursive functions (with accumulators), such as
those obtained from procedural programs.
We can also handle mutually recursive functions, like
$\symb{u}(\avar_1,\ldots,\avar_n) \arrz \symb{w}(r_1,\ldots,r_m)\ 
\constraint{\varphi}$ and $\symb{w}(\bvar_1,\ldots,\bvar_m) \arrz
\symb{u}(q_1,\ldots,q_n)\ \constraint{\psi}$.
It is not suitable for analyzing systems with (only)
non-tail-recursion, however.  Here, the second technique
comes in (\secshort\ref{subsec:abstract}).  Although
we do not claim that this technique is very powerful, it is often
useful to eliminate apparently simple equations.  It is also
straightforward to use in practice.

Note that $\symb{strlen}$ and $\symb{strcpy}$ also have the required
tail-recursive form to successfully use the first generalization
method.
However, here we additionally have to collect multiple clauses into a
quantification before generalizing, as
with equation (LEN.I).

One may wonder if
generalizing initializations loses
too much; e.g., when removing $v_i = \one$,
we also forget that $v_i \geq \nul$.
However, this is usually not an issue: if a rule
is constrained with $v_i \geq \nul$, this clause is
added to the constraint of the equation via
\textsc{Expansion} before we generalize, as in the expansion from (LEN.B).
There \emph{is} a possible issue with losing information
on the relations \emph{between} variables; more
on this in \secshort\ref{subsec:experiments}.

\section{Implementation}
\label{sec:implementation}
\label{sec:experiments}

The method for program verification in this
paper can be broken down into two parts:
\begin{enumerate}
\item\label{imp:transform}
  transforming a procedural program into an LCTRS;
\item\label{imp:prove}
  proving correctness properties on this LCTRS using rewriting
  induction.
\end{enumerate}

\noindent
An initial implementation of part~\ref{imp:transform}, limited to
functions on integers and one-dimensional statically allocated
integer arrays
is available at:
\begin{center}
\url{http://www.trs.css.i.nagoya-u.ac.jp/c2lctrs/}
\end{center}
In future work, we hope to extend this implementation to include the
remaining features discussed in \secshort\ref{sec:transformations}
and \appshort\ref{app:pointers} such as floating points and
explicit pointers.

\medskip
Part~\ref{imp:prove}, the core method on LCTRSs, has been implemented
in our tool \ctrl\ \cite{kop:nis:15}, along with
basic techniques to
verify termination, confluence and quasi-reductivity.
To handle constraints, the tool is coupled both with a small
internal reasoner and the external SMT solver \zeethree\ \cite{z3}.
\zeethree\ is equipped to prove
\emph{unsatisfiability} as well as satisfiability, which is essential
for testing \emph{validity}.

The internal reasoner serves to detect satisfiability or validity
of simple statements quickly, without a call to an SMT solver, and to
preprocess certain kinds of queries which arise often
(e.g., for termination proving by polynomial interpretations, we
preprocess queries with $\exists\forall$-quantifier prefix to
$\exists$-queries).
The reasoner is also used to simplify the constraints of
equations, by for instance combining statements into quantifications
(which is an essential part of the derivations for functions like
$\symb{strlen}$ or $\symb{strcpy}$).

We also
translate our array formulas into the SMT-LIB
array format as discussed in
\secshort\ref{subsec:transition:array},
encoding an array as a function from $\Z$ to $\Z$ with a second
variable for its size.

The latest version of \ctrl\ (tool paper: \cite{kop:nis:15})
can be downloaded at:
\begin{center}
  \url{http://cl-informatik.uibk.ac.at/software/ctrl/}
\end{center}

\subsection{Strategy}
\label{sec:strategy}

Let us discuss the various choices made during a derivation with
rewriting induction.

\subsubsection{What inference rule to apply}
\label{subsubsec:rule_list}
\ctrl\ always selects the
first rule (combination) from:
\begin{enumerate}
\item \textsc{EQ-deletion} (if applicable) immediately followed by
  \textsc{Deletion};
\item \textsc{Disprove}, but without the limitation to
  \complflag\ proof states;
\item \textsc{Constructor};
\item \textsc{Simplification};
\item a limited form of \textsc{Expansion};
\item \textsc{Generalization} using a recursion-abstraction;
\item \textsc{Generalization} of all initialization variables $v_i \in
  \setivars$ at once;
\item the full form of \textsc{Expansion}.
\end{enumerate}

\subsubsection{Generalization and backtracking}
Core to the rewriting induction process is a backtracking
mechanism.
Every proof state $(\E,\H)$ keeps track of all ancestor states on
which \textsc{Generalization} was applied; a state is
$\complflag$ if it has no such ancestors.
The completeness restriction on \textsc{Disprove} is dropped;
however, when \textsc{Disprove} succeeds on an incomplete state, the
prover does not conclude failure, but instead backtracks to the most
recent ancestor and continues without (immediately) generalizing.
Typically, if a \textsc{Generalization} is attempted too soon in the
proof and results in an unsound equation, this can be derived very
quickly, which allows \ctrl\ to conclude failure of the
\textsc{Generalization} step and to move on to the remaining
expansions.

\begin{example}\label{ex:strategy:backtrack}
Following
\exshort\ref{exa:strlen} (but altered with initialization-free
rules), our strategy moves from ($\{$(LEN.A\Prime)$\},\emptyset$)
to ($\{$(LEN.B\Prime)$\},\emptyset$) as before.  But here,
``restricted expansion'' does not apply (as we will see in
\exshort\ref{ex:strategy:expansion}), nor is there a
recursion-abstraction.  So we generalize the initializations,
obtaining:
\[
\left(
\begin{array}{cc}
\left\{
\begin{array}{cc}
\mbox{(BGEN)} &
\uu(\avar,r_0) \approx \return(n) \\
\multicolumn{2}{c}{\phantom{ABC}
  \constraint{\nul \leq n < \symb{size}(\avar) \wedge
  \bquant{\forall}{i \in \intint{\nul}{n-\one}}{\symb{select}(\avar,i)
  \neq \nul} \wedge \symb{select}(\avar,n) = \nul}} \\
\end{array}
\right\},
&
\begin{array}{c}
\emptyset
\end{array}
\end{array}
\right)
\]
We store ($\{$(LEN.B\Prime)$\},\emptyset$) as an ancestor state of
($\{$(BGEN)$\},\emptyset$).  The only option now is \textsc{Expansion}.
Expanding in the left-hand side gives three equations, including:
\[
\begin{array}{c}
\return(r_0) \approx \return(n) \\
  \constraint{\nul \leq n < \symb{size}(\avar) \wedge
  \bquant{\forall}{i \in \intint{\nul}{n-\one}}{\symb{select}(\avar,i)
    \neq \nul} \wedge \symb{select}(\avar,n) = \nul} \\
\end{array}
\]
\textsc{Constructor} gives $r_0 \approx n\ \constraint{\varphi}$,
where $\varphi$ is satisfied by, e.g.,
$[r_0:=\nul,n:=\one,x:=\mathtt{[\one, \nul]}]$; by \textsc{Disprove},
we obtain $\bot$.  However, the state is incomplete as it
has an ancestor stored.
Thus, we backtrack to ($\{$(LEN.B\Prime)$\}, \emptyset$), and
continue with full expansion.
\end{example}

The \textsc{Completeness} rule is implemented via the same mechanism:
if $(\E,\H)$ has a most recent ancestor $(\E',\H')$ with $\E \subseteq
\E'$, then $(\E',\H')$ is dropped from the ancestor list. If a
\textsc{Disprove} succeeds when the list is empty, we conclude failure,
resulting in NO if the system is confluent and MAYBE otherwise.

\begin{example}\label{ex:strategy:completeness}
In \exshort\ref{ex:running:generalize}, we would add
($\{$(FCT.J)$\}$, $\{$(FCT.D$^{-1}$)$\}$) to the list of ancestors
when generalizing (FCT.J) to (FCT.M).  Once (FCT.T) is removed in
\exshort\ref{ex:running:complete}, we are allowed to remove this state
from the list (although since the proof is
finished at that point, it is not really necessary in this example).
\end{example}

Aside from backtracking due to \textsc{Disprove},
there is a second backtracking mechanism:
although
\textsc{Simplification} and \textsc{Expansion} prioritize choices (for
positions and rules) most likely to result in success, sometimes the
first choice does not work out, but the second one does. Thus,
\ctrl\ uses an evaluation limit: when a path has
more than $N$ expansions, it is aborted, and the prover
backtracks to a direct parent.  \ctrl\ starts with $N=2$
and increases this limit if it does not result in a successful
proof or disproof.

\subsubsection{Simplification}
For \textsc{Simplification}, there are
three choices to be made: the position, the rule and how to
instantiate fresh variables in that rule.

For the position, \ctrl\ selects the leftmost, innermost
position where a rule matches.
This prevents a need to reevaluate a term after its subterms change.

For the rule, rules in $\H$ are attempted before rules in
$\Rules$; if a rule leads to a (presumed) divergence, the backtracking
mechanism ensures that the next one is tried.

In some cases---in particular for induction rules---the right-hand
side and perhaps the constraint of a
rule contain variables not occurring in the left-hand side, such as
(FCT.M$^{-1}$) in \exshort\ref{ex:running:postulate}
and (LEN.K) in \exshort\ref{exa:strlen}.
Here, \ctrl\ tries to
instantiate as many variables in the rule by variables in the
equation as possible.
To rewrite an equation $\aterm
\approx \bterm\ \constraint{\varphi_1 \wedge \cdots \wedge \varphi_n}$
at the root of $\aterm$ with a rule $\ell \arrz r\ \constraint{\psi_1
\wedge \cdots \wedge \psi_m}$, we first determine a $\gamma$ such
that $\aterm = \ell\gamma$ and $\gamma(v_i) = v_i$ for all
$v_i \in \setivars$. If any $\psi_i$ has the form $C[\avar,
\bvar_1,\dots,\bvar_k]$ with $\avar \in \domain(\gamma)$ and all
$\bvar_i \notin \domain(\gamma)$, and there is some
$\varphi_j =
C\gamma[\gamma(\avar),\aterm_1,\ldots,\aterm_k]$, then we extend
$\gamma$ with $[\bvar_i:=\aterm_i]$ for all $i$.  This process is
finite and corresponds to the choices for
the equations (FCT.S) and (LEN.O).
Other variables are chosen fresh.

Note: if some rule \emph{can} be applied, but the backtracking
mechanism aborts all attempts, \ctrl\ backtracks to the parent
state rather than continuing with \textsc{Expansion}.  This is
because testing suggests that allowing \textsc{Expansion} to be
applied on terms not in $\Rules$-normal form is generally not
effective and causes an explosive number of states.

\subsubsection{Expansion}
To categorize \textsc{Expansion}s for step (5) and (8) of
\secshort\ref{subsubsec:rule_list}, we analyze \emph{recursion}.
Let $\afun \succsim \bfun$ if $\afun \leadsto^* \bfun$ (following
\defshort\ref{def:recursive}), and let $\afun \succ \bfun$ if $\afun \succsim
\bfun$ and $\bfun \not\succsim \afun$.
Symbols are split into five categories:
\emph{constructors}, \emph{calculation symbols},
\emph{non-recursive defined symbols}, \emph{tail-recursive
symbols}, and
\emph{non-tail-recursive symbols}.
A recursive symbol is \emph{tail-recursive} if its only
defining rules (in $\Rules$) have either the form $f(\ell_1,
\dots,\ell_k) \to \avar\ \constraint{\varphi}$ with $\avar$ a
variable, or the form $f(\ell_1,\dots,\ell_k) \to
g(r_1,\dots,r_m)\ \constraint{\varphi}$ with $f \succ h$
for all $h$ in any $r_i$.
Recursive functions not of this form are
\emph{non-tail-recursive}.

An expansion of $\aterm \simeq \bterm\ \constraint{\varphi}$ at
position $p$ of $\aterm$, with $\subpos{s}{p} = f(\vec{\cterm})$,
is \emph{restricted}---so eligible for step (5)---if (a) $f$ is
non-recursive, or (b) the induction rule $\aterm \arrz \bterm\ 
\constraint{\varphi}$ is admissible and either $f$ is tail-recursive
and $(\FV(s) \cup \FV(t)) \cap \setivars = \emptyset$, or $f$ is
non-tail-recursive.
The induction rule is added only in case (b).
Here, a rule $\rho\colon g(\ell_1,\dots,\ell_k) \arrz
r\ 
\constraint{\varphi}$ is \emph{admissible} if $\Rules \cup \H \cup
\{\rho\}$ is terminating \emph{and} $g \in \Defineds$: we do not add
rules with a constructor or calculation symbol as root symbol $g$, as
this makes it harder to prove termination, which may prevent the
addition of more promising rules later on.

For \emph{unrestricted} expansion, an induction rule is added
when admissible, unless $f$ is tail-recursive.
The unrestricted tail-recursive case concerns rules such as those
got from (FCT.J), (LEN.B), and (LEN.E),
which---testing suggests---are typically not useful. Omitting
them lets \ctrl\ skip many termination checks,
a bottleneck in the process.
Similarly, we do not add induction rules when expanding at
a non-recursive position.

\begin{example}\label{ex:strategy:expansion}
In \exsshort\ref{ex:running}--\ref{ex:running:eqdelete}, the first
expansion occurs in ($\{$(FCT.D$'$)$\},\emptyset,\linebreak
\complflag$), in the right-hand side.  This is not an arbitrary
choice: restricted expansion cannot be used with the
tail-recursive symbol $\symb{iter}$, only
the
non-tail-recursive symbol $\symb{factrec}$.
Then, our strategy closely follows the given derivation.
When we reach ($\{$(FCT.J)$\}$, $\{$(FCT.D$^{-1}$)$\}$,
$\complflag$), restricted expansion is impossible, so we generalize
instead.  After this, an
expansion on the $\symb{iter}$ symbol on either side \emph{is}
restricted.  We can complete the example without backtracking or
using unrestricted \textsc{Expansion}.
\end{example}

For the position to expand at, we follow the same
approach as for \textsc{Simplification}, trying all suitable
positions via the backtracking mechanism.
However, rather than a pure leftmost innermost choice, in the
restricted case (step (5) of \secshort\ref{subsubsec:rule_list}),
we prioritize the more
promising equations by first attempting expansions on a
non-tail-recursive symbol, then those with a non-recursive defined symbol, and
finally those with a tail-recursive one.  In the unrestricted
setting, we follow the leftmost innermost strategy.

Testing shows that this method is very effective for
proving equivalence between a
non-tail-recursive and a
tail-recursive function (as
needed for
equivalence
of a recursive and an iterative C function).
The examples of \secshort\ref{sec:ri} show its effect:
by
eliminating
the
non-tail-recursive functions early on, we are more
likely to arrive at a diverging sequence where all equations have the
same outer shape; e.g., $(\symb{u}(\seq{q}_i) \approx C[\symb{u}(
\seq{v}_i),\symb{u}(\seq{w}_i)] \mid i \in \N)$.
As observed in \secshort\ref{subsec:discussion},
this is ideal for our generalization method.

Following an \textsc{Expansion}, we first process those new
equations in
$\Expdapp{\aterm}{\bterm}{\varphi}{p}$
whose multiset of new symbols is smallest in the recursion
order $\succ$.
Thus, for example in \exshort\ref{ex:running:expand},
after expanding (FCT.D)
we
consider
(FCT.E)---which has new symbols $\{\symb{return},\one\}$---before
(FCT.F)---with  new symbols
$\{\symb{mul},\symb{factrec},\symb{-},\one\}$---since
$\symb{factrec} \succ \symb{return},\one$.
Intuitively, ``smaller'' terms
are ``closer'' to the end of a function,
which allows \textsc{Disprove} to succeed faster
and thus aids the backtracking mechanism.

\subsubsection{Constraint Modification}\label{subsubsec:modif}
Following \textsc{Simplification} and \textsc{Expansion},
\ctrl\ modifies the constraint, as follows.  First,
when a clause $\varphi_i$ in the constraint $\varphi_1
\wedge \dots \wedge \varphi_n$ is implied by the others,
it is removed unless it is a definition clause $v_i = n$.  We
also remove clauses for variables which do not play a role.
Most importantly, \ctrl\ introduces
\emph{ranged quantifications} $\quant{\forall}{
\avar \in \intint{k_1}{k_n}}{\varphi(\avar)}$ whenever possible,
provided $n \geq 3$ (to lessen the effect of coincidence).
Formally, we could describe our approach as follows:

\begin{center}
if $\varphi$ has clauses $C[a], C[b], C[c]$ for some context $C$
and variables $a,b,c$, as well as $b = f(a)$ and $c = f(b)$,
then we may replace the $C$-clauses by $\quant{\forall}{i \in
  \intint{0}{2}}{C[f^i(a)]}$
\end{center}

\noindent
This
is more general than what we use; it
lets us for instance replace $a[i] = 0 \wedge a[j] = 0 \wedge a[k] = 0
\wedge j = i + 2 \wedge k = j + 2$ by $\quant{\forall}{l \in
  \intint{0}{2}}{a[i + 2 \cdot l] = 0}$, for $f = \lambda x.x+2$.
But to represent $f^i$, \ctrl\ must know the relevant theory.
Therefore, we currently
only consider clauses where $b = a + 1$ and $c = b + 1$, and
replace them by $\quant{\forall}{i \in \intint{a}{c}}{C[i]}$.
Since we implement loop counters as integers, this still
captures a large group of constraints.

After $\forall$-introduction, if a boundary of the
range ($0$ and $2$ in the example) is some $v_i \in
\setivars$, we replace it by the value it is defined as,
to avoid generalizing
the starting point of a
quantification.  Thus, e.g.,
$\bquant{\forall}{j \in \intint{v_0}{k}}{\symb{select}(x,j) \neq 0}
\wedge v_0 = 0 \wedge i = v_0 + 1 \wedge k = i + 1$ is replaced by
$\bquant{\forall}{j \in \intint{0}{k}}{\symb{select}(x,j) \neq 0}
\wedge v_0 = 0 \wedge i = v_0 +1 \wedge k = i + 1$.

\subsubsection{Non-Confluence}
Our strategy is admittedly unfair
to non-confluent systems: a successful application
of \textsc{Disprove} is treated as evidence of an unsound equation,
which is not the case without confluence: the non-confluent (LC)TRSs
$\Rules = \{ \symb{f} \to \symb{a}, ~ \symb{f} \to \symb{b}, ~
             \symb{g} \to \symb{a}, ~ \symb{g} \to \symb{b} \}$
along with the inductive theorem
$\symb{f}
\approx
\symb{g}$ highlights that we only have to
prove that two functions \emph{can} produce the same result, not that
they \emph{always} do.

This is deliberate: when proving that two functions produce the
same result, we can see non-confluent LCTRSs as inherently incorrect.
Thus, we focus on confluent systems.  For LCTRSs whose confluence
is unknown, it is preferable to show non-equivalence (which translates
to a MAYBE in the output) over equivalence.

\subsection{Experiments}\label{subsec:experiments}

To assess performance and precision of \ctrl\ empirically,
we tested five assignments from a group of students in the first-year
programming course in Nagoya, all automatically translated to
LCTRSs by \textsf{c2lctrs}:
  \texttt{sum}: given $n$, implement $\sum_{i=1}^n i$;
  \texttt{fib}:
 compute
 the
  $n$th Fibonacci number;
  \texttt{sumfrom}: given $n,m$,
  implement $\sum_{i=n}^m i$;
  \texttt{strlen}
  and
  \texttt{strcpy}.
We compared the first three to LCTRS-versions
of
recursive reference implementations;\footnote{However,
    honesty compels us to mention that
    for \texttt{fib}, we used a manual translation because the one
    obtained from \textsf{c2lctrs} was impractical: where our
    manual translation has a rule
               $\symb{fibrec}(x) \arrz
               \symb{plus}(\symb{fibrec}(x - 1), \symb{fibrec}(x - 2))\ 
               \constraint{x \geq \symb{2}}$,
    the automatic one splits the two recursive calls
    (recall \secshort\ref{subsec:calls}).  Therefore, a more
    sophisticated termination argument is needed, and it is harder
    to eliminate the recursion in the inductive process.
    Handling such cases in the future will likely necessitate an
    additional lemma generation technique.}
for \texttt{strlen}
and \texttt{strcpy} we used a specification as in
\exshort\ref{exa:strlen}
and~\ref{exa:strcpy_proof}.\footnote{Interestingly,
  in \texttt{strcpy02} the student's \texttt{strlen} solution
  is called as a helper function for \texttt{strcpy}.}
We also tested our own
implementations of \texttt{fact} from
\exshort\ref{ex:running} and  \texttt{arrsum} from
\exshort\ref{exa:sum14}, along with 25 function comparisons from the
literature and 12 memory-safety benchmarks from the
\emph{Competition on Software Verification}
\citeA{sv-comp}.
The benchmarks (also from the literature)
are typically fairly small:
the largest, \texttt{lit03\_GS13\_fig6}, has 70 lines of C code
and 55 rewrite rules. We used an Intel i7-5600U CPU at 2.6 GHz under Linux.

We quickly found that many of the student programs had failed to
account for boundary conditions, such as empty strings or negative input.
This causes a NO, or a MAYBE if the system cannot be
proved confluent, so if not all variables are initialized.
To limit the impact of these errors, we did a
second test, where we altered the specification to account for these
mistakes.  The results of both tests are summarized in
Figure~\ref{fig:evaluation}.

\begin{figure}[th]
\begin{tabular}{|c|c|c|c|r|c}
\cline{1-5}
function & YES & NO & MAYBE & \multicolumn{1}{c|}{time} &
  \multirow{11}{*}{\parbox{0.34\textwidth}{\it
  \textbf{Legend:}
  YES indicates that a proof was found, NO a disproof (so a
  conclusion $\bot$); MAYBE denotes that \ctrl\ found no proof or
  disproof, took more than 60 seconds, or failed to prove termination
  of the LCTRS.  The \emph{time} column lists the average run\-time on YES
  and NO results. \\
  }}
  \\[-1pt]
\cline{1-5}
\texttt{sum}        & 9 /  9 & 0 / 0 &  6 /  6 &
  \phantom{0}2.1 / \phantom{0}2.1 \\[-1pt]
\texttt{fib}        & \phantom{0}4 / 10 & 6 / 1 &  3 /  2 & 
  \phantom{0}7.6 / \phantom{0}5.6 \\[-1pt]
\texttt{sumfrom}    & 3 /  3 & 1 / 0 &  2 /  3 &
  \phantom{0}1.8 / \phantom{0}2.1 \\[-1pt]
\texttt{strlen}     & 1 /  2 & 0 / 0 &  5 /  4 &
  \phantom{0}4.1 / \phantom{0}4.0 \\[-1pt]
\texttt{strcpy}     & 3 /  5 & 0 / 0 &  3 /  1 & 21.8 / 17.1 \\[-1pt]
\texttt{arrsum}     & 1 /  1 & 0 / 0 &  0 /  0 &
  \phantom{0}3.9 / \phantom{0}3.9 \\[-1pt]
\texttt{fact}       & 1 /  1 & 0 / 0 &  0 /  0 &
  \phantom{0}2.2 / \phantom{0}2.2 \\
\texttt{literature} & 4 /  5 & 3 / 2 & 18 / 18 &
  \phantom{0}4.0 / \phantom{0}3.9 \\
\texttt{safety}     & 3 /  3 & 2 / 2 &  7 /  7 & 22.3 / 22.3 \\
\cline{1-5}
total              & 29 / 39 & 12 / 5 & 44 / 41 &  \\
\cline{1-5}
\end{tabular}
\caption{Results of \ctrl\ in the initial test (before
the slash), and with obvious mistakes fixed (after the slash).}
\label{fig:evaluation}
\end{figure}

We found five classes of recurring failures.
First, cases where the function was wrong, but \ctrl\ could not answer
NO as it could not prove confluence.
This accounts for six MAYBEs in the initial test
and two in the second, and could be
considered an incorrect implementation.
Second (six failures in either table)
is the termination requirement:
we need termination
independent from the starting symbol, which is often not satisfied
or cannot
be proved by our admittedly limited termination module.

The remaining groups of failures each demonstrate a weakness of
our method.  The third failure occurs when generalization drops
a relation between two variables;
e.g., when $x$ and $y$ are both initialized to $0$
and then increased by $1$ in every loop iteration
(with loops corresponding to tail-recursive functions);
after generalizing, the information that they are equal is lost.
Typically, this manifests as an
\textsc{Expansion} where the non-diverging case can easily be
removed before generalization, but afterwards gives an
equation that can be disproved.
This suggests a natural direction for improvement.

The fourth group are those benchmarks where our primary
generalization technique (\secshort\ref{subsec:generalise}) does not
apply because there are no variables to generalize.  This happens
when both sides have
non-tail-recursive functions
or loops counting
down rather than up.
Recursion-abstraction (\secshort\ref{subsec:abstract}) lets us
solve several benchmarks, but further
lemma generation will be
needed
for the majority.
Nonetheless, this generalization
  technique does allow us to handle
  \exshort\ref{ex:CR}, which can be challenging for existing approaches.

The final group concerns nested loops.  \ctrl's strategy fails because
the counters for the inner and outer loop are generalized at the
same time.  However, inductive proofs with \ctrl's interactive mode
show that such benchmarks
\emph{can} be handled by our method.  Thus, in future work a more
sophisticated generalization strategy would be desirable.

Demonstrative examples of these last three issues are given in
\appshort\ref{app:difficult}.  A full
evaluation page, including exact problem statements, is given at:

\begin{center}
  \url{http://cl-informatik.uibk.ac.at/software/ctrl/tocl/}
\end{center}

\section{Related Work}
\label{sec:related-work}

The related work can be split into two
categories.  First, the
literature on rewriting induction; and second, the work on program
verification and equivalence analysis.

\subsection{Rewriting Induction}
\label{subsec:rewind}
Our inductive theorem proving method
builds on a long literature about rewriting induction
(see e.g.,
\cite{bou:97,fal:kap:12,red:90,sak:nis:sak:sak:kus:09}).
Its core method
extends
existing
techniques to the LCTRS formalism introduced in~\cite{kop:nis:13},
thus generalizing the possibilities of earlier work.

The most relevant related works
are~\cite{fal:kap:12,sak:nis:sak:sak:kus:09},
defining
rewriting induction
for different styles of constrained rewriting.
Both
use only
\emph{integer} functions and predicates; it is not clear how
to generalize
these approaches
to\linebreak
more advanced theories.
The more general setting of LCTRSs enables
rewriting induc\-tion also for systems
with, e.g.,
arrays, bitvectors, or real numbers.
Moreover,
not re-\linebreak stricting the predicates in $\Sigmalogic$
enables
(a limited form of) quantifiers in constraints.

These advantages are enabled by
subtle changes to the
inference rules, in particular \textsc{Simplification} and
\textsc{Expansion}.  Our changes
let us modify
constraints of an equation
and
handle \emph{irregular} rules
with fresh variables in the constraint.
This
additionally
enables \textsc{Expansion} steps
to create such (otherwise infeasible) rules.
The method requires a very different implementation from
previous definitions: we need separate strategies to simplify
constraints (e.g., deriving quantified statements), and,
for the desired generality, must rely primarily on external
solvers to manipulate constraints.

Moreover, we have introduced a
completely new
generalization technique,
as a powerful
tool for analyzing loops in particular.
\citeN{nak:nis:kus:sak:sak:10} use a
similar idea (abstracting the initialization values),
but the execution is very different:
for an equation
$\aterm \approx \bterm\ \constraint{\varphi}$,
first $\aterm \approx \bterm$ is adapted via templates obtained from the
rules, then $\varphi$ is generalized via a set of relations between
positions tracked by the proof process.
In our method, the constraint carries all the information.
We succeed on all examples in~\cite{nak:nis:kus:sak:sak:10},
and on some where their method fails
(cf.~\appshort\ref{sec:examples}; e.g., for non-negative \texttt{n}, a
\texttt{for}-loop summing up from \texttt{1} to \texttt{n}
is compared to \texttt{n*(n+1)/2}).

For \emph{unconstrained} systems, the literature contains
several generalization methods, e.g.,
\cite{kap:sak:03,kap:sub:96,urs:kou:04}.
Mostly, our method in \secshort\ref{subsec:generalise}
is very different from these approaches.
Most similar, perhaps, is~\cite{kap:sak:03},
which also proposes a method to generalize initial values.
As observed by~\citeN{nak:nis:kus:sak:sak:10},
this method is not sufficient for
even our simplest benchmarks $\summ$
and $\fact$, as the argument for the loop variable cannot be generalized;
in contrast, our method has no problem with such variables.
As discussed in \secshort\ref{subsec:abstract}, the
  recursion-abstraction technique presented there essentially lifts a
  technique from explicit induction \cite{aub:79} to constrained
  rewriting induction.

As far as we are aware, there is no other work for lemma generation
of rewrite systems (or functional programs) obtained from procedural
programs.

Like \citeN{gie:kuh:voi:07}, we verify procedural programs
via a transformation to a functional program,
followed by an invocation of an inductive theorem prover.
In an unconstrained setting, they propose an
equivalence-preserving program transformation to a
non-tail-recursive program to eliminate
accumulator arguments.
A combination\linebreak of their approach with ours
could be beneficial e.g.\ for programs with nested loops.

\subsection{Automatic Program Verification and Equivalence Proving}
\label{subsec:related_verification}

Our goal is to
(automatically) verify correctness properties of procedural programs.
Fully automated verifiers for properties like (memory) safety
and termination are regularly assessed at the
\emph{Competition on Software Verification}
\citeA{sv-comp}. However, a comparison with
these tools does not seem useful.
While we can, to some extent, tackle (memory) safety and termination,
our main topic
is \emph{equivalence}, which is not studied in
SV-COMP.
Technically,
equivalence problems can be formulated as
safety problems (by self-composition \cite{bar:dar:rez:11}:
call both programs on equal inputs and
assert that their results are also equal).
However, none of the tools in the ``recursive'' category
of SV-COMP 2015 could prove equivalence for our simplest
(integer) example $\symb{sum}$.

Apart from constrained rewriting, another intermediate
  representation for verification of imperative programs
  is based on (constrained) logic programs or, closely related,
  Horn clauses \cite{alb:gom:hub:pue:07,gup:pop:ryb:11}. It should
  be possible to express our contributions also in this framework,
  provided that constructor terms are supported.

For the
setting of \exshort\ref{exa:motivating}, automated
grading,
\citeN{vuj:nik:tos:kun:13}
apply verification techniques like bounded model checking.
While this enables significant improvements over classic
testing, there is still a non-zero risk of
missing bugs due to under-approximation. Thus, it could
be beneficial to add our approach to the portfolio.

For program equivalence,
we discuss (fully) automated techniques for proving
\emph{partial equivalence} and its special case
\emph{total equivalence}.
Two programs $P_1$ and $P_2$ are \emph{partially equivalent} if
for the same inputs, terminating executions of $P_1$ and $P_2$
return the same value.
They are \emph{totally equivalent} if they
moreover both terminate on all inputs
(see \cite{god:str:08} for a more
extensive
discussion).

This paper addresses total equivalence: we require termination
to analyze partial equivalence. We allow \emph{constrained}
equivalence queries so that only certain inputs are considered.
This includes properties that cannot be checked
programmatically, like the \emph{size} of an array in a C program.
As mentioned in \secshort\ref{sec:strategy}, for non-confluent
programs $P_1$ and $P_2$, we analyze if running $P_1$ on the input
\emph{can} lead to the same result as $P_2$.

\citeN{god:str:08} propose a Hoare-style proof rule for
partial equivalence of recursive programs (among other properties).
To analyze two recursive functions $f_1$ and $f_2$,
these symbols are first replaced in recursive calls in
their bodies by the same \emph{uninterpreted function symbol} $f$.
Under this premise, it is then proved (e.g., by a bounded model
checker) that the bodies of $f_1$ and $f_2$ also have equivalent
results.
In this sense, \citeN{god:str:08} also use inductive reasoning.
However, our approach proves equivalence of
\exshort\ref{ex:CR}
with different recursion base cases, whereas their proof rule
is not applicable.
Moreover,
the use of uninterpreted function symbols requires that the programs
must be deterministic, in contrast to our approach.

\citeN{lop:mon:16} prove partial equivalence for programs
on integers and undefined function symbols
(which may arise also as abstractions
of deterministic complex functions).
They combine self-composition \cite{bar:dar:rez:11},
a safety-preserving transformation of undefined functions to
polynomials (yielding a program on integers only),
recurrence solving for loops, and a standard software model
checker. However, their approach does not support
\emph{mutable} arrays, whose content can be changed during the
program's execution (as in \exshort\ref{exa:strcpy_proof}
for \texttt{strcpy}), in contrast to our method.

\citeN{ver:jan:bru:12} use widening to prove
program equivalence. For validation of compiler optimizations
\cite{nec:00}, they consider programs
with (linear-)affine arithmetic and arrays. A restriction of their
approach is that it does not exploit the \emph{semantics} of
arithmetic operations beyond associativity and commutativity.

Recently, \emph{regression verification} has become an active topic
of research in program equivalence proving
\cite{god:str:13,lah:haw:kaw:reb:12,fel:gre:kle:rum:ulb:14}.
As in regression testing, two programs are compared
that are syntactically \emph{almost} the same, e.g., different
revisions of the same code base with a
refactored function.
Regression verification then
analyzes if the two
programs are semantically equivalent.

\citeN{god:str:13} improve modularity over \cite{god:str:08}
by decomposing the proof obligations into smaller units via
the call graph of the program. \citeN{haw:kaw:lah:reb:13} propose
\emph{mutual summaries}, relating the postconditions of two program functions.
This generalizes uninterpreted functions as summaries and
allows analysis of non-deterministic programs. A challenge is to find
such mutual summaries automatically. \citeN{fel:gre:kle:rum:ulb:14} address
this problem via Horn constraint solving to find \emph{coupling predicates}
over linear arithmetic between program points.
It would be interesting to adapt their approach for lemma generation.
They also analyze \emph{total} equivalence:
a separate termination proof is required.
The web interface of their tool \llreve\ currently fails on
the same example as~\cite{nak:nis:kus:sak:sak:10}
(cf.~\secshort\ref{subsec:rewind}). They mention an extension to arrays
and heap data structures as future work.

\section{Directions for Future Work}
\label{sec:future}
This paper is by no means intended as the end station for inductive
theorem proving on LCTRSs, but rather as the beginning.
The generalization methods we supply are powerful together, but
they do not suffice for more complicated systems or equations.
A mere two methods cannot bypass the need to search for loop
invariants altogether.

A natural extension would thus be both to adapt existing lemma
generation techniques to the constrained setting
and to adapt
techniques for finding loop invariants towards the setting of
rewriting induction, e.g., to suggest suitable lemmas.
It might also be worthwhile to directly look at the constraints
and develop advanced
methods for constraint modification, which could be followed by a
generalization step.
Moreover, our generalization technique from
\secshort\ref{subsec:generalise} could be improved to generalize
not only initializations with constants, but also initializations
with other values, e.g., copies of function parameters. This is
motivated by loops that count down instead of up.
Additionally, inspired by \cite{lop:mon:16}, one might consider
LCTRSs with uninterpreted functions to model functions with unknown
implementations.

For a different direction, we may extend the translation from
\secshort\ref{sec:transformations}, e.g., by translating
structs to term data structures (cf.~\cite{ott:bro:ess:gie:10}).
The ideas from \secshort\ref{sec:transformations}
can also be applied for languages such as Python or Java, enabling
equivalence proofs
between functions in different languages. This could be
particularly interesting for a reference implementation in an
inherently memory-safe language like F\# or Java, and an efficient
implementation in a language like C that
has no such memory safety guarantees.

Finally, we hope to extend the implementation in the future, both to
increase the strength of the inductive theorem proving---adding new
theory and testing for more sophisticated heuristics---and to add more
features to the translation from C code.

\section{Conclusions}
\label{sec:conclusion}

In this paper, we have done two things.  First, we have discussed a
transformation from procedural programs to constrained term rewriting.
By abstracting from the memory model underlying a particular
programming language
and instead encoding concepts like integers and
arrays in an intuitive way, this transformation can be applied to
various different (imperative) programming languages.  The resulting
LCTRS is close to the original program
and has built-in error checking for all mistakes of interest.

Second,
we have extended rewriting induction to the setting of
LCTRSs.  We have shown how this method can be used to prove correctness of
procedural programs.  The LCTRS formalism is a good
analysis backend for this, since the techniques from standard rewriting
can typically be extended to it, and native support for logical constraints
and data types like integers and arrays is present.

We have also introduced two new techniques to generalize
equations.  The idea of the core method is to identify constants
used as \emph{variable initializations}, keep track of them during the
proof process, and abstract from these constants when a proof attempt
diverges.  The LCTRS setting is instrumental in the simplicity of
this method, as it boils down to dropping a (cleverly chosen) part of
a constraint.  The second method recognizes---and
abstracts---recursive calls on semantically equivalent arguments.

In addition to the theory, we provide an implementation of these techniques.
Initial results on a small database of programs from students and
the literature are very promising. In future work, we
aim to increase the strength of our implementations.

\appendixhead{}

\begin{acks}
We are grateful to Stephan Falke, who contributed to an older version
of this work, and for the helpful remarks of the reviewers
for~\cite{kop:nis:14} and for the present paper.
\end{acks}

\newpage
\appendix

\section{Translating C programs to LCTRSs}\label{app:transformations}

This appendix provides further details on the translation from
C programs to LCTRSs.

\subsection{Optimizing LCTRSs}\label{app:optimizations}

After generating the LCTRS, we simplify the (left-linear) result
by the following steps:

\begin{enumerate}
\item \emph{Combining unconstrained rules.}
  Like~\cite{fal:kap:sin:11}, we repeat the following:
  \begin{itemize}
  \item select any unconstrained rule $\rho$ of the form $\symb{u}(x_1,\dots,
    x_n) \to r$ where $\symb{u}$ is not the initial symbol of a
    C function (like $\symb{fact}$ in \exshort\ref{exa:fact1 cnt 1}),
    and $\symb{u}$ neither occurs in $r$ nor in the left-hand side of
    any other rule; the repetition stops if no such $\rho$ exists;
  \item rewrite all right-hand sides with $\rho$;
  \item remove both the rule $\rho$ and the symbol $\symb{u}$.
  \end{itemize}
  This process does not substantially alter the multi-step reduction
  relation $\arrr{}$ as the only symbols removed are those which we
  think of as ``intermediate'' symbols.
\item \emph{Combining constrained rules.}
  If there are distinct rules $\ell \to r\ \constraint{\varphi}$
  and $\ell \to r\ \constraint{\psi}$ (modulo renaming), these are
  combined into $\ell \to r\ \constraint{\varphi \vee \psi}$.  Given
  rules $\ell \to \symb{u}(s_1,\dots,s_m)\ \constraint{\varphi}$ and
  $\symb{u}(x_1,\dots,x_m) \arrz r_i\ \constraint{\psi_i}$ for $1 \leq
  i \leq n$ with all $x_j$ variables, we may replace them by
  $\ell \to r_i[x_1:=s_1,\dots,x_m:=s_m]\ \constraint{\varphi \wedge
  \psi_i[x_1:=s_1,\dots,x_m:=s_m]}$ for $1 \leq i \leq n$, if:
  \begin{itemize}
  \item $\symb{u}$ is not the initial symbol of a function
    and does
    not occur in any other rule, or $\ell$;
  \item the terms $s_j$ do not contain defined symbols (as then
    we might remove a non-terminating subterm, which
    \emph{would} impact the multi-step reduction relation).
  \end{itemize}

\item \emph{Removing unused arguments.}
  For all function symbols and
  all their argument positions, we mark whether the position is
  ``used'':
  \begin{itemize}
  \item all argument(s) of every $\ret{f}$ and initial symbols
    (e.g.~$\symb{fact}$) are used;
  \item for other symbols $\symb{u}_i$ of arity $n$ and every $1 \leq
    j \leq n$: if there is a rule $\symb{u}_i(\ell_1,\dots,\ell_n)
    \to r\ \constraint{\varphi}$ where $\ell_j$ is not a
    variable (which can arise for instance with the transformation in
     \secshort\ref{subsec:calls})
    or occurs in $\varphi$, then argument $j$ is used in
    $\symb{u}_i$;
  \item for all rules $\symb{u}_i(\ell_1,\dots,\ell_n) \to r\ 
    \constraint{\varphi}$ and $1 \leq j \leq n$: argument $j$ is used
    in $\symb{u}_i$ if $\ell_j$ is a variable occurring at a
    \emph{used position} in $r$; here, a position $p$ is used in $s$
    if either $p = \epsilon$ or $p = i \cdot p',\ s = f(\vec{t})$,
    argument $i$ is used in $f$ and position $p'$ is used in $t_i$.
  \end{itemize}
  The last, recursive, step essentially calculates a fixpoint;
  in summary, an argument position is \emph{used} if it is
  possible to reduce to a term where we actually need the subterm
  at that position as part of a constraint or the function's return
  value.  When a variable is not used in any later statement, we will
  avoid carrying it along.
\item \emph{Simplifying constraints.} Constraints may be brought into
  an equivalent form, e.g., by removing duplicate clauses or
  replacing, e.g., $\neg(x > y)$ by $x \leq y$.  
  Here, $\varphi$ is ``equivalent'' to $\psi$ in a rule
  $\ell \arrz r\ \constraint{\varphi}$ if
  $\quant{\forall}{\vec{x}}{\bquant{\exists}{\vec{y}}{\varphi}
  \leftrightarrow \bquant{\exists}{\vec{z}}{\psi}}$ holds, where
  $\FV(\ell) \cup \FV(r) = \{\vec{x}\}$,
  $\FV(\varphi) \setminus \{\vec{x}\} = \{\vec{y}\}$, and
  $\FV(\psi) \setminus \{\vec{x}\} = \{\vec{z}\}$ (much like
  the observation on $\equalgen$ below
  \defshort\ref{def:equivalent-constrained-terms}).
  We typically only remove negations and unused variables.
\end{enumerate}

\begin{example}\label{exa:declaration}
As an example, let us consider the simplification of a toy function.

\begin{minipage}[h]{0.33\textwidth}
\begin{verbatim}
int f(int x) {
  int y,z;
  if (x < 0) return 0;
  z = 0;
  while (x > 0) {
    x--;
    z += x;
  }
  y = z + x;
  return y;
}
\end{verbatim}
\end{minipage}
\vline
\begin{minipage}[h]{0.62\textwidth}
\[
\begin{array}{rcll}
\symb{f}(x) & \arrz & \symb{u}_1(x,y,z) \\
\symb{u}_1(x,y,z) & \arrz & \symb{u}_2(x,y,z) &
  \constraint{x < 0} \\
\symb{u}_1(x,y,z) & \arrz & \symb{u}_3(x,y,z) &
  \constraint{\neg(x < 0)} \\
\symb{u}_2(x,y,z) & \arrz & \ret{\symb{f}}(\nul) \\
\symb{u}_3(x,y,z) & \arrz & \symb{u}_4(x,y,\nul) \\
\symb{u}_4(x,y,z) & \arrz & \symb{u}_5(x,y,z) &
  \constraint{x > \nul} \\
\symb{u}_4(x,y,z) & \arrz & \symb{u}_7(x,y,z) &
  \constraint{\neg (x > \nul)} \\
\symb{u}_5(x,y,z) & \arrz & \symb{u}_6(x-\one,y,z) \\
\symb{u}_6(x,y,z) & \arrz & \symb{u}_4(x,y,z+x) \\
\symb{u}_7(x,y,z) & \arrz & \symb{u}_8(x,z + x,z) \\
\symb{u}_8(x,y,z) & \arrz & \ret{\symb{f}}(y) \\
\end{array}
\]
\end{minipage}

The
rule $\symb{f}(x) \arrz \symb{u}_1(x,y,z)$ has
unconstrained variables $y$ and $z$
in the right-hand side which do not occur on the left.
A step with this rule instantiates $y$ and $z$ by arbitrary
type-correct values.  This reflects that in the
C program the variables \texttt{y} and \texttt{z} are at first not
initialized and may contain
an arbitrary value (depending on the compiler).
In the simplified version, this
does not occur; consider the remainder
obtained from combining rules:

\begin{minipage}[h]{0.54\textwidth}
\vspace{-3pt}
\[
\begin{array}{rcll}
\symb{f}(x) & \arrz & \ret{\symb{f}}(\nul) & \constraint{x < \nul} \\
\symb{f}(x) & \arrz & \symb{u}_4(x,y,\nul) &
  \constraint{\neg(x < \nul)} \\
\symb{u}_4(x,y,z) & \arrz & \symb{u}_4(x-\one,y,z+x-\one) &
  \constraint{x > \nul} \\
\symb{u}_4(x,y,z) & \arrz & \ret{\symb{f}}(z+x) &
  \constraint{\neg (x > \nul)} \\
\end{array}
\]
\vspace{-3pt} 
\end{minipage}
\begin{minipage}[h]{0.42\textwidth}
\vspace{4pt}
Now, the first and third arguments of $\symb{u}_4$ are used (in the
constraint and return value), but the second is not: it is merely
passed along in the\linebreak
\vspace{-12pt}
\end{minipage}
\noindent
recursive call.  Removing this variable and simplifying the
constraints, we obtain:

\begin{minipage}[h]{0.48\textwidth}
\[
\begin{array}{rcll}
\symb{f}(x) & \arrz & \ret{\symb{f}}(\nul) & \constraint{x < \nul} \\
\symb{f}(x) & \arrz & \symb{u}_4(x,\nul) & \constraint{x \geq \nul} \\
\symb{u}_4(x,z) & \arrz & \symb{u}_4(x-\one,z+x-\one) &
  \constraint{x > \nul} \\
\symb{u}_4(x,z) & \arrz & \ret{\symb{f}}(z+x) &
  \constraint{x \leq \nul} \\
\end{array}
\]
\end{minipage}
\begin{minipage}[h]{0.485\textwidth}
\vspace{4pt}
This system is \emph{orthogonal} in the sense of~\cite{kop:nis:13}
and thus \emph{confluent}, which is beneficial for
analysis.  The original LCTRS was also confluent, but this was harder
to see.
\end{minipage}

\end{example}

Correctness relies on the fact that the LCTRSs created using the
transformation described in \secshort\ref{sec:transformations} are
``well behaved''; most importantly, all rules are left-linear.

\subsection{Translating C Programs with Explicit Pointers}\label{app:pointers}

As observed at the end of \secshort\ref{subsec:transition:array}, the
simple translation explored there has both up- and downsides.  On the
one hand, by abstracting from the memory model, we can simplify
analysis.  On the other hand, there are certain programs
we cannot handle.

For
C programs with dynamically allocated arrays and/or
explicit pointer use, we consider the memory model
from the C standard.  Declaring or allocating an array selects
an amount of currently unused space in memory and designates it for
use by the given array.
The allocated space is \emph{not} guaranteed to be
at a given position in memory relative to existing
declarations; when an array is indexed out of its declared bounds, the
resulting behavior is undefined---so this can safely be considered an
error (see paragraph 6.5.6:9 in:
   \url{http://www.open-std.org/jtc1/sc22/wg14/www/docs/n1570.pdf}).

We will think of a program's memory as a set of \emph{blocks}, each
block corresponding to a sequence of values.  A pointer then becomes
a location in such a block.  In an LCTRS we will model this using
a ``global memory'' variable, which lists the blocks as a
sequence of arrays; a pointer is a pair of integers, selecting a
memory block and its offset.

Limiting interest to programs on (dynamically allocated) integer or char arrays,
we will use a memory variable of sort $\symb{array}(\symb{array}(
\Int))$, which represents a sequence of integer arrays
(i.e.~$(\Z^*)^*$);
the default value $0_{\symb{array}(\Int)}$ is the
empty sequence $\langle\rangle \in \Z^*$.
We use a theory signature with the array symbols introduced in
\secshort\ref{subsec:transition:array}, along with:
\begin{itemize}
\item $\symb{allocate} : [\symb{array}(\symb{array}(\Int)) \times
  \symb{array}(\Int)] \arrtype \symb{array}(\symb{array}(\Int))$,
  where $\J_{\symb{allocate}}(\langle a_0,\ldots,a_k \rangle, b)
  = \langle a_0,\ldots,a_k,b \rangle$; that is, $\symb{allocate}(
  \mathit{mem},\mathit{arr})$ adds the new sequence $\mathit{arr}$ to
  the memory;
\item $\symb{free} : [\symb{array}(\symb{array}(\Int)) \times
  \Int] \arrtype \symb{array}(\symb{array}(\Int))$, where
  $\J_{\symb{free}}(\langle a_0,\ldots,a_k \rangle, n) =
  \langle a_0,\dots,\linebreak
  a_{n-1},\langle\rangle,a_{n+1},\dots,a_k\rangle$
  if $0 \leq n \leq k$ and $\langle a_0,\ldots,a_k \rangle$ otherwise;
  that is, the memory block indexed by $n$ is considered empty, and
  any further attempt to address a location in that memory block
  should be considered an error.
\end{itemize}
A pointer is represented by a pair $(b,o)$ of a block index and an
offset within that block.  The NULL-pointer is represented by
$(-1,0)$.

\pagebreak

\begin{example}\label{ex:difficultarray}
Consider the following example C++ function:
\begin{verbatim}
int *create(int k) {
  int *a = new int[k];
  int *b = a + 1;
  for (int i = 0; i < k; i += 2) b[i] = 42;
  return a;
}
\end{verbatim}
Now, \texttt{a} and \texttt{b} share memory, and new memory is
allocated.  We might encode this as:
\[
\begin{array}{rclr}
\symb{create}(mem,k) & \arrz &
  \uu(\symb{allocate}(mem,\avar),k,\symb{size}(mem),\nul) &
  \constraint{\symb{size}(\avar) = k} \\
\uu(mem,k,ai,ao) & \arrz & \symb{v}(mem,k,ai,ao,ai,ao + \one,\nul) \\
\symb{v}(mem,k,ai,ao,bi,bo,i) & \arrz &
  \symb{w}(mem,k,ai,ao,bi,bo,i) & \constraint{i < k} \\
\symb{v}(mem,k,ai,ao,bi,bo,i) & \arrz & \return(mem,ai,ao) &
  \constraint{i \geq k} \\
\symb{w}(mem,k,ai,ao,bi,bo,i) & \arrz &
  \multicolumn{2}{l}{\symb{error}\ \ \ \ \ \ \ \ \ \ \ \ \ \ \ \ \ 
  \constraint{bo + i < \nul \vee bo + i \geq
  \symb{size}(\symb{select}(mem,bi))}} \\
\symb{w}(mem,k,ai,ao,bi,bo,i) & \arrz &
  \multicolumn{2}{l}{\symb{v}(\symb{store}(mem,bi,
  \symb{store}(\symb{select}(mem,bi),bo+i,\symb{42})),k,} \\
& &
  \multicolumn{2}{l}{\ \ \ ai,ao,bi,bo,i+\symb{2})\ \ \ \ \ \ 
  \constraint{\nul \leq bo + i < \symb{size}(\symb{select}(mem,bi))}} \\
\end{array}
\]
(For clarity, we omit the optimization step that combines the first
two rules, and the one that combines the third with the last two.)

Consider how this example is executed, starting from empty memory.
We will use $\langle\cdot\rangle$ to refer to specific
arrays of type $\symb{array}(\symb{array}(\Int))$ and $[\cdot]$ for
arrays of type $\symb{array}(\Int)$.
\begin{enumerate}
\item We call $\symb{create}(\langle\rangle,\symb{2})$, representing a
  function call when no arrays have been allocated.
\item By the first rule, we get
  $\symb{u}(\symb{allocate}(\langle\rangle,x),\symb{2},
  \symb{size}(\langle\rangle),\nul)$,
  where $x$ is a \emph{random} array.
  All we know is that
  it has size 2---this rule uses irregularity to represent the
  randomness involved in an allocation.  Thus,
  assume the sequence $[\symb{-4},\symb{9}]$ is chosen.  Using calculation steps to evaluate
  $\symb{allocate}$ and $\symb{size}$, we
  get $\symb{u}(\langle[
  \symb{-4},\symb{9}]\rangle,\symb{2},\nul,\nul)$.  Here, the pair
  $(0,0)$ represents the array $a$: the first block in memory, read
  from the start (offset $0$).
\item Then by the second rule (and a calculation), we reduce to
  $\symb{v}(\langle[\symb{-4},\symb{9}]\rangle,\symb{2},\nul,\nul,
  \nul,\one,\nul)$.
  The new pair $(0,1)$ represents $b$: the same memory block as $a$,
  but with offset $1$.  This location points to the sequence
  $[\symb{9}]$.  The final $\symb{0}$ is the index for the loop counter $i$.
\item Entering the loop (as indeed $\symb{0} < \symb{2}$), we reduce to
  $\symb{w}(\langle[\symb{-4},\symb{9}]\rangle,\symb{2},\nul,\nul,
  \nul,\one,\nul)$.
\item\label{it:pointer:test}
  Here, we do an array store: \texttt{b[i] = 42;}.  The LCTRS
  first tests whether \texttt{b[i]} corresponds to a position in
  allocated memory
  and reduces to an error state if not.  This is
  done by selecting the corresponding block from $\mathit{mem}$, then
  testing whether the offset for \texttt{b} and \texttt{i} together
  exceed the block's bounds.  We succeed, as $\symb{0} \leq \symb{0}+\symb{1} <
  \symb{size}(\symb{select}(\langle[\symb{-4},\symb{9}]\rangle
  ,\symb{0}))
  \Leftrightarrow \symb{0} \leq \symb{1} < \symb{size}([\symb{-4},\symb{9}])
  \Leftrightarrow \symb{0} \leq \symb{1} < \symb{2}$.
\item Thus, the update is done: we reduce to: \\
  $\symb{v}(\symb{store}(\langle[\symb{-4},\symb{9}]\rangle,\nul,
  \symb{store}(\symb{select}(\langle[\symb{-4},\symb{9}]\rangle
  ,\symb{0}),
  \one+\nul,\symb{42})),\symb{2},\nul,\nul,\nul,\one,\nul+\symb{2})$ \\
  $\arrzcalc^*
  \symb{v}(\symb{store}(\langle[\symb{-4},\symb{9}]\rangle,\nul,
  \symb{store}([\symb{-4},\symb{9}],\one,\symb{42})),\symb{2},\nul,
  \nul,\nul,\one,\symb{2})$ \\
  $\arrzcalc^*
  \symb{v}(\symb{store}(\langle[\symb{-4},\symb{9}]\rangle,\nul,
  [\symb{-4},\symb{42}]),\symb{2},\nul,\nul,\nul,\one,\symb{2})$ \\
  $\arrzcalc
  \symb{v}(\langle[\symb{-4},\symb{42}]\rangle,\symb{2},\nul,\nul,
  \nul,\one,\symb{2})$ \\
So we retrieve the space for \texttt{b} from memory
  (getting
  the full block $[\symb{-4},\symb{9}]$), update the position
  corresponding to \texttt{b[0]} (which is the same as \texttt{a[1]}),
  get $[\symb{-4},\symb{42}]$, and store the result into the
  corresponding position in memory. Then we carry on with $i+\symb{2}$.
\item Since $\symb{2} \geq \symb{2}$, we reduce to $\return(\langle[\symb{-4},
  \symb{42}]\rangle,\nul,\nul)$, returning the dynamic array
  $[\symb{-4},\symb{42}]$.
\end{enumerate}
Note that in step~\ref{it:pointer:test}, we do not
test whether \texttt{b} corresponds to currently
allocated memory.  This is safe because, if \texttt{b} is the
NULL-pointer or corresponds to previously freed memory,
then $\symb{select}(\mathit{mem},bi)$
is
$\langle\rangle$, and any indexing in this array will cause an error
regardless.
Note also that this function gives a non-error result only for
\emph{even} $k$.
\end{example}

While \exshort\ref{ex:difficultarray} considers only integer arrays,
we could also handle programs with dynamically allocated arrays of varying types.
In this case, we would simply use \emph{multiple} memory variables
with different type declarations.

\section{Correctness proof}
\label{sec:correctness}

In this appendix, we give the full correctness proof, which was
only sketched in \secshort\ref{subsec:ri-correctness-proof}.

\setcounter{tmp-section-Counter}{\value{section}}
\setcounter{tmp-lemma-Counter}{\value{theorem}}
\setcounter{section}{\value{key-lemma-section-Counter}}
\setcounter{theorem}{\value{key-lemma-lemma-Counter}}
\def\thesection{\arabic{section}}

First, we prove \lemshort\ref{lem:theorem:alternative}, reformulated
as follows:

\begin{lemma}
The following statements are equivalent:
\begin{itemize}
\item all equations in $\E$ are inductive theorems;
\item $\lrarr{\E}\ \subseteq\ \lrarrr{\Rules}$ on ground terms (so
  if $\aterm,\bterm$ are ground and $\aterm \lrarr{\E} \bterm$, then
  also $\aterm \lrarrr{\Rules} \bterm$).
\end{itemize}
\end{lemma}

\begin{proof}
Suppose $\lrarr{\E}\ \subseteq\ \lrarrr{\Rules}$ on ground terms.
If $\aterm \approx \bterm\ \constraint{\varphi} \in \E$ and
the ground
constructor substitution $\gamma$ respects this equation, then
$\aterm\gamma$ and $\bterm\gamma$ are ground (since, by
definition of ``respects'' (\defshort\ref{def:constrained_equation}),
$\FV(\aterm) \cup \FV(\bterm) \subseteq \domain(\gamma)$).  Since
obviously $\aterm\gamma \lrarr{\E} \bterm\gamma$ (with empty $C$),
by assumption $\aterm\gamma \lrarrr{\Rules} \bterm\gamma$.  Thus,
$\aterm \approx \bterm\ \constraint{\varphi}$ is an inductive
theorem.

Suppose that all equations in $\E$ are inductive theorems, and
$\cterm \lrarr{\E} \dterm$ for ground $\cterm,\dterm$; we must see
that $\cterm \lrarrr{\Rules} \bterm$.  We have $\cterm = C[\aterm
\gamma]$ and
$\dterm = C[\bterm\gamma]$
for some $\aterm
\approx \bterm\ \constraint{\varphi} \in \E$ and substitution $\gamma$
that respects $\varphi$ and maps all variables in $\aterm,\bterm$ to
ground terms.  Let $\delta$ be a substitution such that each
$\delta(\avar)$ is a normal form of $\gamma(\avar)$; by termination
of $\Rules$, such a $\delta$ exists, and by quasi-reductivity, it is
a ground constructor substitution.  As values cannot be reduced, also
$\delta$ respects $\varphi$.  Therefore $\aterm\delta \lrarr{\E} \bterm\delta$,
which implies $\aterm\delta \lrarrr{\Rules} \bterm\delta$.  We
conclude: $C[\aterm\gamma] \lrarrr{\Rules} C[\aterm\delta]
\lrarrr{\Rules} C[\bterm\delta] \lrarrr{\Rules} C[\bterm\gamma]$,
giving the desired result.
\qed
\end{proof}

Recall also the following key lemma (whose proof has been given in the
main text):

\begin{lemma}[\cite{sak:nis:sak:sak:kus:09}]
Let $\arrz_1$ and $\arrz_2$ be binary relations over some set $A$.
Then, $\leftrightarrow^*_1$ $=$ $\leftrightarrow^*_2$ if all of the
following hold:
\begin{itemize}
  \item $\arrz_1$ $\subseteq$ $\arrz_2$,
  \item $\arrz_2$ is well founded, and
  \item $\arrz_2$ $\subseteq$ $\left(\arrz_1 \cdot \arrz^*_2 \cdot
    \leftrightarrow_1^* \cdot \gets^*_2 \right)$. 
\end{itemize}
\end{lemma}
\setcounter{section}{\value{tmp-section-Counter}}
\setcounter{theorem}{\value{tmp-lemma-Counter}}
\def\thesection{\Alph{section}}

\lemshort\ref{lem:ri-correctness:sound} in the main text is the
combination of the following
Lemmas~\ref{lem:expansion}--\ref{lem:properties_of_ri-step}.

\begin{lemma}
\label{lem:expansion}
Let $s,t$ be terms, $\varphi$ a constraint and $p$ a position of $s$
such that $\subpos{\aterm}{p}$ has the form $\afun(\aterm_1,\ldots,
\aterm_n)$ with $\afun$ a defined symbol and all $\aterm_i$
constructor terms.
Suppose that
the variables in $\aterm,\bterm,\varphi$ are distinct
from those in $\Rules$.  Then:
\begin{enumerate}
\item For any ground constructor substitution $\gamma$
  which respects
  $\aterm \approx \bterm\ \constraint{\varphi}$ and any choice
  of $\Expdapp{\aterm}{\bterm}{\varphi}{p}$, we have:
  \vspace{-2pt}
  \[
   \aterm\gamma ~\left(\arrz_{\Rules,p} \cdot
   \leftrightarrow_{\Expdapp{\aterm}{\bterm}{\varphi}{p}}\right)~ t\gamma
  \]
  Here, $\arrz_{\Rules,p}$ indicates a reduction at position $p$ with
  a rule in $\Rules \cup \Rulescalc$.
\item For any $\aterm' \arrz \bterm'\ \constraint{\varphi'}$ in any
  choice of $\Expdapp{\aterm}{\bterm}{\varphi}{p}$ and any ground
  constructor substitution $\delta$
  which respects
  $\aterm' \approx \bterm'\ \constraint{\varphi'}$, we have:
  \[
   s'\delta ~\left(\leftarrow_{\Rules,p} \cdot
   \leftrightarrow_{\{s ~\approx ~t\ \constraint{\varphi}\}}\right)~
   t'\delta
  \]
\end{enumerate}
\end{lemma}

\begin{proof}
$\subpos{\aterm\gamma}{p} = \subpos{\aterm}{p}\gamma = \afun(
\aterm_1\gamma,\ldots,\aterm_n\gamma)$, where all $\aterm_i\gamma$
are ground constructor terms.  Since $\afun$ is defined,
$\afun(\seq{\aterm}\gamma)$ reduces by quasi-reductivity, which can
only be a root reduction.
Thus, $\aterm\gamma = \subreplace{(\aterm\gamma)}{\ell\delta}{p}$ for
some rule $\ell \arrz r\ \constraint{\psi}$ and substitution $\delta$
which respects $\psi$.  Since
the rule variables are distinct from the ones in the equation,
we can assume that $\delta$ is an extension of $\gamma$, so $\aterm\gamma
= \subreplace{\aterm}{\ell}{p}\delta$.
Clearly, both $\varphi\delta$ and $\psi\delta$ evaluate to $\top$, and
$\delta(\avar)$ is a value for all $\avar \in \FV(\varphi) \cup \FV(\psi)$.

As $\delta$ unifies $\subpos{\aterm}{p}$ and $\ell$, there is a
most general unifier $\eta$, so $\subpos{\aterm}{p}\eta =
\ell\eta$ and we can write $\delta = \delta' \circ \eta$ for
some $\delta'$.  As $\delta(\avar)$ is a value for all $\avar \in
\FV(\varphi) \cup \FV(\psi)$, $\eta(\avar)$ can only be the same value,
or a variable.  Now, by definition of constrained term reduction,
any choice of $\Expdapp{\aterm}{\bterm}{\varphi}{p}$ has an element
$\aterm' \approx \bterm'\ \constraint{\varphi'}$ where we can write
(for suitable $\cterm,\eta'$ etc.):
\[
\begin{array}{rrcll}
 & \subreplace{\aterm\eta}{\ell\eta}{p} & \approx &
  \bterm\eta & \constraint{\varphi\eta \wedge
  \psi\eta} \\
\equalgen & \subreplace{\cterm}{\ell\eta'}{p} & \approx & \bterm'' &
  \constraint{\varphi''} \\
\arr{\ell \arrz r \constraint{\psi}\mathbf{,} 1 \cdot p} &
\subreplace{\cterm}{r\eta'}{p} & \approx & \bterm'' &
  \constraint{\varphi''} \\
\equalgen & \aterm' & \approx & \bterm' & \constraint{\varphi'} \\
\end{array}
\]
Consider the ``term'' $\aterm\delta \approx \bterm\delta$.  This is
an instance of the first constrained term in this reduction, so by
\thmshort\ref{thm:constrainedterm}, this ``term'' reduces at position
$1 \cdot p$ to $\aterm'\delta'' \approx \bterm'\delta''$ for some
substitution $\delta''$ which respects $\varphi'$.  As the reduction
happens inside $\aterm\delta$, we see that $\bterm\delta =
\bterm'\delta''$.
Thus, $\aterm\gamma = \aterm\delta \arr{\Rules} \aterm'\delta''
\leftrightarrow_{\Expdapp{\aterm}{\bterm}{\varphi}{p}} \bterm'\delta'' =
\bterm\delta = \bterm\gamma$.

As for the second part, note that by definition of $\Expd$ there are
a substitution $\gamma$ and constraint $\psi$ such that the
constrained term $\coterm{\aterm\gamma \approx \bterm\gamma}{\varphi
\gamma \wedge \psi\gamma}$ reduces to $\coterm{\aterm' \approx
\bterm'}{\varphi'}$ at position $1 \cdot p$.
By \thmshort\ref{thm:constrainedterm:reverse}, we find a
substitution $\eta$ which respects $\varphi\gamma \wedge \psi
\gamma$, such that $\aterm\gamma\eta \approx \bterm\gamma\eta
\arr{\Rules} \aterm'\delta \approx \bterm'\delta$ at position
$1 \cdot p$.  Since the reduction takes place in the left part of
$\approx$, we have $\bterm\gamma\eta = \bterm'\delta$ and
$\aterm\gamma\eta \arr{\Rules} \aterm'\delta$.  We are done if
also $\aterm\gamma\eta \leftrightarrow_{\aterm \approx \bterm\ 
\constraint{\varphi}} \bterm\gamma\eta$, which indeed holds because
$\eta \circ \gamma$ respects $\varphi$ (as
$(\varphi\gamma \wedge \psi\gamma)\eta$ implies $\varphi\gamma\eta$).
\qed
\end{proof}

\begin{lemma}
\label{lem:simplification}
Suppose that $(\E,\H,\flag)$ $\vdashri$ $(\E',\H',\flag')$ by any
inference rule other than \textsc{Completeness}.
Then,
\[
 \parlr{\E}
 {\subseteq}
 \left(\arrz^*_{\Rules\cup\H'} \cdot
 \parlr{\E'}
 \cdot \gets^*_{\Rules\cup\H'}\right)
\]
on ground terms. 
\end{lemma}

Here, $\parlr{\E'}$ denotes a \emph{parallel} application
of zero or more $\leftrightarrow_{\E'}$ steps.

\begin{proof}
It suffices to show that $\leftrightarrow_{\E\setminus\E'}$ $\subseteq$
$\left(\arrz^*_{\Rules\cup\H'} \cdot
\parlr{\E'}
\cdot \gets^*_{\Rules\cup\H'}\right)$ on ground terms: if
$C[u_1,\dots,u_n] \parlr{\E} C[v_1,\dots,v_n]$ because each $u_i
\leftrightarrow_{\rho_i} v_i$ for some $\rho_i \in \E$, then this
gives $u_i \arrz^*_{\Rules\cup\H'} u_i' \parlr{\E'} v_i' \gets^*_{
\Rules\cup\H'} v_i$ if $\rho_i \notin \E'$ and $u_i = u_i'
\parlr{\E'} v_i' = v_i$ if $\rho_i \in \E'$, so (sequentializing
parallel steps) $C[u_1,\dots,u_n] \arrz^*_{\Rules\cup\H'}
C[u_1',\dots,u_n'] \parlr{\E'} C[v_1',\dots,v_n']
\gets^*_{\Rules\cup\H'} C[v_1,\dots,v_n]$ as desired.
For all inference rules (except \textsc{Completeness}) either
$\E\setminus\E' = \emptyset$ or we can write $\E\setminus\E'$ $=$
$\{ s \simeq t\ \constraint{\varphi}\}$.
Consider which inference rule is applied for $\vdashri$.
\begin{itemize}
 \item (\textsc{Simplification}).
       Suppose that $s \simeq t\ \constraint{\varphi}$ is replaced by $u
       \approx t\ \constraint{\psi}$ where $\coterm{s \approx t}{\varphi}$
       $\arrz_{\Rules \cup \H}$ $\coterm{u \approx t}{\psi}$. 
       Let $C[s\gamma]$ $\leftrightarrow_{\{s\simeq t\ \constraint{
       \varphi}\}}$ $C[t\gamma]$, where $\gamma$
       is a substitution which respects $\varphi$.
       It follows from \thmshort\ref{thm:constrainedterm} that
       $s\gamma \approx t\gamma$ $\arr{\Rules \cup \H}$
       $u\delta \approx t\delta$ where $\delta$ is a
       substitution which respects $\psi$, and thus, as $\approx$ is a
       constructor,
       $C[s\gamma]$ $\arr{\Rules \cup \H}$ $C[u\delta]$ and
       $t\gamma = t\delta$.
       Then, $C[u\delta]$ $\leftrightarrow_{\{u \approx t\
       \constraint{\psi}\}}$ $C[t\delta] = C[t\gamma]$, and we have
       $C[s\gamma] \arr{\Rules \cup \H} \cdot \leftrightarrow_{\E'}
       C[t\gamma]$.  Symmetrically, if $C[\bterm\gamma]
       \leftrightarrow_{\{s\simeq t\ \constraint{\varphi}\}}
       C[\aterm\gamma]$, then
       $C[t\gamma] \leftrightarrow_{\E'} \cdot \gets_{\Rules \cup \H}
       C[s\gamma]$.
       Thus, $\leftrightarrow_{s \simeq t\ \constraint{\varphi}}$
       $\subseteq$ $(\to^*_{\Rules \cup \H}\cdot
       \leftrightarrow_{\E'} \cdot \gets^*_{\Rules \cup \H})$.
       This suffices because in this case $\H = \H'$.

\item (\textsc{Deletion}).
       In the case that $s$ $=$ $t$, the relation
       $\leftrightarrow_{\E\setminus\E'}$ is the identity.
       Otherwise, $\varphi$ is unsatisfiable, so $s \simeq t\ 
       \constraint{\varphi}$ is never used, i.e.,
       $\leftrightarrow_{\E\setminus\E'}$ $=$ $\emptyset$.

\item (\textsc{Expansion}).
       Suppose $C[s\gamma] \leftrightarrow_{\aterm \simeq \bterm\ 
       \constraint{\varphi}} C[t\gamma]$, where $\gamma$ respects
       $\aterm \simeq \bterm\ \constraint{\varphi}$; as we only
       consider ground terms, $\gamma(\avar)$ is ground for all
       variables in its domain.
       Noting that by quasi-reductivity and termination every
       ground term
       reduces to a ground constructor term, let $\delta$ be a
       substitution where for each $\avar \in
       \domain(\gamma)$, $\delta(\avar)$
       is a constructor term such that $\gamma(\avar) \arrr{\Rules}
       \delta(\avar)$.
       Then it follows from \lemshort\ref{lem:expansion} that
       $C[s\gamma] \arrr{\Rules} C[s\delta]\ (\arrz_{\Rules} \cdot
       \leftrightarrow_{\E'})\ C[t\delta] \leftarrow_\Rules^*
       C[t\gamma]$.
       The situation where $C[\bterm\gamma] \leftrightarrow_{\aterm
       \simeq \bterm\ \constraint{\varphi}} C[\aterm\gamma]$ is
       symmetric.

\item (\textsc{EQ-deletion}).
       Let $s$ $=$ $C[s_1,\ldots,s_n]$ and $t$ $=$ $C[t_1,\ldots,t_n]$ 
       where $s_1,t_1,\ldots,$ $s_n,t_n \in \Terms(\Sigmalogic,\FV(\varphi))$.
       Any ground substitution $\gamma$ which
       respects $\varphi$, and whose domain contains all variables in
       the terms $\aterm_i$ and $\bterm_i$, must map these variables to
       values.  Therefore,
       $\aterm_i\gamma \arrzcalc^* v_i$ and $\bterm_i\gamma
       \arrzcalc^* w_i$, where $v_i$ is the value of $\aterm_i\gamma$
       and $w_i$ is the value of $\bterm_i\gamma$.
       Now, suppose $q \leftrightarrow_{\aterm \simeq \bterm
       \constraint{\varphi}} u$ for ground $q,u$.
       Then (a) $q = D[C[\aterm_1,\ldots,\aterm_n]]\gamma$ and $u =
       D[C[\bterm_1,\ldots,\bterm_n]]\gamma$ for some ground
       $\gamma$ which respects $\varphi$, or (b)
       $u = D[C[\seq{\bterm}]]\gamma$ and $q = D[C[
       \seq{\aterm}]]\gamma$.  In case (a),
       $q \arrr{\Rules} D\gamma[C\gamma[v_1,\ldots,v_n]]$ and
       $u \arrr{\Rules} D\gamma[C\gamma[w_1,\ldots,w_m]]$.
       If each $v_i = w_i$, then clearly $q \rightarrow^*_\Rules
       \cdot \leftarrow^*_\Rules u$.
       Otherwise, $(\varphi \wedge \neg (\aterm_1 = \bterm_1 \wedge
       \cdots \wedge \aterm_n = \bterm_n))\gamma$ is valid, so
       we easily get the desired
       $q \arrr{\Rules} \cdot \lrarr{\E'} \cdot \leftarrow^*_\Rules u$.
       Case (b) is symmetric.

\item (\textsc{Disprove})
  In this case we do not have $(\E,\H,b) \vdashri (\E',\H',b')$.

\item (\textsc{Constructor})
  Let $\aterm = \afun(\aterm_1,\ldots,\aterm_n),\ \bterm =
  \afun(\bterm_1,\ldots,\bterm_n)$, and suppose $C[\aterm\gamma]
  \leftrightarrow_{\{\aterm \approx \bterm\ \constraint{\varphi}\}}
  C[\afun(\bterm\gamma)]$, where $\gamma$ is a substitution which
  respects $\varphi$.  Since $\E'$ contains all equations
  $\aterm_i \approx \bterm_i\ \constraint{\varphi}$, we have
  $C[\aterm\gamma] = C[\afun(\aterm_1\gamma,\ldots,\aterm_n\gamma)]
  \parlr{\E'}
  C[\afun(\bterm_1\gamma,\ldots,\bterm_n\gamma)] =
  C[\bterm\gamma]$.

\item (\textsc{Postulate})
  $\E \setminus \E' = \emptyset$, so there is nothing to prove!

\item (\textsc{Generalization})
  Suppose that $\aterm \approx \bterm\ \constraint{\varphi}$ is
  replaced by $\aterm' \approx \bterm'\ \constraint{\psi}$.
  Suppose that $C[\aterm\gamma] \leftrightarrow_{\{\aterm \approx
  \bterm\ \constraint{\varphi}\}} C[\bterm\gamma]$ for some
  substitution $\gamma$ which respects $\varphi$.  Then there
  exists a substitution $\delta$ which respects $\psi$ such that
  $C[\aterm\gamma] = C[\aterm'\delta] \leftrightarrow_{\{\aterm'
  \approx \bterm' \constraint{\varphi}} C[\bterm'\delta] =
  C[\bterm\gamma]$.
\qed
\end{itemize}
\end{proof}

\begin{lemma}
\label{lem:H-rules}
 Suppose that
 $(\E,\H,\flag)$ $\vdashri$
 $(\E',\H',\flag')$
 by any inference rule other than \textsc{Completeness}.
If
\[
 \arrz_{\Rules \cup \H} {\subseteq} 
 \left(\arrz_{\Rules} \cdot \arrz^*_{\Rules\cup\H} \cdot
 \parlr{\E}
 \cdot \gets^*_{\Rules\cup\H}\right)
\]
 on ground terms, then
\[
 \arrz_{\Rules \cup \H'} {\subseteq} 
 \left(\arrz_{\Rules} \cdot \arrz^*_{\Rules\cup\H'} \cdot
 \parlr{\E'}
 \cdot \gets^*_{\Rules\cup\H'}\right)
\]
 on ground terms. 
\end{lemma}

\begin{proof}
 It suffices to consider the case that \textsc{Expansion} is applied (for the
 other cases, we use \lemshort\ref{lem:simplification}).
 Suppose that $s$ $\arrz_{\H' \setminus \H}$ $t$.
 Using that, by quasi-reductivity and termination, every
 ground term reduces to a ground constructor term,
it follows from \lemshort\ref{lem:expansion} that there
exist ground constructor terms $s',t'$ such that $s$
$\rightarrow^*_\Rules$ $s'$ $\left(\arrz_{\Rules} \cdot
 \leftrightarrow_{\E'}\right)$ $t'$ $\leftarrow^*_\Rules$ $t$, and
 hence:
 \[
 s\ \left(\arrz_{\Rules} \cdot \arrz^*_{\Rules\cup\H'} \cdot
 \parlr{\E'}
 \cdot \gets^*_{\Rules\cup\H'}\right)\ t
 \]
\qed
\end{proof}

\begin{lemma}
\label{lem:properties_of_ri-step}
 Suppose that
$(\E,\H,\flag)$ $\vdashri \cdots \vdashri$ $(\E',\H',\flag')$. 
 Then:
 \begin{enumerate}
  \item $\parlr{\E}$
        $\subseteq$
        $\left(\arrz^*_{\Rules\cup\H'} \cdot
        \parlr{\E'}
        \cdot \gets^*_{\Rules\cup\H'}\right)$
        on ground terms,
  \item if 
        $\arrz_{\Rules \cup \H}$ $\subseteq$
        $\left(\arrz_{\Rules} \cdot \arrz^*_{\Rules\cup\H} \cdot
        \parlr{\E}
        \cdot \gets^*_{\Rules\cup\H}\right)$
        on ground terms, then
        \[
        \arrz_{\Rules \cup \H'} {\subseteq} 
        \left(\arrz_{\Rules} \cdot \arrz^*_{\Rules\cup\H'} \cdot
        \parlr{\E'}
        \cdot \gets^*_{\Rules\cup\H'}\right)
        \]
        on ground terms, and
  \item if $\Rules\cup\H$ is terminating, then so is
        $\Rules\cup\H'$.
 \end{enumerate}
\end{lemma}

\begin{proof}\ 
In the following, we will consider relations limited to
ground terms only.
We prove the statements by induction on the number of
$\vdashrinoblank$-steps,
where steps in the premise of a \textsc{Completeness} step are also counted.
The base case is evident, so suppose $(\E,\H,\flag) \vdashri (\E_1,\H_1,
\flag_1) \vdashristar (\E',\H',\flag')$.
\begin{enumerate}
 \item If the first step uses inference rule \textsc{Completeness}, then
       $(\E,\H,\flag) \vdashristar (\E_1,\H_1,\incomplflag)$ in fewer
       steps, so by the induction hypothesis:
       \[
       \parlr{\E}
       {\subseteq}
       \left(\arrz^*_{\Rules\cup\H_1} \cdot
       \parlr{\E_1}
       \cdot \gets^*_{\Rules\cup\H_1}\right)
       \]
       If the first step uses another inference rule, this same property
       follows from \lemshort\ref{lem:simplification}.  By the induction
       hypothesis we have
       \[
       \parlr{\E_1}
       {\subseteq}
       \left(\arrz^*_{\Rules\cup\H'} \cdot
       \parlr{\E'}
       \cdot \gets^*_{\Rules\cup\H'}\right)
       \]
       It follows from $\H_1$ $\subseteq$ $\H'$ that
       \[
       \parlr{\E}
       {\subseteq}
       \left(\arrz^*_{\Rules\cup\H'} \cdot
       \parlr{\E_1}
       \cdot \gets^*_{\Rules\cup\H'}\right)
       \]
       By replacing
       $\parlr{\E_1}$
       with
       $\left(\arrz^*_{\Rules\cup\H'} \cdot
       \parlr{\E'}
       \cdot \gets^*_{\Rules\cup\H'}\right)$,
       we
       thus obtain
       \[
       \parlr{\E}
       \subseteq
       \left(\arrz^*_{\Rules\cup\H'} \cdot
       \arrz^*_{\Rules\cup\H'} \cdot
       \parlr{\E'}
       \cdot \gets^*_{\Rules\cup\H'} \cdot
       \cdot \gets^*_{\Rules\cup\H'}\right)
       \]
 \item Assume
       $
       \arrz_{\Rules \cup \H} {\subseteq} 
       \left(\arrz_{\Rules} \cdot \arrz^*_{\Rules\cup\H} \cdot
       \parlr{\E}
       \cdot \gets^*_{\Rules\cup\H}\right)
       $.
       By the induction hypothesis (in case of \textsc{Completeness})
       or \lemshort\ref{lem:H-rules} (otherwise),
       \[
       \arrz_{\Rules \cup \H_1} {\subseteq} 
       \left(\arrz_{\Rules} \cdot \arrz^*_{\Rules\cup\H_1} \cdot
       \parlr{\E_1}
       \cdot \gets^*_{\Rules\cup\H_1}\right)
       \]
       We complete by the induction hypothesis on
       $(\E_1,\H_1,\flag_1) \vdashristar (\E',\H',\flag')$.

 \item Trivial with the induction hypothesis, with the first step 
       using either the induction hypothesis again (in case of
       \textsc{Completeness}), the definition of \textsc{Expansion},
       or the observation that other inference rules do not alter
       $\H$.
\qed
\end{enumerate}
\end{proof}

\noindent
Thus we obtain \lemshort\ref{lem:ri-correctness:sound} or,
equivalently, \lemshort\ref{lem:ri-correctness:part1},
as the
first part of \thmshort\ref{thm:ri-correctness}.

\begin{lemma}\label{lem:ri-correctness:part1}
If $(\E,\emptyset,\flag) \vdashri \cdots \vdashri (\emptyset,\H,
\flag')$, then every equation in $\E$ is an inductive theorem of
$\Rules$.
\end{lemma}

\begin{proof}
It is clear that $\arrz_{\Rules}$ $\subseteq$ $\arrz_{\Rules\cup\H}$.
It follows from \lemshort\ref{lem:properties_of_ri-step} that:
\begin{itemize}
 \item $\leftrightarrow_\E$ $\subseteq$
       $\parlr{\E}$ $\subseteq$ $\arrz^*_{\Rules\cup\H} \cdot
       \gets^*_{\Rules\cup\H}$ on ground terms, 
 \item $\arrz_{\Rules \cup \H}$ $\subseteq$ $\arrz_\Rules \cdot
       \arrz^*_{\Rules\cup\H} \cdot \gets^*_{\Rules\cup\H}$ on ground terms, and
 \item $\Rules\cup\H$ is terminating.
\end{itemize}
By \lemshort\ref{lem:principle} (as equality is included in
$\leftrightarrow_\Rules^*$) we find that $\leftrightarrow^*_\Rules$ $=$
$\leftrightarrow^*_{\Rules\cup\H}$,
and hence $\leftrightarrow_\E$ $\subseteq$ $\leftrightarrow^*_\Rules$,
on ground terms.
We complete with Lemma~\ref{lem:theorem:alternative}.
\qed
\end{proof}

Moving on to \emph{disproving}, we need two auxiliary lemmas:

\begin{lemma}\label{lem:completeness:final}
If $\Rules$ is confluent and $(\E,\H,\complflag) \vdashri \bot$, then
$\E$ contains an equation $\aterm \approx \bterm\ 
\constraint{\varphi}$ which is not an inductive theorem.
\end{lemma}

\begin{proof}
By confluence and termination together, we can speak of \emph{the}
normal form $\cterm\!\downarrow_\Rules$ of any term $\cterm$; if
$\cterm$ is ground, then by quasi-reductivity its normal form is a
ground constructor term.
A property of confluence is that if $\dterm \leftrightarrow_\Rules^*
\eterm$, then
$\dterm\!\downarrow_\Rules = \eterm\!\downarrow_\Rules$.
So, it suffices to prove that for some $\aterm \approx \bterm\ 
\constraint{\varphi} \in \E$ there is a ground constructor
substitution $\gamma$ which respects this equation, such that
$\aterm\gamma$ and $\bterm\gamma$ have distinct normal forms.

The only inference rule that could be used to obtain $(\E,\H,
\complflag) \vdashri \bot$ is \textsc{Disprove}, so $\E = \E' \cup
\{ \aterm \simeq \bterm\ \constraint{\varphi} \}$ and one of the
following holds:
\begin{enumerate}
\item $\aterm,\bterm \in \Terms(\Sigmalogic,\setvars)$ with
  $\varphi \wedge \aterm \neq \bterm$ satisfiable.  That is, there
  is a substitution $\gamma$ mapping all variables in the equation to
  values, such that $\varphi\gamma$ is valid and $\aterm\gamma$ and
  $\bterm\gamma$ reduce to different values by $\arrzcalc$.
  We are done since all values are normal forms.
\item $\aterm = \afun(\aterm_1,\ldots,\aterm_n)$ and $\bterm =
  \bfun(\bterm_1,\ldots,\bterm_{m})$ with $\afun$ and $\bfun$ different
  constructors, and $\varphi$ is satisfiable, so there is a
  substitution $\delta$ mapping all variables in $\varphi$ to values
  such that $\varphi\delta$ is valid.  Let $\gamma$ be an extension
  of $\delta$ which additionally maps all other variables in $\aterm,
  \bterm$ to ground terms (by assumption, ground instances of all
  variables exist).  Then $\varphi\gamma$ is still valid, and
  $\aterm\gamma$ and $\bterm\gamma$ are ground terms with
  $\aterm\gamma \rightarrow_\Rules^* (\aterm\gamma)\!\!
  \downarrow_\Rules = f((\vec{\aterm}\gamma)\!\!\downarrow_\Rules)
  \neq g((\vec{\bterm}\gamma)\!\!\downarrow_\Rules) =
  (\bterm\gamma)\!\!\downarrow_\Rules \leftarrow_\Rules^* \bterm
  \gamma$.
\item $\aterm : \asort$ is a variable not occurring in $\varphi$,\ 
  $\varphi$ is satisfiable, there are at least two different
  constructors $\afun,\bfun$ with output sort $\asort$ and either
  $\bterm$ is a variable distinct from $\aterm$ or $\bterm$ has a
  constructor symbol at the root.
  By satisfiability of $\varphi$, a substitution $\delta$ exists
  whose domain does not contain $\aterm$, with $\varphi\delta$ valid.
  If $\bterm$ is a variable, let $\gamma$ be an extension of $\delta$
  mapping $\aterm$ to some ground term rooted by $\afun$ and $\bterm$
  to a ground term rooted by $\bfun$ (by assumption ground instances
  always exist).
  If $\bterm = \afun(\seq{\bterm})$, then let $\gamma$ be an extension
  of $\delta$ mapping $\aterm$ to some ground term rooted by $\bfun$
  and mapping all other variables in $\bterm$ to ground terms as well.
  Either way, $\varphi\gamma$ is valid and $
  (\aterm\gamma)\!\downarrow_\Rules\ \neq\ 
  (\bterm\gamma)\!\downarrow_\Rules
  $.
\qed
\end{enumerate}
\end{proof}

\begin{lemma}\label{lem:completeness:induction}
Suppose that $\rightarrow_{\Rules \cup \H}$ is terminating and
that
$\rightarrow_{\Rules \cup \H} \ \subseteq\ \rightarrow_\Rules \cdot
\rightarrow_{\Rules \cup \H}^* \cdot
\parlr{\E}
\cdot \leftarrow_{\Rules \cup \H}^*$.  If, moreover,
$\Rules$ is confluent, $(\E,\H,\complflag) \vdashri (\E',\H',
\complflag)$, and $\leftrightarrow_{\E} \cup \leftrightarrow_{\H}\ 
\subseteq\ \leftrightarrow_\Rules^*$ on ground terms, then
$\leftrightarrow_{\E'} \cup \leftrightarrow_{\H'}\ \subseteq\ 
\leftrightarrow_\Rules^*$ on ground terms.
\end{lemma}

\begin{proof}
Assume that all conditions are satisfied; we consider the inference
rule used to derive $(\E,\H,\complflag) \vdashri (\E',\H',
\complflag)$.

First, suppose the rule used was \textsc{Completeness}, so $(\E,\H,
\complflag) \vdashristar (\E',\H',\incomplflag)$ and $\E' \subseteq
\E$.  As we have assumed that $\leftrightarrow_\E \cup
\leftrightarrow_\H\ \subseteq\ \leftrightarrow_\Rules^*$, certainly
$\leftrightarrow_{\E'}\ \subseteq\ \leftrightarrow_\E\ \subseteq\ 
\leftrightarrow_\Rules^*$.  As for $\leftrightarrow_{\H'}$,
Lemma~\ref{lem:properties_of_ri-step} gives us that
$\rightarrow_{\Rules \cup \H'}\ \subseteq\ \rightarrow_\Rules \cdot
\rightarrow_{\Rules \cup \H'}^* \cdot
\parlr{\E'}
\cdot \leftarrow_{\Rules \cup \H'}^*$, so (using again that
$\leftrightarrow_{\E'}\ \subseteq\ \leftrightarrow_\Rules$ and
that $\parlr{\E'}\ \subseteq\ \leftrightarrow_{\E'}^*$)
we can
apply Lemma~\ref{lem:principle} and termination of $\rightarrow_{
\Rules \cup \H'}$ to obtain $\leftrightarrow_{\H'}\ \subseteq\ 
\leftrightarrow_{\Rules \cup \H'}\ \subseteq\ 
\leftrightarrow_{\Rules}^*$.

If a different rule was applied, then each
element in $\H'$ either also belongs to $\H$ or (in the case of
\textsc{Expansion}) corresponds to an equation in $\E$.  Thus,
$\leftrightarrow_{\H'}\ \subseteq\ \leftrightarrow_{\E \cup \H}
\subseteq \leftrightarrow_\Rules^*$.
So let $\aterm \approx \bterm\ \constraint{\varphi} \in \E' \setminus
\E$; we must see that $\leftrightarrow_{\{\aterm \approx \bterm\ 
\constraint{\varphi}\}}\ \subseteq\ \lrarrr{\Rules}$ on ground
terms.  By
Lemma~\ref{lem:theorem:alternative}, it suffices if for all ground
constructor substitutions $\gamma$ which respect this equation,
$\aterm\gamma \lrarrr{\Rules} \bterm\gamma$.
We fix $\gamma$ and use a case analysis on the applied inference
rule.

\begin{itemize}
\item (\textsc{Simplification}).
  There is $\aterm' \simeq \bterm'\ \constraint{\varphi'} \in
  \E$ such that $\aterm' \approx \bterm'\ \constraint{\varphi'}
  \arr{\Rules \cup \H} \aterm \approx \bterm\ \constraint{\varphi}$
  at position $1 \cdot p$.
  By \thmshort\ref{thm:constrainedterm:reverse}, we can find
  $\delta$ which respects $\varphi'$ such that
  $\aterm'\delta \arr{\Rules \cup \H} \aterm\gamma$ at position $p$
  and $\bterm'\delta = \bterm\gamma$.  As $\arr{\Rules} \cup
  \arr{\H}\ \subseteq\ \leftrightarrow_\Rules^*$ by the
  assumption, $\aterm\gamma \leftrightarrow_\Rules^*
  \aterm'\delta \leftrightarrow_\E \bterm'\delta = \bterm\gamma$,
  which suffices because $\lrarr{\E}\ \subseteq\ \lrarrr{\Rules}$.
\item (\textsc{Deletion}).
  No equations are added in this case.
\item (\textsc{Expansion}).
  There is $\aterm' \simeq \bterm'\ \constraint{\varphi'} \in \E$
  such that $\aterm \approx \bterm\ \constraint{\varphi} \in
  \Expdapp{\aterm'}{\bterm'}{\varphi'}{p}$ for some $p$.
  By \lemshort\ref{lem:expansion}(2), we have $\aterm\gamma\ 
  (\leftarrow_\Rules \cdot \leftrightarrow_\E)\ 
  \bterm\gamma$, which suffices because $\leftrightarrow_\E\ 
  \subseteq\ \leftrightarrow_\Rules^*$.
\item (\textsc{EQ-deletion})
  $\aterm \simeq \bterm\ \constraint{\varphi'} \in \E$, where
  $\varphi = \varphi' \wedge \neg (\aterm_1 = \bterm_1 \wedge
  \cdots \wedge \aterm_n = \bterm_n)$, and $\aterm = C[\aterm_1,
  \ldots,\aterm_n],\ \bterm = C[\bterm_1,\ldots,\bterm_n]$ for
  some $C,\seq{\aterm},\seq{\bterm}$.  Since any substitution
  which respects $\varphi$ also respects $\varphi'$, we must have
  $\aterm\gamma \lrarr{\E} \bterm\gamma$, so $\aterm\gamma
  \lrarrr{\Rules} \bterm\gamma$.
\item (\textsc{Disprove})
  A reduction with this rule does not have the required form.
\item (\textsc{Constructor}) There is $\afun(\ldots,\aterm,\ldots)
  \approx \afun(\ldots,\bterm,\ldots)\ \constraint{\varphi} \in \E$,
  and by assumption $\afun(\ldots,\aterm,\ldots)\gamma
  \leftrightarrow^*_\Rules \afun(\ldots,\bterm,\ldots)\gamma$.  By
  confluence, this means that $\afun(\ldots,\aterm\gamma,\ldots)\!
  \downarrow_\Rules = \afun(\ldots,\bterm\gamma,\ldots)\!\downarrow_{
  \Rules}$, which implies that $(\aterm\gamma)\!\downarrow_\Rules =
  (\bterm\gamma)\!\downarrow_\Rules$.
\item (\textsc{Postulate}, \textsc{Generalization})
  A reduction with these rules does not have the form
  required by the lemma (as the \complflag\ flag is removed).
\qed
\end{itemize}
\end{proof}

This leads to the second part of
\thmshort\ref{thm:ri-correctness}, which largely corresponds to
\lemshort\ref{lem:ri-correctness:complete}:

\begin{lemma}\label{lem:ri-correctness:part2}
If $\Rules$ is confluent
and $(\E,\emptyset,\complflag) \vdashri
\cdots \vdashri \bot$, then there is some equation in $\E$ which is
not an inductive theorem of $\Rules$.
\end{lemma}

\begin{proof}
If $(\E,\emptyset,\complflag) = (\E_1,\H_1,\flag_1) \vdashri \cdots
\vdashri (\E_n,\H_n,\flag_n) \vdashri \bot$, then we easily see that
$\flag_i = \complflag$ for all $i$.
By Lemma~\ref{lem:completeness:final}, $\E_n$ contains an equation
$\aterm \approx \bterm\ \constraint{\varphi}$ which is not an
inductive theorem.  Then $\leftrightarrow_{\E_n}\ \not\subseteq\ 
\leftrightarrow_\Rules^*$ on ground terms.  By
\lemshort\ref{lem:completeness:induction},
Lemma~\ref{lem:properties_of_ri-step}, and induction on $n-i$, this
means that $\leftrightarrow_\E \cup \leftrightarrow_\emptyset\ 
\not\subseteq\ \leftrightarrow_\Rules^*$ on ground terms, so by
Lemma~\ref{lem:theorem:alternative}, not all $e \in \E$ are
inductive theorems.
\qed
\end{proof}

\medskip\noindent
\emph{Proof of \thmshort\ref{thm:ri-correctness}}.
Immediately by Lemmas~\ref{lem:ri-correctness:part1}
and~\ref{lem:ri-correctness:part2}.
\qed

\section{Simple \texttt{sum}}
\label{sec:examples}

To demonstrate the difference in power between our technique and
earlier work, even when not considering advanced data structures which
were not supported in~\cite{nak:nis:kus:sak:sak:10}
or~\cite{fal:kap:12}, we have included an example that can be handled
with the technique in this paper (and is automatically proved by
\ctrl), but not with~\cite{nak:nis:kus:sak:sak:10}
or~\cite{fal:kap:12} (the latter of which is not surprising, as it
does not use any lemma generation at all).

\begin{example}
In the programming course in Nagoya, students in the first
lecture were asked to implement a function $\texttt{sum}$
which
computes the summation from $0$ to a given non-negative integer $x$.
The teacher's reference implementation was:
\smallskip

\begin{verbatim}
int sum(int x) {
    int z = 0;
    for (int i = 1; i <= x; i++) {
        z += i;
    }
    return z;
}
\end{verbatim}
\smallskip

\noindent
Some of the students solved (or tried to solve) this in the clever
way instead:
\medskip

\noindent
\begin{minipage}[t]{0.5\textwidth}
\begin{verbatim}
int sum1(int x) {
    return x * (x + 1) / 2;
}
\end{verbatim}
\end{minipage}
\begin{minipage}[t]{0.5\textwidth}
\begin{verbatim}
int sum2(int x) {
    return x * (x - 1) / 2;
}
\end{verbatim}
\end{minipage}

\bigskip\noindent
To stay close to the transformation from~\cite{nak:nis:kus:sak:sak:10}
(which does not use the $\symb{return}$ and $\symb{error}$ symbols),
we consider the following translation:
\[
\begin{array}{rcl}
\symb{sum}(x) & \arrz & \symb{u}(x,\one,\nul) \\
\symb{u}(x,i,z) & \arrz & \symb{u}(x,i+\one,z+i) ~~\hfill \constraint{i \leq x} \\
\symb{u}(x,i,z) & \arrz & z \hfill \constraint{i > x} \\
\symb{sum1}(x) & \arrz & x * (x + \one)\ \symb{div}\ \two \\
\symb{sum2}(x) & \arrz & x * (x - \one)\ \symb{div}\ \two \\
\end{array}
\]
Our implementation succeeds in proving that
$\symb{sum}(n) \approx \symb{sum1}(n)\ \constraint{n \geq \nul}$ is an
inductive theorem and that
$\symb{sum}(n) \approx \symb{sum2}(n)\ \constraint{n \geq \nul}$ is not.
We also succeed on the translation using the methods in the current
paper.
On the other hand, the method in~\cite{nak:nis:kus:sak:sak:10} fails to
prove or disprove these claims. 
\end{example}

\section{Some examples we cannot handle.}\label{app:difficult}

To demonstrate the kind of problems \ctrl\ cannot yet handle, we
compare a recursive definition $\texttt{sum}$ of the function $n
\mapsto \sum_{i=1}^n i$ with three iterative implementations.

\vspace{5pt}
\noindent
\begin{tabular}{c|c}
\begin{minipage}[t]{0.3\textwidth}
\begin{verbatim}
int sum(n) {
  if (n < 0) return 0;
  return n + sum(n-1);
}
\end{verbatim}
\end{minipage}
&
\begin{minipage}[t]{0.6\textwidth}
\begin{verbatim}
int sum1(n) {
  int i = 0, j = 0, sum = 0;
  for (; i <= n; i++,j++) sum += j;
  return len;
}
\end{verbatim}
\end{minipage} \\
\\
\hline
\\
\begin{minipage}[t]{0.3\textwidth}
\begin{verbatim}
int sum2(int n){
  int i,sum=0;
  for (i=n;i>=0;i--)
    sum=sum+i;
  return sum;
}

\end{verbatim}
\end{minipage}
&
\begin{minipage}[t]{0.6\textwidth}
\begin{verbatim}
int sum3(n) {
  int ret = 0;
  for (int i = 0; i <= n; i++)
    for (int j = 0; j < i; j++) ret++;
  return ret;
}

\end{verbatim}
\end{minipage}
\end{tabular}

Equivalence between $\texttt{sum}$ and each of $\texttt{sum1}$,
$\texttt{sum2}$ and $\texttt{sum3}$ fails for the three main reasons
discussed in \secshort\ref{subsec:experiments}.
For $\texttt{sum1}$, generalizing the initialization variables loses the
information that always $i = j$.  For $\texttt{sum2}$, our main
generalization method (\secshort\ref{subsec:generalise})
does not apply because we do not recognize $i=n$
as an initialization.  For $\texttt{sum3}$, our strategy fails because
the two loop counters are generalized together.

\bibliographystyle{acmtrans}
\bibliography{references}

\end{document}